\documentclass[11pt]{article}
\usepackage[T1]{fontenc}
\usepackage[latin9]{inputenc}
\usepackage{geometry}
\usepackage{amsmath}
\usepackage{amsfonts}
\usepackage{amssymb}
\usepackage{amsthm}
\usepackage{comment}
\usepackage{algorithm}
\usepackage{algorithmicx}
\usepackage{algpseudocode}
\usepackage{enumerate}
\usepackage{mathtools}
\usepackage{thm-restate}
\usepackage{mathrsfs}
\usepackage{csquotes}
\usepackage{apptools}
\usepackage{graphicx}
\usepackage{verbatim}
\usepackage{float}
\usepackage{chngcntr}
\usepackage{xr}
\usepackage{bm}
\usepackage{todonotes}
\usepackage[shortlabels]{enumitem}
\usepackage{tikz}
\usepackage{pdfsync}
\usepackage{nicefrac}
\usepackage[numbers]{natbib}
\usepackage{crossreftools}
\usepackage[hypertexnames=false,backref=page]{hyperref}
\usepackage[capitalize, nameinlink]{cleveref}

\makeatletter

\numberwithin{equation}{section}
\numberwithin{figure}{section}

\theoremstyle{plain}
\newtheorem{theorem}{Theorem}

\newtheorem{corollary}[theorem]{Corollary}
\newtheorem{lemma}[theorem]{Lemma}

\newtheorem{observation}[theorem]{Observation}

\newtheorem{claim}[theorem]{Claim}

\theoremstyle{remark}
\newtheorem{remark}[theorem]{Remark}

\theoremstyle{definition}
\newtheorem{definition}[theorem]{Definition}

\makeatother
 \newcommand{\Sally}[1]{\textbf{\color{magenta}[Sally: #1]}}

\algnewcommand{\LeftComment}[1]{\(\triangleright\) #1}

\global\long\def\defeq{\stackrel{\mathrm{{\scriptscriptstyle def}}}{=}}%
\global\long\def\norm#1{\left\Vert #1\right\Vert }%

\newcommand{\energy}[2]{\mathcal{E}_{#1}(#2)}

\def\eps{\varepsilon}
\global\long\def\R{\mathbb{R}}%
\global\long\def\diag{\mathrm{diag}}%
\global\long\def\mm{\mathbf{M}}%
\global\long\def\ml{\mathbf{L}}%
\global\long\def\mi{\mathbf{I}}%
\global\long\def\mj{\mathbf{J}}%
\global\long\def\mzero{\mathbf{0}}%
\global\long\def\ma{\mathbf{A}}%
\global\long\def\md{\mathbf{D}}%
\global\long\def\mx{\mathbf{X}}%
\global\long\def\mproj{\mathbf{P}}%

\global\long\def\omm{\overline{\mm}}%
\global\long\def\mt{\mathbf{T}}%
\global\long\def\sc{\mathbf{Sc}}%
\global\long\def\tsc{\widetilde{\mathbf{Sc}}}%
\global\long\def\mb{\mathbf{B}}%
\global\long\def\mw{\mathbf{W}}%
\global\long\def\mmu{\mathbf{U}}%
\global\long\def\mpi{\mathbf{\Pi}}%
\global\long\def\mga{\mathbf{\Gamma}}%
\global\long\def\mphi{\mathbf{\Phi}}

\global\long\def\epssc{{\epsilon_{\mproj}}}
\global\long\def\epslevel{{\delta}}

\global\long\def\ct{\mathcal{T}}%
\global\long\def\cf{\mathcal{F}}

\global\long\def\sketchlen{w} 

\newcommand{\zprev}{{\vz^{(\mathrm{step})}}}
\newcommand{\tz}{\vu}

\newcommand{\zsum}{{\vz^{(\mathrm{sum})}}}
\newcommand{\vnew}{\vv^{(\mathrm{new})}}

\global\long\def\ot{\overline{t}}%

\global\long\def\vx{\bm{x}}%
\global\long\def\vz{\bm{z}}%
\global\long\def\vd{\bm{d}}%
\global\long\def\vf{\bm{f}}%
\global\long\def\vq{\bm{q}}%
\global\long\def\vb{\bm{b}}%
\global\long\def\vc{\bm{c}}%
\global\long\def\vf{\bm{f}}%
\global\long\def\vs{\bm{s}}%
\global\long\def\vv{\bm{v}}%
\global\long\def\vw{\bm{w}}%
\global\long\def\vv{\bm{v}}%
\global\long\def\vy{\bm{y}}%
\global\long\def\vl{\bm{l}}%
\global\long\def\vu{\bm{u}}%
\global\long\def\vx{\bm{x}}%

\global\long\def\ox{\overline{\vx}}%
\global\long\def\of{\overline{\vf}}%
\global\long\def\os{\overline{\vs}}%
\global\long\def\tf{\tilde{\bm{f}}}%
\global\long\def\pf{\bm{f}^{\perp}}%
\global\long\def\ts{\tilde{\bm{s}}}%
\global\long\def\uf{\hat{\bm{f}}}%

\global\long\def\otilde{\widetilde{O}}%

\global\long\def\O{\widetilde{O}}%

\global\long\def\new{{(\mathrm{new})}}
\global\long\def\old{{(\mathrm{old})}}

\global\long\def\init{{(\mathrm{init})}}

\global\long\def\flow{{(\mathrm{flow})}}
\global\long\def\slack{{(\mathrm{slack})}}

\global\long\def\vbeta{\bm{\beta}}%
\global\long\def\vzero{\bm{0}}%
\global\long\def\vone{\bm{1}}%

\global\long\def\region{H}

\newcommand{\sepConst}{\alpha}
\newcommand{\bdry}[1]{\partial #1}

\newcommand{\pathT}[1]{\mathcal{P}_{\mathcal{T}}(#1)}

\global\long\def\collN{\mathcal{H}}
\newcommand{\child}[2]{\mathrm{c}_{#2}(#1)}
\newcommand{\sep}[1]{S(#1)}
\newcommand{\elim}[1]{{F_{#1}}}

\newcommand{\graphclass}{\mathcal{C}}

\newcommand{\nfrac}[2]{\nicefrac{1}{2}}
\newcommand{\flowSketch}{\texttt{bar\_f}}
\newcommand{\slackSketch}{\texttt{bar\_s}}

\newcommand{\maintainRep}{\texttt{maintainRep}}

\global\long\def\range#1{\mathrm{Range}(#1)}%
\global\long\def\poly{\mathrm{poly}}%
\global\long\def\ker#1{\mathrm{Kernel}(#1)}%
\global\long\def\septime{K^{1-\sepConst} m^\sepConst}
\let\ref\cref

\author{
	Sally Dong \\ University of Washington \\ sallyqd@uw.edu \and 
	Yu Gao \\ Georgia Tech \\  ygao380@gatech.edu 
	\and 
	Gramoz Goranci
	\footnote{Part of this work was done while the author was at University of Toronto.}
\\ University of Glasgow \\ gramoz.goranci@glasgow.ac.uk 
	\and
	Yin Tat Lee \thanks{Supported by NSF awards CCF-1749609, DMS-1839116, DMS-2023166, CCF-2105772, a Microsoft Research Faculty Fellowship, a Sloan Research Fellowship, and a Packard Fellowship.} \\ University of Washington \\ yintat@uw.edu
	\and
	Richard Peng \thanks{Supported by NSF award CCF-1846218 and CCF-2106444. Part of this work was done while the author was at Georgia Tech.}
	\\ University of Waterloo \\ y5peng@uwaterloo.ca
	\and
	Sushant Sachdeva \thanks{Supported by a Discovery grant
          awarded by NSERC, and an Ontario Early Researcher Award.} \\ University of Toronto \\ sachdeva@cs.toronto.edu 
	\and 
	Guanghao Ye\footnote{Supported by an MIT Presidential Fellowship. Part of this work was done while the author was a student at the University of Washington.}
	\\ Massachusetts Institute of Technology \\ ghye@mit.edu
}

\begin{document}
\title{Nested Dissection Meets IPMs: \\
  Planar Min-Cost Flow in Nearly-Linear Time\footnote{A preliminary
    version of this work was published at SODA 2022.}}
\date{}
\maketitle
\begin{abstract}
	We present a nearly-linear time algorithm for finding a minimum-cost
	flow in planar graphs with polynomially bounded integer costs and
	capacities. The previous fastest algorithm for this problem is
	based on interior point methods (IPMs) and works for general sparse
	graphs in $O(n^{1.5}\text{poly}(\log n))$ time [Daitch-Spielman,
	STOC'08].

	Intuitively, $\Omega(n^{1.5})$ is a natural runtime barrier for IPM-based
	methods, since they require $\sqrt{n}$ iterations, each
	routing a possibly-dense electrical flow.
	To break this barrier, we develop a new implicit representation for
	flows based on generalized nested-dissection [Lipton-Rose-Tarjan,
	JSTOR'79] and approximate Schur complements [Kyng-Sachdeva,
	FOCS'16]. This implicit representation permits us to design a data
	structure to route an electrical flow with sparse demands in roughly
	$\sqrt{n}$ update time, resulting in a total running time of
	$O(n\cdot\text{poly}(\log n))$.

	Our results immediately extend to all families of separable graphs.
\end{abstract}
\newpage

\tableofcontents
\newpage

\section{Introduction}

The minimum cost flow problem on planar graphs is a foundational
problem in combinatorial optimization studied since the 1950's. It has
diverse applications including network design, VLSI layout, and
computer vision. The seminal paper of Ford and Fulkerson in the 1950's
\cite{ford1956maximal} presented an $O(n^{2})$ time algorithm for the
special case of max-flow on $s,t$-planar graphs, i.e., planar graphs
with both the source and sink lying on the same face. 
%
Over the decades since, a number of nearly-linear time max-flow
algorithms have been developed for special graph classes, including
undirected planar graphs by Reif, and Hassin-Johnson
\cite{reif1983minimum, hassin1985n}, planar graphs by
Borradaile-Klein \cite{BK09:journal}, 
and finally bounded genus graphs by Chambers-Erickson-Nayyeri
\cite{CEN12:journal}.  However, for the more general min-cost flow
problem, there is no known result specializing on planar graphs with
better guarantees than on general graphs.
In this paper, we present the first nearly-linear time algorithm for min-cost flow on planar graphs:

\begin{restatable}[Main result]{theorem}{mainthm}\label{thm:mincostflow}
	Let $G=(V,E)$ be a directed planar graph with $n$ vertices and $m$ edges.
	Assume that the demands $\vd$, edge capacities $\vu$ and costs $\vc$ are all integers and bounded by $M$ in absolute value. 
	Then there is an algorithm that computes a minimum cost flow satisfying demand $\vd$ in $\O(n\log M)$ \footnote{Throughout the paper, we use $\otilde(f(n))$ to denote $O(f(n) \log^{O(1)}f(n))$.} expected time.
\end{restatable}

Our algorithm is fairly general and uses the planarity assumption
minimally. It builds on a combination of interior point methods (IPMs),
approximate Schur complements, and nested-dissection, 
with the latter being the only component that exploits planarity.
Specifically, we require that for any subgraph of the input graph with $k$
vertices, we can find an $O(\sqrt{k})$-sized balanced vertex separator
in nearly-linear time.
As a result, the algorithm naturally generalizes to all graphs
with small separators:
Given a class $\graphclass$ of graphs closed under taking subgraphs, we say it is \emph{$\sepConst$-separable} if there are constants $0< b < 1$ and $c > 0$ such that every graph in $\graphclass$ with $n$ vertices and $m$ edges has a balanced vertex separator with at most $c m^\sepConst$ vertices, and both components obtained after removing the separator have at most $bm$ edges. 
Then, our algorithm generalizes as follows:

\begin{restatable}[Separable min-cost flow]{corollary}{separableMincostflow}\label{thm:separable_mincostflow}
  	Let $\mathcal{C}$ be an $\sepConst$-separable graph class such that we can compute a balanced separator for any graph in $\mathcal{C}$ with $m$ edges in $s(m)$ time for some convex
  	 function $s$.  
  	Given a graph $G\in \mathcal{C}$ with $n$ vertices and $m$ edges, integer demands $\vd$, edge capacities $\vu$ and costs $\vc$, all bounded by $M$ in absolute value,
  	there is an algorithm that computes a minimum cost flow on $G$ satisfying demand $\vd$ in $\O((m+m^{1/2+\sepConst})\log M+s(m))$ expected time.
  \end{restatable}

Beyond the study of structured graphs, we believe our paper is of
broader interest. The study of efficient optimization algorithms on
geometrically structured graphs is a topic at the intersection of
computational geometry, graph theory, combinatorial optimization, and
scientific computing, that has had a profound impact on each of these
areas. Connections between planarity testing and $3$-vertex
connectivity motivated the study of depth-first search algorithms
\cite{tarjan1971efficient}, and using geometric structures to find
faster solvers for structured linear systems provided foundations of
Laplacian algorithms as well as combinatorial scientific computing
\cite{lipton1979generalized, gremban1996combinatorial}.
Several surprising insights from our nearly-linear time algorithm are:
\begin{enumerate}
\item We are able to design a data structure for maintaining a
  feasible primal-dual (flow/slack) solution that allows sublinear time updates --
  requiring $\otilde(\sqrt{nK})$
  time for a batch update consisting of updating the flow value of $K$ edges.
  This ends up not being a bottleneck for the overall performance
  because the interior point method only takes roughly $\sqrt{n}$
  iterations and makes $K$-sparse updates roughly $\sqrt{n / K}$
  times, resulting in a total running time of $\otilde(n)$.
\item We show that the subspace constraints on the feasible primal-dual solutions 
can be maintained implicitly under dynamic updates to the solutions.
%
%
  This circumvents the need to track the infeasibility of primal
  solutions (flows), which was required in previous works.
\end{enumerate}

We hope our result provides both a host of new tools for devising
algorithms for separable graphs, as well as insights on how to further
improve such algorithms for general graphs.

\subsection{Previous work}

The min-cost flow problem is well studied in both structured graphs
and general graphs. 
\cref{tab:runtime} summarizes the best algorithms for different settings prior to this work.

\begin{table}[H]
\begin{centering}
\begin{tabular}{|c|c|c|}
\hline 
Min-cost flow & Time bound & Reference\tabularnewline
\hline 
\hline 
Strongly polytime & $O(m^{2}\log n+mn\log^2 n)$ & \cite{orlin1988faster}\tabularnewline
\hline 
Weakly polytime & $\O((m + n^{3/2}) \log^2 M)$ & \cite{BrandLLSSSW21} \tabularnewline
\hline 
Unit-capacity & $m^{\frac{4}{3}+o(1)}\log M$ & \cite{axiotis2020circulation}\tabularnewline
\hline 
\textbf{Planar graph} & $\O(n\log M)$ & \textbf{this paper}\tabularnewline
\hline 
Unit-capacity planar graph & $O(n^{4/3}\log M)$ & \cite{KS19}\tabularnewline
\hline 
Graph with treewidth $\tau$ & $\O(n\tau^{2}\log M)$ & \cite{treeLP}\tabularnewline
\hline 
Outerplanar graph & $O(n\log^{2}n)$ & \cite{kaplan2013min}\tabularnewline
\hline 
Unidirectional, bidirectional cycle & $O(n)$, $O(n\log n)$ & \cite{vaidyanathan2010fast}\tabularnewline
\hline 
\end{tabular}
\par\end{centering}
\caption{Fastest known exact algorithms for the min-cost flow problem, ordered by the generality of the result. Here, $n$ is the number of vertices,
$m$ is the number of edges, and $M$ is the maximum of edge capacity and cost value. 
After the preliminary version of this work was published at SODA 2022, the best weakly polytime algorithm was improved to $\O(m^{1 + o(1)} \log^2 M)$ by \cite{almostlinearmincostflow}.
\label{tab:runtime}}
%

\end{table}

\paragraph{Min-cost flow / max-flow on general graphs.}

Here, we focus on recent exact max-flow and min-cost flow algorithms.
For an earlier history, we refer the reader to the monographs \cite{king1994faster,ahuja1988network}.
For the approximate max-flow problem, we refer the reader to the recent
papers \cite{christiano2011electrical,sherman2013nearly,kelner2014almost,sherman2017area,sidford2018coordinate,bernstein2021deterministic}. 

To understand the recent progress, we view the max-flow problem as
finding a unit $s,t$-flow with minimum $\ell_{\infty}$-norm, 
and the shortest path problem as finding a unit $s,t$-flow with minimum $\ell_{1}$-norm.
Prior to 2008, almost all max-flow algorithms reduced
this $\ell_{\infty}$ problem to a sequence of $\ell_{1}$ problems,
(shortest path) since the latter can be solved efficiently.
This changed with the celebrated work of Spielman and Teng, which
showed how to find electrical flows ($\ell_{2}$-minimizing unit
$s,t$-flow) in nearly-linear time \cite{spielman2004nearly}. 
Since the $\ell_{2}$-norm is closer to
$\ell_{\infty}$ than $\ell_{1}$, this gives a more powerful primitive
for the max-flow problem.
In 2008, Daitch and Spielman demonstrated that one could apply
interior point methods (IPMs) to reduce min-cost flow to
roughly $\sqrt{m}$ electrical flow computations.
This follows from the fact that IPMs take $\otilde(\sqrt{m})$
iterations and each iteration requires solving an electrical flow
problem, which can now be solved in $\otilde(m)$ time due to the
work of Spielman and Teng.
Consequently, they obtained an algorithm with a $\otilde(m^{3/2}\log M)$ runtime \cite{daitch2008faster}.
Since then, several algorithms have utilized electrical flows and
other stronger primitives for solving max-flow and min-cost flow
problems.

For graphs with unit capacities, M\k{a}dry gave a
$\otilde(m^{10/7})$-time max-flow algorithm, the first that broke the
$3/2$-exponent barrier \cite{madry2013navigating}. It was later
improved and generalized to $O(m^{4/3+o(1)}\log M)$ \cite{axiotis2020circulation} for the min-cost flow problem. 
Kathuria et al. \cite{KathuriaLS20} gave a similar runtime of $O(m^{4/3 + o(1)} U^{1/3})$ where $U$ is the max capacity.
The runtime improvement comes from decreasing the number of iterations of IPM to
$\otilde(m^{1/3})$ via a more powerful primitive of
$\ell_{2}+\ell_{p}$ minimizing flows \cite{kyng2019flows}.

For general capacities, the runtime has recently been improved to
$\otilde((m+n^{3/2})\log^2 M)$ for min-cost flow on dense graphs~\cite{BrandLLSSSW21},
and $\otilde(m^{\frac{3}{2}-\frac{1}{328}}\log M)$ for max-flow on sparse graphs~\cite{GaoLP21:arxiv}.
These algorithms focus on decreasing
the per-iteration cost of IPMs by dynamically maintaining electrical
flows. After the preliminary version of this work was accepted to SODA
2022, \cite{BGJLLPS21} gave a runtime of $\otilde(m^{\frac{3}{2}-\frac{1}{58}}\log^2 M)$ for general min-cost flow following the dynamic electrical flow framework. 
Most recently, \cite{almostlinearmincostflow} improved the runtime for
general min-cost flow to $\O(m^{1+o(1)}\log^2 M)$ by solving a sequence of approximate undirected minimum-ratio cycles.

\paragraph{Max-flow on planar graphs.}

The planar max-flow problem has an equally long history. We refer the
reader to the thesis \cite{borradaile2008exploiting} for a detailed
exposition.  In the seminal work of Ford and Fulkerson that introduced the max-flow min-cut theorem,
they also gave a max-flow algorithm for $s,t$-planar graphs (planar
graphs where the source and sink lie on the same face)\cite{ford1956maximal}.  
This algorithm iteratively sends flow along the top-most augmenting path.
Itai and Shiloach showed how to implement each step in
$O(\log n)$ time, thus giving an $O(n\log n)$ time algorithm for
$s,t$-planar graphs~\cite{itai1979maximum}. 
In this setting, Hassin also showed that the max-flow can be computed
using shortest-path distances in the planar dual in $O(n \log n)$ time~\cite{Hassin}.
Building on Hassin's work, the current best runtime is $O(n)$ 
by Henzinger, Klein, Rao, and Subramanian \cite{henzinger1997faster}.

For undirected planar graphs, Reif first gave an $O(n\log^{2}n)$ time
algorithm for finding the max-flow value \cite{reif1983minimum}.
Hassin and Johnson then showed how to compute the flow in the same
runtime \cite{hassin1985n}. The current best runtime is
$O(n\log\log n)$ by Italiano, Nussbaum, Sankowski, and
Wulff-Nilsen~\cite{INSW11}.

For general planar graphs, Weihe gave the first $O(n\log n)$ time
algorithm, assuming the graph satisfies certain
connectivity conditions~\cite{weihe1997maximum}. 
Later, Borradaile and Klein gave an $O(n\log n)$
time algorithm for any planar graph \cite{BK09:journal}.

The multiple-source multiple-sink version of max-flow is considered
much harder on planar graphs. The first result of $O(n^{1.5})$
time was by Miller and Naor when sources and sinks are all on same face
\cite{MN95:journal}. This was then improved to $O(n\log^{3}n)$
in~\cite{BKMNW17:journal}.

For generalizations of planar graphs, Chambers, Ericskon and Nayyeri gave the first nearly-linear time
algorithm for max-flow on graphs embedded on bounded-genus surfaces~\cite{CEN12:journal}.  
Miller and Peng gave an $\otilde(n^{6/5})$-time algorithm for approximating undirected max-flow for the class of $O(\sqrt{n})$-separable graphs~\cite{MP13}, 
although this is superseded by the previously mentioned works for general graphs~\cite{sherman2013nearly,kelner2014almost}.

\paragraph{Min-cost flow on planar graphs.}

Imai and Iwano gave a $O(n^{1.594}\log M)$ time algorithm for min-cost
flow for the more general class of $O(\sqrt{n})$-separable
graphs~\cite{II90}.
To the best of our knowledge, there is little else known
about min-cost flow on general planar graphs.  In the special case of
unit capacities, \cite{AKLR18,LR19} gives an $O(n^{6/5}\log M)$ time
algorithm for min-cost perfect matching in bipartite planar graphs,
and Karczmarz and Sankowski gives a $O(n^{4/3}\log M)$ time algorithm
for min-cost flow~\cite{KS19}.  
Currently, bounded treewidth graphs is the only graph family we know that admits min-cost flow algorithms that run in nearly-linear time \cite{treeLP}.

\subsection{Challenges}

Here, we discuss some of the challenges in developing faster
algorithms for the planar min-cost flow problem from a convex
optimization perspective.  For a discussion on challenges in designing
combinatorial algorithms, we refer the reader to
\cite{khuller1993lattice}. Prior to our result, the fastest min-cost
flow algorithm for planar graphs is based on interior point methods (IPMs) and takes $\otilde(n^{3/2}\log M)$ time~\cite{daitch2008faster}.
Intuitively, $\Omega(n^{3/2})$ is a natural runtime barrier for IPM-based methods, 
since they require $\Omega(\sqrt{n})$ iterations, each computing a possibly-dense electrical flow.

\paragraph*{Challenges in improving the number of iterations.}

The $\Omega(\sqrt{n})$ term comes from the fact that IPM uses the
electrical flow problem ($\ell_{2}$-type problem) to approximate the
shortest path problem ($\ell_{1}$-type problem). This 
$\Omega(\sqrt{n})$ term is analogous to the flow decomposition
barrier: in the worst case, we need $\Omega(n)$ shortest paths
($\ell_{1}$-type problem) to solve the max-flow problem
($\ell_{\infty}$-type problem). Since $\ell_{2}$ and $\ell_{\infty}$
problems differ a lot when there are $s-t$ paths with drastically
different lengths, difficult instances for electrical flow-based
max-flow methods are often serial-parallel (see Figure 3 in
\cite{christiano2011electrical} for an example).  Therefore, planarity
does not help to improve the $\sqrt{n}$ term.
Although more general $\ell_{2}+\ell_{p}$ primitives have been
developed \cite{AdilKPS19, kyng2019flows, AdilS20, AdilBKS21}, exploiting their
power in designing current algorithms for exact max-flow problem has
been limited to perturbing the IPM trajectory, and such a perturbation
only works when the residual flow value is large.  In all previous
works tweaking IPMs for breaking the 3/2-exponent barrier
\cite{madry2013navigating,madry2016computing,cohen2017negative,KathuriaLS20,axiotis2020circulation},
an augmenting path algorithm is used to send the remaining flow at the
end. Due to the residual flow restriction, all these results assume
unit-capacities on edges, and it seems unlikely that planarity can be
utilized to design an algorithm for polynomially-large capacities with
fewer than $\sqrt{n}$ IPM iterations.

\paragraph*{Challenges in improving the cost per iteration.}

Recently, there has been much progress on utilizing data structures
for designing faster IPM algorithms for general linear programs and
flow problems on general graphs. For general linear programs, robust
interior point methods have been developed recently with running times
that essentially match the matrix multiplication cost
\cite{CohenLS21, van2020deterministic, BrandLSS20, huang2021solving,
  van2021unifying}.  This version of IPM ensures that the $\ell_{2}$
problem solved changes in a sparse manner from iteration to iteration.
When used to design graph algorithms, the $i$-th iteration of
a robust IPM involves computing an electrical flow on some graph
$G_{i}$. The edge support remains unchanged between iterations, though
the edge weights change. Further, if $K_{i}$ is the number of edges
with weight changes between $G_{i}$ and $G_{i+1}$, then robust IPMs
guarantee that
\[
\sum_{i}\sqrt{K_{i}}=\otilde(\sqrt{m}\log M).
\]
Roughly, this says that, on average, each edge weight changes only
poly-log many times throughout the algorithm. Unfortunately, any
sparsity bound is not enough to achieve nearly-linear time. Unlike the
shortest path problem, changing \emph{any} edge in a connected graph
will result in the electrical flow changing on essentially
\emph{every} edge. Therefore, it is very difficult to implement
(robust) IPMs in sublinear time per iteration, even if the subproblem
barely changes every iteration. On moderately dense graphs with
$m=\Omega(n^{1.5})$, this issue can be avoided by first approximating
the graph by sparse graphs and solving the electrical flow on the
sparse graphs. This leads to $\otilde(n) \ll \otilde(m)$ time cost per
step~\cite{BrandLSS20}. However, on sparse graphs, significant obstacles remain.
Recently, there has been a major breakthrough in this direction by
using random walks to approximate the electrical flow
\cite{GaoLP21:arxiv, BGJLLPS21}.  Unfortunately, this still requires
$m^{1-\frac{1}{58}}$ time per iteration.

Finally, we note that \cite{treeLP} gives an
$\otilde(n \tau^2 \log M)$-time algorithm for linear programs with
$\tau$ treewidth. Their algorithm maintains the solution using an
implicit representation.  This implicit representation involves a
$\tau \times \tau$ matrix that records the interaction between every
variable within the vertex separator set.  Each step of the algorithm updates
this matrix once and it is not the bottleneck for the
$\otilde(n \tau^2 \log M)$-time budget.  However, for planar graphs,
this $\tau \times \tau$ matrix is a dense graph on $\sqrt{n}$ vertices
given by the Schur complement on the separator.  Hence, updating this
using their method requires $\Omega(n)$ time per step.

Our paper follows the approach in \cite{treeLP} and shows that this
dense graph can be sparsified.  This is however subtle. Each step of
the IPM makes a global update via the implicit representation, hence
checking whether the flow is feasible takes at least linear time.
Therefore, we need to ensure each step is exactly feasible despite the
approximation.
If we are unable to do that, the algorithm will need to fix the flow
by augmenting paths at the end like \cite{KathuriaLS20,
  axiotis2020circulation}, resulting in super-linear time and
polynomial dependence on capacities, rather than logarithmic.


\subsection{Our approaches}

In this section, we introduce our approach and explain how we overcome the difficulties we mentioned. 
The min-cost flow problem can be reformulated into a linear program
in the following primal-dual form:
\[
\mathrm{(Primal)}=\min_{\mb^{\top}\vf=\vzero,\;\vl\leq\vf\leq\vu}\vc^{\top}\vf\quad\text{and}\quad\mathrm{(Dual)}=\min_{\mb\vy+\vs=\vc}\sum_{i}\min(\vl_{i}\vs_{i},\vu_{i}\vs_{i}),
\]
where $\mb\in\R^{m\times n}$ is an edge-vertex incidence matrix of the
graph, $\vf$ is the flow and $\vs$ is the slack (or adjusted cost
vector). The primal is the min-cost circulation problem and the
dual is a variant of the min-cut problem. Our algorithm for
min-cost flow is composed of a novel application of IPM (\cref{subsec:IPM}) and new data structures
(\cref{subsec:overview_representation}). The IPM method reduces
solving a linear program to applying a sequence of
$\otilde(\sqrt{m}\log M)$ projections and the data structures
implement the primal and dual projection steps roughly in
$\otilde(\sqrt{m})$ amortized time.

\paragraph{Robust IPM.}

We first explain the IPM briefly. To minimize $\vc^\top \vf$, each step of the IPM method moves the flow vector $\vf$ to the direction of $-\vc$. However, such $\vf$ may exceed the maximum or minimum capacities. IPM incorporates these capacity constraints by routing flows slower when they are approaching their capacity bounds. This is achieved by controlling the edge weights $\mw$ and direction $\vv$ in each projection step. 
Both $\mw$ and $\vv$ are roughly chosen from some explicit entry-wise formula
of $\vf$ and $\vs$, namely, $\mw_{ii}=\psi_{1}(\vf_{i},\vs_{i})$
and $\vv_{i}=\psi_{2}(\vf_{i},\vs_{i})$. Hence, the main bottleneck
is to implement the projection step (computing $\mproj_{\vw}\vv$). 
For the min-cost flow problem, this projection step corresponds to an electrical flow computation.

Recently, it has been observed that there is a lot of freedom in
choosing the weight $\mw$ and the direction $\vv$ (see for example
\cite{CohenLS21}). Instead of computing them exactly, we maintain some
entry-wise approximation $\of,\os$ of $\vf,\vs$ and use them to
compute $\mw$ and $\vv$. By updating $\of_{i},\os_{i}$ only when
$\vf_{i},\vs_{i}$ changed significantly, we can ensure $\of,\os$ has
mostly sparse updates. Since $\mw$ and $\vv$ are given by some
entry-wise formula of $\of$ and $\os$, this ensures that $\mw,\vv$ change
sparsely and in turn allows us to maintain the corresponding projection
$\mproj_{\vw}$ via low-rank updates.

We refer to IPMs that use approximate $\of$ and $\os$ as robust IPMs.
In this paper, we apply the version given in \cite{treeLP} in a
black-box manner.  In \cref{subsec:IPM}, we state the IPM we use. The
key challenge is implementing each step in roughly $\otilde (\sqrt{m})$
time.

\paragraph{Separators and Nested Dissection.}

Our data structures rely on the separability property of the input graph, 
which dates back to the nested dissection algorithms for
solving planar linear systems~\cite{lipton1979generalized,
  gilbert1987}.  By recursively partitioning the graph into
edge-disjoint subgraphs (i.e. regions) using balanced vertex separators, we
can construct a hierarchical decomposition of a planar graph $G$ which
is called a \textit{separator tree} \cite{fakcharoenphol2006planar}.  This is
a binary search tree over the edges in $G$.  Each node in the
separator tree represents a region in $G$. In planar graphs, for a
region $\region$ with $|\region|$ vertices, an
$O(\sqrt{|\region|})$-sized vertex separator suffices to partition it into
two balanced sub-regions which are represented by the two children of
$\region$ in the separator tree. The two subregions partition the edges in $\region$ and
share only vertices in the separator. We call the set of vertices in a
region $\region$ that appear in the separators of its ancestors the
\textit{boundary} of $\region$. Any two regions can only share
vertices on their boundaries unless one of them is an ancestor of the
other.

Nested dissection algorithms~\cite{lipton1979generalized,gilbert1987}
essentially replace each region by a graph involving only its boundary vertices, 
in a bottom-up manner.
For planar linear systems, solving the dense $\sqrt{n} \times \sqrt{n}$ submatrix corresponding
to the top level vertex separator leads to a runtime of $n^{\omega / 2}$ where $\omega$ is the
matrix multiplication exponent.
For other problems as shortest path, this primitive involving
dense graphs can be further accelerated using additional properties
of distance matrices~\cite{fakcharoenphol2006planar}.

\paragraph{Technique 1: Approximate Nested Dissection and Lazy Propagation}

Our representation of the Laplacian inverse, and in turn the
projection matrix, hinges upon a sparsified version of the
nested dissection representation.
That is, instead of a dense inverse involving all pairs of boundary
vertices, we maintain a sparse approximation.
This sparsified nested dissection has been used in the approximate
undirected planar flow algorithm from~\cite{MP13}.
However, that work pre-dated (and in some sense motivated) subsequent
works on nearly-linear time approximations of Schur complements
on general graphs~\cite{KLPSS16, KyngS16, Kyng17}.
Re-incorporating these sparsified algorithms gives runtime dependencies
that are nearly-linear, instead of quadratic, in separator sizes,
with an overall error that is acceptable to the robust IPM framework.
%

By maintaining objects with size nearly equal to the separator size in each node of the separator tree, we can support updating an single edge or a batch of edges in the graph efficiently. Our data structures for maintaining the approximate Schur complements and the slack and flow projection matrices all utilize this idea. 
For example, to maintain the Schur complement of a region $\region$ onto
its boundary (which is required in implementating the IPM step), we
maintain (1) Schur complements of its children onto their boundaries
recursively and (2) Schur complement of the children's boundaries onto
the boundary $\region$. Thus, to update an edge, the path in the
separator tree from the leaf node containing the edge to
the root is visited. To update multiple edges in a batch, each node in
the union of the tree paths is visited. The runtime is nearly
linear in the total number of boundary vertices of all nodes (regions)
in the union. 
For $K$ edges being updated, the runtime is bounded by $\otilde(\sqrt{mK})$. 
Step $i$ of our IPM algorithm takes $\otilde(\sqrt{mK_i})$ time, where $K_{i}$
is the number of coordinates changed in $\mw$ and $\vv$ in the step.
Such a recursive approximate Schur complement structure was used
in~\cite{GHP18}, where the authors achieved a running time of $\otilde(\sqrt{m} K_i).$

\paragraph{Technique 2: Batching the changes.} It is known that over
$t$ iterations of an IPM, the number of coordinate changes (by
more than a constant factor) in $\mw$ and $\vv$ is bounded by $O(t^2)$. 
This directly gives $\sum_{i=1}^{\otilde(\sqrt{m})}K_{i}=m$ and thus a total runtime
of
$\sqrt{m}  \left( \sum_{i=1}^{ \otilde(\sqrt{m}) } \sqrt{K_i} \right)
= \otilde(m^{1.25}).$ In order to obtain a nearly-linear runtime, the
robust IPM carefully batches the updates in different steps. In the
$i$-th step, if the change in an edge variable has exceeded some fixed
threshold compared to its value in the $(i-2^l)$-th step for some
$l\le \ell_i$, we adjust its approximation. (Here, $\ell_i$ is the
number of trailing zeros in the binary representation of $i,$
i.e. $2^{\ell_i}$ is the largest power of $2$ that divides $i$.) 
This ensures that $K_i,$ the number of coordinate changes at step $i$, 
is bounded by $\otilde(2^{2\ell_i}).$ 
Since each value of $\ell_i$ arises once every $2^{\ell_i}$ steps,
we can prove that the sum of square roots of the number of changes over all steps is bounded
by $\otilde(m),$ i.e.,
$\sum_{i=1}^{\otilde(\sqrt{m})} \sqrt{K_{i}}=\otilde(\sqrt{m}).$
Combined with the runtime of the data structures, this gives an
$\otilde(m)$ overall runtime.

\paragraph{Technique 3: Maintaining feasibility via two projections.}
A major difficulty in the IPM is maintaining a flow vector $\vf$ that satisfies
the demands exactly and a slack vector $\vs$ that can be expressed as
$\vs = \vc - \mb \vy$. If we simply project $\vv$ approximately in each
step, the flow we send is not exactly a circulation. Traditionally,
this can be fixed by computing the excess demand each step and sending
flow to fix this demand. Since our edge capacities can be polynomially
large, this step can take $\Omega(m)$ time. 
To overcome this feasibility problem, we note that distinct projection operators
$\mproj_{\vw}$ can be used in IPMs for $\vf$ and $\vs$ as long as each
projection is close to the true projection and that the step satisfies
$\mb^{\top}\Delta\vf=\vzero$ and $\mb\Delta\vy+\Delta\vs=\vzero$ for
some $\Delta\vy$. 

This two-operator scheme is essential to
our improvement since one can prove that any projection that gives
feasible steps for $\vf$ and $\vs$ simultaneously must be the exact electrical
projection, which takes linear time to compute.

\section{Overview}
In this section, we give formal statements of the main theorems proved
in the paper, along with the proof for our main result.
We provide a high-level explanation of the algorithm, sometimes using a simplified setup.

The main components of this paper are:
the IPM from \cite{treeLP} (\cref{subsec:IPM});
a data structure to maintain a collection of Schur complements via nested dissection of the graph (\cref{subsec:overview_sc});
abstract data structures to maintain the solutions $\vs, \vf$ implicitly, notably using an abstract tree operator (\cref{subsec:overview_representation});
a sketching-based data structure to maintain the approximations $\overline \vs$ and $\overline\vf$ needed in the IPM (\cref{subsec:overview_sketch}); and finally, the definition of the tree operators for slack and flow corresponding to the IPM projection matrices onto their respective feasible subspaces, along with the complete IPM data structure for slack and flow (\cref{subsec:dual_overview,subsec:primal_overview}). 

We extend our result to $\alpha$-separable graphs in \cref{sec:separable}.

\subsection{Robust interior point method\label{subsec:IPM}}

In this subsection, we explain the robust interior point method
developed in \cite{treeLP}, which is a refinement of the methods in
\cite{CohenLS21, van2020deterministic}.
Although there are many other robust interior point methods, we simply
refer to this method as RIPM. Consider a linear program of the
form\footnote{Although the min-cost flow problem can be written as a
	one-sided linear program, it is more convenience for the linear
	program solver to have both sides. Everything in this section works
	for general linear programs and hence we will not use the fact
	$m = O(n)$ in this subsection.}
\begin{equation}
	\min_{\vf\in\mathcal{F}} \vc^{\top}\vf\quad\text{where}\quad\mathcal{F}=\{\mb^{\top}\vf=\vb,\; \vl\leq\vf\leq\vu\}\label{eq:LP}
\end{equation}
for some matrix $\mb\in\mathbb{R}^{m\times n}$. 
As with many other IPMs, RIPM follows the central path $\vf(t)$ from an interior point
($t\gg0$) to the optimal solution ($t=0$):
\[
\vf(t)\defeq\arg\min_{\vf\in\mathcal{F}}\vc^{\top}\vf-t\phi(\vf)\quad\text{where }\phi(\vf)\defeq-\sum_{i}\log(\vf_{i}-\vl_{i})-\sum_{i}\log(\vu_{i}-\vf_{i}),
\]
where the term $\phi$ controls how close the flow $\vf_{i}$ can be 
to the capacity constraints $\vu_{i}$ and $\vl_{i}$.
Following the central path exactly is expensive. Instead, RIPM
maintains feasible primal and dual solution
$(\vf, \vs) \in \mathcal{F} \times \mathcal{S}$, where $\mathcal{S}$
is the dual space given by
$\mathcal{S} = \{\vs:\mb\vy+\vs=\vc\text{ for some }\vy\}$, and
ensures $\vf(t)$ is an approximate minimizer.
Specifically, the optimality condition for $\vf(t)$ is given by
\begin{align}
	\mu^{t}(\vf,\vs)& \defeq\vs/t+\nabla\phi(\vf) = \vzero\label{eq:mu_t_def}\\
	(\vf,\vs) &\in \mathcal{F}\times\mathcal{S}\nonumber 
\end{align}
where $\mu^{t}(\vf,\vs)$ measures how close $\vf$ is to the minimizer $\vf(t)$.
RIPM maintains $(\vf,\vs)$ such that 
\begin{equation}
	\|\gamma^{t}(\vf,\vs)\|_{\infty}\leq\frac{1}{C\log m}\text{ where }\gamma^{t}(\vf,\vs)_{i}=\frac{\mu^{t}(\vf,\vs)_{i}}{(\nabla^{2}\phi (\vf))_{ii}^{1/2}}, \label{eq:gamma_t_def}
\end{equation}
for some universal constant $C$. The normalization term
$(\nabla^{2}\phi)_{ii}^{1/2}$ makes the centrality measure
$\|\gamma^{t}(\vf,\vs)\|_{\infty}$ scale-invariant in $\vl$ and $\vu$.

The key subroutine $\textsc{Centering}$ takes as input a point close
to the central path
$(\vf(t_{\mathrm{start}}),\vs(t_{\mathrm{start}}))$, and outputs
another point on the central path
$(\vf(t_{\mathrm{end}}),\vs(t_{\mathrm{end}}))$.  Each step of the
subroutine decreases $t$ by a multiplicative factor of
$(1-\frac{1}{\sqrt{m}\log m})$ and moves $(\vf,\vs)$ within
$\mathcal{F}\times\mathcal{S}$ such that $\vs/t+\nabla\phi(\vf)$ is
smaller for the current $t$.  \cite{treeLP} proved that even if each
step is computed approximately, \textsc{Centering} still outputs a
point close to $(\vf(t_{\mathrm{end}}),\vs(t_{\mathrm{end}}))$ using
$\widetilde{O}(\sqrt{m}\log(t_{\mathrm{end}}/t_{\mathrm{start}}))$
steps.  See \cref{alg:IPM_centering} for a simplified version.
\begin{algorithm}
	\caption{Robust Interior Point Method from \cite{treeLP}\label{alg:IPM_centering}}
	\begin{algorithmic}[1]
		
		\Procedure{$\textsc{RIPM}$}{$\mb \in \mathbb{R}^{m \times n}, \vb, \vc,\vl,\vu,\epsilon$}
		
		\State Let $L=\| \vc \|_{2}$ and $R=\|\vu-\vl\|_{2}$
		
		\State Define $\phi_{i}(x)\defeq-\log(\vu_{i}-x)-\log(x-\vl_{i})$ 
		
		\Statex
		
		\Statex $\triangleright$ Modify the linear program and obtain an
		initial $(x,s)$ for modified linear program
		
		\State Let $t=2^{21}m^{5}\cdot\frac{LR}{128}\cdot\frac{R}{r}$
		
		\State Compute $\vf_{c}=\arg\min_{\vl\leq\vf\leq\vu}\vc^{\top}\vf+t\phi(\vf)$
		and $\vf_{\circ}=\arg\min_{\mb^{\top}\vf=\vb}\|\vf-\vf_{c}\|_{2}$
		
		\State Let $\vf=(\vf_{c},3R+\vf_{\circ}-\vf_{c},3R)$ and $\vs=(-t  \nabla\phi(\vf_{c}),\frac{t}{3R+\vf_{\circ}-\vf_{c}},\frac{t}{3R})$
		
		\State Let the new matrix $\mb^{\mathrm{new}} \defeq [\mb;\mb;-\mb]$, the
		new barrier
		\[
		\phi_{i}^{\mathrm{new}}(x)
		=\begin{cases}
			\phi_{i}(x) & \text{if }i\in[m],\\
			-\log x & \text{else}.
		\end{cases}
		\]
		
		\Statex 
		
		\Statex $\triangleright$ Find an initial $(\vf,\vs)$ for the original
		linear program
		
		\State $((\vf^{(1)},\vf^{(2)},\vf^{(3)}),(\vs^{(1)},\vs^{(2)},\vs^{(3)}))\leftarrow\textsc{Centering}(\mb^{\mathrm{new}},\phi^{\mathrm{new}},\vf,\vs,t,LR)$
		
		\State $(\vf,\vs)\leftarrow(\vf^{(1)}+\vf^{(2)}-\vf^{(3)},\vs^{(1)})$
		
		\Statex 
		
		\Statex $\triangleright$ Optimize the original linear program
		
		\State $(\vf,\vs)\leftarrow\textsc{Centering}(\mb,\phi,\vf,\vs,LR,\frac{\epsilon}{4m})$
		
		\State \Return $\vf$
		
		\EndProcedure
		
		\Statex
		
		\Procedure{$\textsc{Centering}$}{$\mb,\phi,\vf,\vs,t_{\mathrm{start}},t_{\mathrm{end}}$}
		
		\State Let $\alpha=\frac{1}{2^{20}\lambda}$ and $\lambda=64\log(256m^{2})$
		where $m$ is the number of rows in $\mb$
		
		\State Let $t\leftarrow t_{\mathrm{start}}$, $\of\leftarrow\vf,\os\leftarrow\vs,\ot\leftarrow t$
		
		\While{$t\geq t_{\mathrm{end}}$}
		
		\State Set $t\leftarrow\max((1-\frac{\alpha}{\sqrt{m}})t,t_{\mathrm{end}})$
		
		\State Update $h=-\alpha/\|\cosh(\lambda\gamma^{\ot}(\of,\os))\|_{2}$
		where $\gamma$ is defined in \cref{eq:mu_t_def}\label{line:step_given_begin}
		
		\State Update the diagonal weight matrix $\mw=\nabla^{2}\phi(\of)^{-1}$\label{line:step_given_3}
		
		\State Update the direction $\vv$ where  $\vv_{i}=\sinh(\lambda\gamma^{\ot}(\of,\os)_{i})$\label{line:step_given_end}
		
		\State Pick $\vv^{\|}$ and $\vv^{\perp}$ such that $\mw^{-1/2}\vv^{\|}\in\mathrm{Range}(\mb)$,
		$\mb^{\top}\mw^{1/2}\vv^{\perp}=\vzero$ and\label{line:step_user_begin}
		\begin{align*}
			\|\vv^{\|}-\mproj_{\vw}\vv\|_{2} & \leq\alpha\|\vv\|_{2},\\
			\|\vv^{\perp}-(\mi-\mproj_{\vw})\vv\|_{2} & \leq\alpha\|\vv\|_{2}
			\tag{$\mproj_{\vw}\defeq\mw^{1/2}\mb(\mb^{\top}\mw\mb)^{-1}\mb^{\top}\mw^{1/2}$}
		\end{align*}
		
		\State Implicitly update $\vf\leftarrow\vf+h\mw^{1/2}\vv^{\perp}$,
		$\vs\leftarrow\vs+\ot h\mw^{-1/2}\vv^{\|}$ \label{line:maintain_f_s}
		
		\State Explicitly maintain $\of,\os$ such that $\|\mw^{-1/2}(\of-\vf)\|_{\infty}\leq\alpha$
		and $\|\mw^{1/2}(\os-\vs)\|_{\infty}\leq\ot\alpha$ \label{line:step_user_end}
		
		\State Update $\ot\leftarrow t$ if $|\ot-t|\geq\alpha\ot$
		
		\EndWhile
		
		\State \Return $(\vf,\vs)$\label{line:step_user_output}
		
		\EndProcedure
		
	\end{algorithmic}
	
\end{algorithm}

RIPM calls $\textsc{Centering}$ twice. The first call to
$\textsc{Centering}$ finds a feasible point by following the central
path of the following modified linear program
\[
\min_{ \substack{\mb^{\top}(\vf^{(1)}+\vf^{(2)}-\vf^{(3)})=\vb \\
		\vl\leq\vf^{(1)}\leq\vu,\; 
		\vf^{(2)}\geq\vzero,\; 
		\vf^{(3)}\geq\vzero}}
\vc^{(1)\top}\vf^{(1)}+\vc^{(2)\top}\vf^{(3)}+\vc^{(2)\top}\vf^{(3)}
\]
where $\vc^{(1)}=\vc$, and $\vc^{(2)},\vc^{(3)}$ are some positive
large vectors. The above modified linear program is chosen so that we
know an explicit point on its central path, and any approximate
minimizer to this new linear program gives an approximate central path
point for the original problem.
The second call to $\textsc{Centering}$ finds an approximate solution
by following the central path of the original linear program. Note
that both calls run the same algorithm on essentially the same graph:
The only difference is that in the first call to $\textsc{Centering}$,
each edge $e$ of $G$ becomes three copies of the edge with flow value
$\vf_{e}^{(1)},\vf_{e}^{(2)},\vf_{e}^{(3)}$. Note that this edge
duplication does not affect planarity.

We note that the IPM algorithm only requires access to
$(\of,\os)$, but not $(\vf,\vs)$ during the main while loop.
Hence, $(\vf,\vs)$ can be implicitly maintained via any data
structure.  We only require $(\vf,\vs)$ explicitly when returning the
approximately optimal solution at the end of the algorithm
\cref{line:step_user_output}.
\begin{theorem}
	\label{thm:IPM}
	Consider the linear program 
	\[
	\min_{\mb^{\top}\vf=\vb,\; \vl\leq\vf\leq\vu}\vc^{\top}\vf
	\]
	with $\mb\in\R^{m\times n}$. We are given a scalar $r>0$ such that \emph{there exists} some interior point $\vf_{\circ}$ satisfying 
	$\mb^{\top}\vf_{\circ}=\vb$ and $\vl+r\leq \vf_{\circ} \leq \vu-r.$\footnote{For any vector $\vv$ and scalar $x$, we define $\vv+x$ to be the vector obtained by adding $x$ to each coordinate of $\vv$. We define $\vv-x$ to be the vector obtained by subtracting $x$ from each coordinate of $\vv$.}
	Let $L=\|\vc\|_{2}$ and $R=\|\vu-\vl\|_{2}$. For any $0<\epsilon\leq1/2$,
	the algorithm $\textsc{RIPM}$ (\cref{alg:IPM_centering}) finds $\vf$ such that $\mb^{\top}\vf=\vb$,
	$\vl\leq\vf\leq\vu$ and
	\[
	\vc^{\top}\vf\leq\min_{\mb^{\top}\vf=\vb,\; \vl\leq\vf\leq\vu}\vc^{\top}\vf+\epsilon LR.
	\]
	Furthermore, the algorithm has the following properties:
	\begin{itemize}
		\item Each call of $\textsc{Centering}$ involves $O(\sqrt{m}\log m\log(\frac{mR}{\epsilon r}))$
		many steps, and $\ot$ is only updated $O(\log m\log(\frac{mR}{\epsilon r}))$
		times.
		\item In each step of \textsc{Centering}, the coordinate $i$ in $\mw,\vv$ changes only if $\of_{i}$
		or $\os_{i}$ changes.
		\item In each step of \textsc{Centering}, $h\|\vv\|_{2}=O(\frac{1}{\log m})$.
		\item \cref{line:step_given_begin} to \cref{line:step_given_end} takes $O(K)$
		time in total, where $K$ is the total number of coordinate changes
		in $\of,\os$.
	\end{itemize}
\end{theorem}
\begin{proof}
	The number of steps follows from Theorem A.1 in \cite{treeLPArxivV2}, with
	the parameter $w_{i}=\nu_{i}=1$ for all $i$. 
	The number of coordinate
	changes in $\mw,\vv$ and the runtime of \cref{line:step_given_begin}
	to \cref{line:step_given_end} follows directly from the formula of
	$\mu^{t}(\vf,\vs)_{i}$ and $\gamma^{t}(\vf,\vs)_{i}$. For the bound
	for $h\|\vv\|_{2}$, it follows from
	\[	h\|\vv\|_{2}\leq\alpha\frac{\|\sinh(\lambda\gamma^{\ot}(\of,\os))\|_{2}}{\|\cosh(\lambda\gamma^{\ot}(\of,\os))\|_{2}}\leq\alpha=O\left(\frac{1}{\log m}\right).
	\]
\end{proof}

A key idea in our paper involves the computation of projection
matrices required for the RIPM. Recall from the definition of
$\mproj_{\vw}$ in \cref{alg:IPM_centering},
the true projection matrix is
\begin{align}
	\mproj_{\vw} &\defeq\mw^{1/2}\mb(\mb^{\top}\mw\mb)^{-1}\mb^{\top}\mw^{1/2}. \nonumber
	\intertext{We let $\ml$ denote the weighted Laplacian where $\ml=\mb^{\top}\mw\mb$, so that}
	\mproj_{\vw} &= \mw^{1/2}\mb\ml^{-1}\mb^{\top}\mw^{1/2}. \label{eq:overview:mproj}
\end{align}

\begin{lemma}
  \label{lem:approx-projections}
  To implement \cref{line:step_user_begin} in
  \cref{alg:IPM_centering}, it suffices to find an approximate
  slack projection matrix $\widetilde{\mproj}_{\vw}$ satisfying
  $\norm{ 
  	\left(\widetilde{\mproj}_{\vw} - \mproj_{\vw} \right) \vv}_{2} \leq
  \alpha \norm{\vv}_2$ and
  $\mw^{-1/2} \widetilde{\mproj}_{\vw} \vv \in \mathrm{Range}(\mb)$;
  and an approximation flow projection matrix
  $\widetilde{\mproj}'_{\vw}$ satisfying
  $\norm{ \left (\widetilde{\mproj}'_{\vw} -\mproj_{\vw} \right) \vv}_{2} \leq \alpha \norm{\vv}_2$ and
  $\mb^\top \mw^{1/2} \widetilde{\mproj}'_{\vw} \vv =  \mb^\top \mw^{1/2} \vv$.
\end{lemma}
\begin{proof}
	We simply observe that setting $\vv^{\parallel} = \widetilde{\mproj}_{\vw} \vv$ and $\vv^{\perp} = \vv -  \widetilde{\mproj}'_{\vw} \vv$ suffices.
\end{proof}

In finding these approximate projection matrices, we apply ideas from nested dissection and approximate Schur complements to the matrix $\ml$.


\subsection{Nested dissection and approximate Schur complements}\label{subsec:overview_sc}
In this subsection, we discuss nested dissection and the corresponding Schur complements,
and explain how it relates to our goal of finding the approximate projection matrices for \cref{lem:approx-projections}.

As we will discuss later in the main proof, our LP formulation for the IPM uses a
\emph{modified planar graph} which includes two additional vertices and $O(n)$
additional edges to the original planar graph. 
Although the modified graph is no longer planar, 
it has only two additional vertices.
We may add these two vertices to any relevant sets defined in nested dissection without changing the overall complexity.
As such, we can apply nested dissection as we would for planar graphs. 

We first illustrate the key ideas using a two-layer nested dissection scheme. 
By the well-known planar separator theorem~\cite{LiptonT79}, 
a planar graph $G$ can be decomposed into two edge-disjoint (not vertex-disjoint)
subgraphs $H_1$ and $H_2$ called \emph{regions}, such that each
subgraph has at most $2n/3$ vertices. 
Let $\partial H_i$
denote the \emph{boundary} of region $H_i$, that is, the set of
vertices $v \in H_i$ such that $v$ is adjacent to some $u\notin H_i$.
Then $\partial H_i$ has size bounded by $O(\sqrt{n})$.
Let $F_{H_i} = V(H_i) \setminus \partial H_i$ denote the remaining interior vertices \emph{eliminated} at region $H_i$.

Let $C = \partial H_1 \cup \partial H_2$ denote the union of the boundaries, 
and let $F = F_{H_1} \cup F_{H_2}$ be the disjoint union of the two interior sets.
Note that $C$ is a \emph{balanced vertex separator} of $G$, with size
\[
|C|\leq |\partial H_1| + |\partial H_2| = O(\sqrt{n}).
\]

Furthermore, $F$ and $C$ give a natural partition of the vertices of $G$. Using
block Cholesky decomposition, we can now write\footnote{To keep
  notation simple, $\mm^{-1}$ will denote the Moore-Penrose
  pseudo-inverse for non-invertible matrices.}
\begin{equation} \label{eq:overview:Linv-2layer}
	\ml^{-1} = 
	\left[\begin{array}{cc}
		\mi & -\ml_{F,C}{\ml_{F,F}}^{-1}\\
		\mzero & \mi
	\end{array}\right]
	\left[\begin{array}{cc}
		{\ml_{F,F}}^{-1} & \mzero\\
		\mzero & {\sc(\ml,C)}^{-1}
	\end{array}\right]
	\left[\begin{array}{cc}
	\mi & \mzero \\
	-\ml_{F,C}{\ml_{F,F}}^{-1} & \mi
	\end{array} \right],
\end{equation}
where $\sc(\ml,C) \defeq \ml_{C,C}-\ml_{C,F}{\ml_{F,F}}^{-1}\ml_{F,C}$ is the \emph{Schur complement} of $\ml$ onto vertex set $C$,
and $\ml_{F,C}\in\R^{F\times C}$ is the $F\times C$-indexed submatrix of $\ml$.

The IPM in \cref{alg:IPM_centering} involves updating $\ml^{-1}$ in every step;
written as the above decomposition, we must in turn update the Schur
complement $\sc(\ml, C)$ in every step.  Hence, the update cost must
be sub-linear in $n$. Computing $\sc(\ml,C)$ exactly takes
$\Omega(|C|^{2})=\Omega(n)$ time, which is already too expensive.
Our key idea here is to maintain an approximate Schur complement,
which is of a smaller size based on the graph decomposition, 
and can be maintained in amortized $\sqrt{n}$ time per step throughout the IPM.

Let $\ml[H_i]$ denote the weighted Laplacian of the region $H_i$ for $i = 1,2$. 
Since these regions are edge-disjoint, we can write the Laplacian $\ml$ as the sum
\[
\ml=\ml[H_1]+\ml[H_2].
\]
Based on the graph decomposition, we have the Schur complement decomposition
\[
\sc(\ml,C) = \sc(\ml[H_1], C) +  \sc(\ml[H_2], C).
\]
This decomposition allows us to localize edge weight updates.
Namely, if the weight of edge $e$ is updated,
and $e$ is contained in region $H_i$,  
we only need to recompute the single Schur complement term for $H_i$, rather than both terms in the sum.

For the appropriate projection matrices in the IPM, 
it further suffices to maintain a sparse \emph{approximate Schur complement}
$\widetilde{\sc}(\ml[H_i], C) \approx \sc(\ml[H_i], C)$ for each region $H_i$ rather than the exact.
Then, the approximate Schur complement of $\ml$ on $C$ is given by
\begin{equation}
	\widetilde{\sc}(\ml,C) \defeq 
	\widetilde{\sc}(\ml[H_1],C) + \widetilde{\sc}(\ml[H_2],C). \label{eq:ScL_sum}
\end{equation}
Each term $\widetilde{\sc}(\ml[H_i], C)$ can be computed in time nearly-linear in the size of $H_i$.
Furthermore, $\widetilde{\sc}(\ml[H_i],C)$ is supported only
on the vertex set $\partial H_i$, which is of size $O(\sqrt{n})$. 
Hence, any \emph{sparse} approximate Schur complement has only $\otilde(\sqrt{n})$ edges.
When we need to compute $\tsc(\ml, C)^{-1} \vx$ for some vector $\vx$, we use a generic SDD-solver which runs in  $\widetilde{O}(|C|)$ time; this is crucial in bounding the overall runtime. 

To extend the two-level scheme to more layers, we apply nested dissection recursively to each region $H_i$, until the regions are of constant size. 
This recursive procedure naturally gives rise to a \emph{separator tree} $\ct$ of the input graph $G$, which we discuss in detail in \cref{subsec:construct_tree}. 
Each node of $\ct$ correspond to a region of $G$, and can be obtained by taking the edge-disjoint union of the regions of its two children. 
Taking the union over all leaf regions gives the original graph $G$. 
The separator tree $\ct$ allows us to define a set $F_H$ of \emph{eliminated vertices} and a set $\bdry{H}$ of \emph{boundary vertices} for each node $H$, analogous to what was shown in the two-layer dissection.
Moreover, if we let $F_i$ denote the disjoint union of sets $F_H$ over all nodes $H$ at level $i$, 
and $C_i$ denote the union of sets $\bdry{H}$, 
then we essentially generalize the set $C$ from the two-layer dissection to $V(G) = C_{-1} \supset C_0 \supset \dots \supset C_{\eta-1} \supset C_{\eta} = \emptyset$, where each $C_{i}$ is some vertex separator of $G \setminus C_{i-1}$, 
and generalize the set $F$ to $F_0, \dots, F_\eta$ partitioning $V(G)$, where $F_i \defeq C_{i-1} \setminus C_i$.
With a height-$\eta$ separator tree, we can write
\begin{equation}\label{eq:overview:Linv}
	\ml^{-1} = \mmu^{(0)\top}\cdots\mmu^{(\eta-1)\top}
	\left[
	\begin{array}{ccc}
		{\sc(\ml, C_{-1})_{F_0, F_0}}^{-1} & \mzero & \mzero\\
		\mzero & \ddots & \mzero\\
		\mzero & \mzero & {\sc(\ml, C_{\eta-1})_{F_\eta, F_\eta}}^{-1}
	\end{array}
	\right]
	\mmu^{(\eta-1)}\cdots\mmu^{(0)},
\end{equation}
for some explicit upper triangular matrices $\mmu^{(i)}$. Here, $\sc(\ml, C_{i})_{F_{i+1},F_{i+1}}$ denotes the $F_{i+1} \times F_{i+1}$ submatrix of $\sc(\ml, C_{i})$. 

In the expression \cref{eq:overview:Linv}, the Schur complement term $\sc(\ml, C_i)$ at level $i$ can further be decomposed at according to the nodes at the level. 
Then, we can obtain an approximation to $\ml^{-1}$ by using approximate Schur complements as follows:
\begin{restatable}[$\ml^{-1}$ approximation]{theorem}{LinvApprox}\label{thm:overview-L-inv-approx}
	Suppose for each $H \in \ct$, we have a Laplacian $\ml^{(H)}$ satisfying
	\[
	\ml^{(\region)} \approx_{\epssc} \sc(\ml[\region], \bdry{\region}\cup \elim{\region}).
	\]
	Then, we have
	\begin{equation} \label{eq:overview_Linv_approx}
		\ml^{-1} \approx_{\eta \epssc}
		\mpi^{(0)\top}\cdots\mpi^{(\eta-1)\top} \widetilde{\mga}
		\mpi^{(\eta-1)}\cdots\mpi^{(0)},
	\end{equation} 
	where
	\[
	\widetilde \mga = 
	\left[
	\begin{array}{cccc}
		\sum_{H \in \ct(0)} \left(\ml^{(H)}_{F_H, F_H}\right)^{-1} & \mzero & \mzero\\
		\mzero &  \ddots & \mzero\\
		\mzero & \mzero & \sum_{H \in \ct(\eta)} \left(\ml^{(H)}_{F_H, F_H}\right)^{-1}
	\end{array}
	\right].
	\]
	and
	\[
	\mpi^{(i)} = \mi - \sum_{H \in \ct(i)} \ml^{(H)}_{\bdry{\region}, F_H} \left( \ml^{(H)}_{F_H, F_H}\right)^{-1},
	\]
	where $\ct(i)$ denotes the set of nodes at level $i$ of $\ct$, and $\mi$ is the $n \times n$ identity matrix.
\end{restatable}

	Compared to $\cref{eq:overview:Linv}$, we see that $\widetilde\mga$ approximates the middle block-diagonal matrix, and $\mpi^{(i)}$ approximates $\mmu^{(i)}$.
	
	To compute and maintain the necessary $\ml^{(H)}$'s as the edge weights undergo updates throughout the IPM, we have the following data structure:

\begin{restatable}[Schur complements maintenance]
	{theorem}{apxscMaintain}\label{thm:apxsc}
	Given a modified planar graph $G$ with $m$ edges and its separator tree $\ct$ with height $\eta = O(\log m)$, the deterministic data structure \textsc{DynamicSC} (\cref{alg:dynamicSC}) maintains the edge weights $\vw$ from the IPM, and at every node $H \in \ct$, maintains two vertex sets $F_H$ and $\bdry{H}$, and two Laplacians $\ml^{(H)}$ and $\tsc(\ml^{(H)}, \partial H \cup F_H)$ dependent on $\vw$.	It supports the following procedures:
	\begin{itemize}
		\item $\textsc{Initialize}(G,\vw\in\R_{>0}^{m}, \epssc > 0)$:
		Given a graph $G$, initial weights $\vw$, projection matrix approximation accuracy $\epssc$, preprocess in $\widetilde{O}(\epssc^{-2} m)$ time.
		\item $\textsc{Reweight}(\vw\in\R_{>0}^{m}, \text{given implicitly as a set of changed coordinates})$: 
		Update the weights to $\vw$, and update the relevant Schur complements in $\widetilde{O}( \epssc^{-2}  \sqrt{mK})$ time, where $K$ is the number of coordinates changed in $\vw$. 
		
		If $\collN$ is the set of leaf nodes in $\ct$ that contain an edge whose weight is updated,
		then $\ml^{(H)}$ and $\tsc(\ml^{(H)}, \bdry{\region})$ are updated only for nodes $H\in \pathT{\collN}$.
		\item Access to Laplacian $\ml^{(H)}$ at any node $H \in \ct$ in time $\widetilde{O}\left( \epssc^{-2} |\partial H \cup F_H|\right)$.
		\item Access to Laplacian $\tsc(\ml^{(H)}, \bdry{\region})$ at any node $H \in \ct$ in time $\widetilde{O}\left( \epssc^{-2} |\partial H|\right)$.
	\end{itemize}
	Furthermore, the $\ml^{(H)}$'s maintained by the data structure satisfy
	\begin{equation} \label{eq:L^(H)}
		 \ml^{(\region)} \approx_{\epssc} \sc(\ml[H], \bdry{H}\cup \elim{H}),
	\end{equation}
	for all $H\in \ct$ with high probability. The $\tsc(\ml^{(H)}, \bdry{\region})$'s maintained satisfy
	\begin{equation} \label{eq:tscL^(H)}
		 \tsc(\ml^{(H)}, \bdry{\region})\approx_{\epssc} \sc(\ml[\region], \bdry{\region})
	\end{equation}
	for all $H\in \ct$ with high probability.
\end{restatable}


\subsection{Implicit representations using tree operator}
\label{subsec:overview_representation}

In this section, we outline the data structures for maintaining the flow and slack solutions $\vf, \vs$ as needed in \cref{alg:IPM_centering}, \cref{line:maintain_f_s}.
Recall from \cref{lem:approx-projections}, at IPM step $k$ with step direction $\vv^{(k)}$, we want to update
\begin{align*}
	\vs &\leftarrow\vs+\ot h\mw^{-1/2}\widetilde{\mproj}_{\vw} \vv^{(k)}, \\
	\vf &\leftarrow\vf+ h \mw^{1/2}\vv^{(k)} - h \mw^{1/2} \widetilde{\mproj}'_{\vw} \vv^{(k)},
\end{align*}
for some approximate projection matrices $\widetilde{\mproj}_{\vw}$ and $\widetilde{\mproj}'_{\vw}$ satisfying $\range{\mw^{-1/2}\widetilde{\mproj}_{\vw}}\subseteq \range{\mb}$ and $\mb^\top \mw^{1/2}\widetilde{\mproj}'_{\vw} = \mb^\top \mw^{1/2}$.
The first term for the flow update is straightforward to maintain.
For this overview, we therefore focus on maintaining the second term
\[
\pf \leftarrow \pf + h \mw^{1/2} \widetilde{\mproj}'_{\vw} \vv^{(k)}.
\]
Computing $\widetilde{\mproj}_{\vw} \vv^{(k)}$ and $ \widetilde{\mproj}'_{\vw} \vv^{(k)}$ respectively
is too costly to do at every IPM step.
Instead, we maintain vectors $\vs_0, \pf_0, \vz$, and implicitly maintain two linear operators $\mm^{\slack}, \mm^{\flow}$ which depend on the weights $\vw$,
so at the end of every IPM step, the correct current solutions $\vs, \pf$ are recoverable via the identity
\begin{align*}
	\vs &= \vs_0 + \mm^{\slack} \vz \\
	\pf &= \pf_0 + \mm^{\flow} \vz.
\end{align*}
In this subsection, we abstract away the difference between slack and flow, and give a general data structure \textsc{MaintainRep} to maintain $\vx = \vy + \mm \vz$ for $\mm$ with a special tree structure.

At a high level, \textsc{MaintainRep} implements the IPM operations \textsc{Move} and \textsc{Reweight} as follows: 
To move in step $k$ with direction $\vv^{(k)}$ and step size $\alpha^{(k)}$, 
the data structure first computes $\vz^{(k)}$ as a function of $\vv^{(k)}$,
then updates $\vz \leftarrow \vz + \alpha^{(k)} \vz^{(k)}$, which translates to the desired overall update in $\vx$ of $\vx \leftarrow \vx + \mm (\alpha^{(k)} \vz^{(k)})$. 
To reweight with new weights $\vw^\new$ (which does not change the value of $\vx$), 
the data structure first computes $\mm^{\new}$ using $\vw^\new$ and $\Delta \mm \defeq \mm^{\new} - \mm$,
then updates $\mm \leftarrow \mm^{\new}$. This causes an increase in value in the $\mm \vz$ term by $\Delta \mm \vz$, 
which is then offset in the $\vy$ term with $\vy \leftarrow \vy - \Delta \mm \vz$.

In later sections, we will define $\mm^{\slack}$ and $\mm^{\flow}$ so that $\mm^{\slack} \vz^{(k)} = \mw^{1/2} \widetilde{\mproj}_{\vw} \vv^{(k)}$ and $\mm^{\flow} \vz^{(k)} = \mw^{-1/2} \widetilde{\mproj}'_{\vw} \vv^{(k)}$ for the desired approximate projection matrices.
With these operators appropriately defined, observed that \textsc{MaintainRep} correctly captures the updates to $\vs$ and $\pf$ at every IPM step. 

Let us now discuss the definition of $\vz$, which is common to both slack and flow:
Recall the \textsc{DynamicSC} data structure from the previous section maintains some Laplacian $\ml^{(H)}$ for every node $H$ in the separator tree $\ct$, so that at each IPM step, we can implicitly represent the matrices $\mpi^{(0)}, \cdots, \mpi^{(\eta-1)}, \widetilde{\mga}$ based on the current weights $\vw$, which together give an $\eta \epssc$-approximation of $\ml^{-1}$. 
\textsc{MaintainRep} will contain a \textsc{DynamicSC} data structure, so we can use these Laplacians in the definition of $\vz$:

At step $k$, let
\[
	\vz^{(k)} \defeq \widetilde{\mga} \mpi^{(\eta-1)} \cdots \mpi^{(0)} \mb^{\top} \mw^{1/2} \vv^{(k)},
\]
where $\widetilde{\mga}$, the $\mpi^{(i)}$'s, and $\mw$ are based on the state of the data structure at the end of step $k$.
$\vz$ is defined to be the accumulation of $\alpha^{(i)} \vz^{(i)}$'s up to the current step; that is, at the end of step $k$,
\[
	\vz = \sum_{i=1}^{k} \alpha^{(i)} \vz^{(i)}.
\]
Rather than naively maintaining $\vz$, we decompose $\vz$ and explicitly maintaining $c, \zprev$, and $\zsum$,
such that
\[
\vz \defeq c \cdot \zprev + \zsum,
\]
where we have the additional guarantee that at the end of IPM step $k$,
\[
\zprev = \widetilde\mga \mpi^{(\eta - 1)} \cdots \mpi^{(0)} \mb^{\top} \mw^{1/2} \vv^{(k)}.
\]
The other term, $\zsum$, is some remaining accumulation so that the overall representation is correct.

The purpose of this decomposition of $\vz$ is to facilitate sparse updates to $\vv$ between IPM steps: 
Suppose $\vv^{(k)}$ differ from $\vv^{(k-1)}$ on $K$ coordinates, then we can update $\zprev$ and $\zsum$ with runtime as a function of $K$, while producing the correct overall update in $\vz$.
Specifically, we decompose $ \vv^{(k)} =  \vv^{(k-1)} + \Delta \vv$.
We compute $\Delta \zprev =  \widetilde \mga \mpi^{(\eta-1)}\cdots \mpi^{(0)}  \mb^{\top} \mw^{1/2} \Delta \vv$, and then set
\[
\zprev \leftarrow \zprev + \Delta \zprev, \;
\zsum \leftarrow \zsum - c \cdot \Delta \zprev,
c \leftarrow c + \alpha, \;
\]
which can be performed in $O(nnz(\Delta \zprev))$ time.

Let us briefly discuss how to compute $\widetilde \mga \mpi^{(\eta-1)}\cdots \mpi^{(0)} \vd$ for some vector $\vd$. We use the two-layer nested dissection setup from \cref{subsec:overview_sc} for intuition, so
\[
\widetilde \mga \mpi^{(0)} \vd = 
\left[\begin{array}{cc}
	\ml_{F,F}^{-1} & \mzero\\
	\mzero & \tsc(\ml,C)^{-1}
\end{array}\right]\left[\begin{array}{cc}
	\mi & 0\\
	-\ml_{C,F}\ml_{F,F}^{-1} & \mi
\end{array}\right] \vd.
\]

The only difficult part for the next left matrix multiplication is $-\ml_{C,F}\ml_{F,F}^{-1}$.
However, we note that $\ml_{F,F}$ is block-diagonal with two blocks, each corresponding to a region generated during nested dissection. Hence, we can solve the Laplacians on the two subgraphs separately.
%
Next, we note that the two terms of $\ml_{C,F}\ml_{F,F}^{-1} \vd$ are both fed into $\tsc(\ml,C)^{-1}$, and we solve this Laplacian in time linear in the size of $\tsc(\ml,C)$. The rest of the terms are not the bottleneck in the overall runtime.
In the more general nested-dissection setting with $O(\log n)$ layers, we solve a sequence of Laplacians corresponding to the regions given by paths in the separator tree. We can bound the runtime of these Laplacian solves by the size of the corresponding regions for the desired overall runtime.

On the other hand, to work with $\mm$ efficiently, we define the notion of a \emph{tree operator} $\mm$ supported on a tree.
In our setting, we use the separator tree $\ct$.
Informally, our tree operator is a linear operator mapping $\R^{V(G)}$ to $\R^{E(G)}$.
It is constructed from the concatenation of a collection of \emph{edge operators} and \emph{leaf operators} defined on the edges and leaves of $\ct$.
If $H$ is a node in $\ct$ with parent $P$, then the edge operator for edge $(H,P)$ will map vectors supported on $\bdry{P} \cup F_P$ to vectors supported on $\bdry{H} \cup F_H$.
If $H$ is a leaf node, the leaf operator for $H$ will map vectors on $\bdry{H} \cup F_H$ to vectors on $E(H)$.
In this way, we take advantage of the recursive partitioning of $G$ via $\ct$ to map a vector supported on $V(G)$ recursive to be supported on smaller vertex subsets and finally to the edges.
Furthermore, we will show that when edge weights update, the change to $\mm$ can be localized to a small collection of edge and leaf operators along some tree paths, thus allowing for an efficient implementation.
We postpone the formal definition of the operator until \cref{subsec:tree_operator}.

\begin{restatable}[Implicit representation maintenance] {theorem}{MaintainRepresentation} \label{thm:maintain_representation}
	Given a modified planar graph $G$ with $n$ vertices and $m$ edges, and its separator tree $\ct$ with height $\eta$,
	the deterministic data structure \textsc{MaintainRep} (\cref{alg:maintain_representation})
	 maintains the following variables correctly at the end of every IPM step:
	\begin{itemize}
		\item the dynamic edge weights $\vw$ and step direction $\vv$ from the current IPM step,
		\item a \textsc{DynamicSC} data structure on $\ct$ based on the current edge weights $\vw$,
		\item an implicitly represented tree operator $\mm$ supported on $\ct$ with complexity $T(K)$, \emph{computable using information from \textsc{DynamicSC}},
		\item scalar $c$ and vectors $\zprev, \zsum$, which together represent $\vz = c \zprev + \zsum$,
		such that at the end of step $k$,
		\[
		\vz = \sum_{i=1}^{k} \alpha^{(i)} \vz^{(i)},
		\]
		where $\alpha^{(i)}$ is the step size $\alpha$ given in \textsc{Move} for step $i$,
		\item $\zprev$ satisfies $\zprev = \widetilde{\mga} \mpi^{(\eta-1)} \cdots \mpi^{(0)} \mb^{\top} \mw^{1/2} \vv$,
		\item an offset vector $\vy$ which together with $\mm, \vz$ represent $\vx=\vy+\mm\vz$, such that after step $k$,
		\[
			\vx = \vx^{\init}+\sum_{i=1}^{k} \mm^{(i)} (\alpha^{(i)} \vz^{(i)}),
		\]
		where $\vx^{\init}$ is an initial value from \textsc{Initialize}, and $\mm^{(i)}$ is the state of $\mm$ after step $i$.
	\end{itemize}
	The data structure supports the following procedures:
	\begin{itemize}
		\item $\textsc{Initialize}(G, \ct, \mm, \vv\in\R^{m},\vw\in\R_{>0}^{m}, \vx^{\init} \in \R^m, \epsilon_{\mproj} > 0)$:
		Given a graph $G$, its separator tree $\ct$, a tree operator $\mm$ supported on $\ct$ with complexity $T$,
		initial step direction $\vv$, initial weights $\vw$, initial vector $\vx^{\init}$, and target projection matrix accuracy $\epsilon_{\mproj}$, preprocess in $\widetilde{O}(\epssc^{-2}m+T(m))$ time and set $\vx \leftarrow \vx^{\init}$.

		\item $\textsc{Reweight}(\vw \in\R_{>0}^{m}$ given implicitly as a set of changed coordinates):
		Update the weights to $\vw$.
		Update the implicit representation of $\vx$ without changing its value, so that all the variables in the data structure are based on the new weights.

		The procedure runs in 
		$\widetilde{O}(\epsilon_{\mproj}^{-2}\sqrt{mK}+T(K))$ total time,
		where $K$ is an upper bound on the number of coordinates changed in $\vw$ and the number of leaf or edge operators changed in $\mm$.
		There are most $\O(K)$ nodes $\region\in \ct$ for which $\zprev|_{F_H}$ and $\zsum|_{F_H}$ are updated.

		\item $\textsc{Move}(\alpha \in \R$, $\vv \in \R^{n}$ given implicitly as a set of changed coordinates):
		Update the current direction to $\vv$, and then $\zprev$ to maintain the claimed invariant.
		Update the implicit representation of $\vx$ to reflect the following change in value:
		\[
			\vx \leftarrow \vx + \mm (\alpha \zprev).
		\]
		The procedure runs in $\widetilde{O}(\epsilon_{\mproj}^{-2} \sqrt{mK})$ time,
		where $K$ is the number of coordinates changed in $\vv$ compared to the previous IPM step.

		\item $\textsc{Exact}()$:
		Output the current exact value of $\vx=\vy + \mm \vz$ in $\O(T(m))$ time.
	\end{itemize}
\end{restatable}

\subsection{Solution approximation}
\label{subsec:overview_sketch}

In the flow and slack maintenance data structures,
one key operation is to maintain vectors $\of,\os$ that are close to $\vf,\vs$ throughout the IPM.
Since we have implicit representations of the solutions of the form $\vx = \vy + \mm \vz$, 
we now show how to maintain $\ox$ close to $\vx$.
To accomplish this, we will give a meta data structure that solves this in
a more general setting. The data structure involves three steps; 
the first two steps are similar to \cite{treeLP} and the key contribution is the last step:
\begin{enumerate}
	\item We maintain an approximate vector by detecting coordinates of the exact vector $\vx$ with large changes. 
	In step $k$ of the IPM, for every $\ell$ such that $2^\ell | k$, we consider all coordinates of the approximate vector $\ox$ that did not change in the last $2^\ell$ steps. 
	If any of them is off by more than $\frac{\delta}{2\left\lceil \log m\right\rceil}$ from $\vx$, it is updated. 
	We can prove that each coordinate of $\ox$ has additive error at most $\delta$ compared to $\vx$. 
	The number of updates to $\ox$ will be roughly $O(2^{2\ell_k})$, where $2^{\ell_k}$ is the largest power of $2$ that divides $k$. This guarantees that $K$-sparse updates only happen $\sqrt{m/K}$ times throughout the IPM algorithm.
	\item We detect coordinates with large changes in $\vx$ via a random sketch and sampling using the separator tree. 
	We can sample a coordinate with probability exactly proportional to the magnitude of its change, when given access to the approximate sum of probabilities in each region of the separator tree and to the exact value of any single coordinate of $\vx$.
	\item We show how to maintain random sketches for vectors of the form
	$\vx = \vy + \mm \vz$,
	where $\mm$ is an implicit tree operator supported on a tree $\ct$. 
	Specifically, to maintain sketches of $\mm \vz$, we store intermediate sketches for every complete subtree of $\ct$ at their roots. When an edge operator of $\mm$ or a coordinate of $\vz$ is modified, we only need to update the sketches along a path in $\ct$ from a node to the root. 
	For our use case, the cost of updating the sketches at a node $\region$ will be proportional to its separator size, so that a $K$-sparse update takes $\otilde(\sqrt{mK})$ time. 
\end{enumerate}

While the data structure is randomized, it is guaranteed to
work against an adaptive adversary that is allowed to see the entire
internal state of the data structure, including the random bits.

\begin{restatable}[Approximate vector maintenance with tree operator]{theorem}{vectorTreeMaintain} \label{thm:VectorTreeMaintain}
Given a constant degree tree $\ct$ with height $\eta$ that supports tree operator $\mm$ with complexity $T$, 
there is a randomized data structure \textsc{MaintainApprox} that
takes as input the dynamic variables $\mm, c, \zprev, \zsum, \vy, \md$ at every IPM step,
and maintains the approximation $\ox$ to $\vx \defeq \vy + \mm \vz = \vy + \mm(c \cdot \zprev + \zsum)$ satisfying $\norm{\md^{1/2} (\vx - \ox)}_\infty \leq \delta$.
It supports the following procedures: 
\begin{itemize}
\item $\textsc{Initialize}(\text{tree } \ct, \text{tree operator } \mm, c \in \R, \zprev \in \R^n, \zsum \in \R^n, \vy\in \R^m, \md \in \R^{n \times n}, \rho>0, \delta>0)$:
Initialize the data structure with initial vector $\vx = \vy + \mm (c \zprev+ \zsum)$, diagonal scaling matrix $\md$, 
target approximation accuracy $\delta$, success probability $1-\rho$, in $O(m\eta^2\log m\log (\frac{m}{\rho}))$ time.
Initialize $\ox \leftarrow \vx$.
\item $\textsc{Approximate}(\mm, c, {\zprev}, {\zsum}, \vy, \md)$: 
Update the internal variables to their new iterations as given.
Then output a vector $\ox$ such that $\|\md^{1/2}(\vx-\ox)\|_{\infty}\leq\delta$
for the current vector $\vx$ and the current diagonal scaling $\md$. 

\end{itemize}
Suppose $\|\vx^{(k+1)}-\vx^{(k)}\|_{\md^{(k+1)}}\leq\beta$ \emph{for all $k$}, 
where $\md^{(k)}$ and $\vx^{(k)}$ are the $\md$ and
$\vx$ at the $k$-th call to \textsc{Approximate}. Then, for the $k$-th call to \textsc{Approximate}, we have
\begin{itemize}
\item the data structure first updates $\ox_i\leftarrow \vx_i^{(k-1)}$ for the coordinates $i$ with $\md_{ii}^{(k)}\neq\md_{ii}^{(k-1)}$, then updates $\ox_i\leftarrow \vx_i^{(k)}$ for 
$O(N_k \defeq 2^{2\ell_{k}}(\beta/\delta)^{2}\log^{2}m)$ 
coordinates, where $\ell_{k}$ is the largest integer $\ell$ with $k=0\bmod2^{\ell}$.
\item The amortized time cost of $\textsc{Approximate}$ is
\[
\Theta(\eta^{2}\log(\frac{m}{\rho})\log m)\cdot T(\eta \cdot(N_{k-2^{\ell_k}}+| \mathcal{S}|)),
\]
where $\mathcal{S}$ is the set of nodes $H$ where either $\mm_{(H,P)}$, $\mj_H$, $\zprev|_{F_H}$, or $\zsum|_{F_H}$ changed, or where $\vy_e$ or $\md_{e, e}$ changed for some edge $e$ in $H$, compared to the $(k-1)$-th step.
\end{itemize}
\end{restatable}

\subsection{Slack projection} \label{subsec:dual_overview}

We want to use a \textsc{MaintainRep} data structure to implicitly maintain the slack solution $\vs$ throughout the IPM,
and use a \textsc{MaintainApprox} data structure to explicitly maintain
the approximate slack solution $\os$.

To use \textsc{MaintainRep}, it remains to define a suitable tree operator $\mm^{\slack}$, so that at IPM step $k$, the update in \textsc{MaintainRep} is the correct IPM slack update; that is:
\[
	\mm^{\slack} (\bar{t}h \cdot \zprev) = \bar{t} h \mw^{-1/2} \widetilde{\mproj}_{\vw} \vv^{(k)}.
\]

Let $\widetilde{\ml}^{-1}$ denote the approximation of $\ml^{-1}$ from \cref{eq:overview_Linv_approx}, maintained and computable with a \textsc{DynamicSC} data structure. We define
\[
	\widetilde{\mproj}_{\vw} = \mw^{1/2} \mb \widetilde{\ml}^{-1} \mb^\top \mw^{1/2} = 
	\mw^{1/2} \mb \mpi^{(0)} \cdots \mpi^{(\eta-1)} \widetilde\mga \mpi^{(\eta-1)} \cdots \mpi^{(0)} \mb^\top \mw^{1/2}.
\]
then $\widetilde{\mproj}_{\vw} \approx_{\eta \epssc} \mproj_{\vw}$, and $\range{\widetilde{\mproj}_{\vw}}=\range{\mproj_{\vw}}$ by definition. Hence, this suffices as our approximate slack projection matrix.

Using \cref{subsec:overview_representation}, we can write
\begin{equation}  \label{eq:overview-slack-update}
\widetilde{\mproj}_{\vw} \vv^{(k)} = \mw^{1/2} \mb \mpi^{(0) \top} \cdots \mpi^{(\eta - 1) \top} \zprev,
\end{equation}
where $\zprev = \widetilde{\mga} \mpi^{(\eta-1)} \cdots \mpi^{(0)} \mb^\top \mw^{1/2} \vv^{(k)}$ at the end of IPM step $k$, as defined in the previous section. 
The remaining matrix multiplication on the left in \cref{eq:overview-slack-update} can indeed be represented by a tree operator $\mm$ on the tree $\ct$. 
Intuitively, observe that each $\mpi^{(i)}$ operates on level $i$ of $\ct$ and can be decomposed to be written in terms of the nodes at level $i$.
Furthermore, the $\mpi^{(i)}$'s are applied in order of descending level in $\ct$. 
Finally, at the leaf level, $\mw^{1/2} \mb$ maps vectors on vertices to vectors on edges.
In \cref{sec:slack_projection}, we present the exact tree operator and its correctness proof.
With it, we have
\[
	\widetilde{\mproj}_{\vw}\vv^{(k)} = \mw^{1/2} \mm \zprev.
\]
We set $\mm^{\slack}$ to be $\mw^{1/2} \mm$, which is also a valid tree operator.

Now, we state the full data structure for maintaining slack.

\begin{restatable}[Slack maintenance]{theorem}{SlackMaintain}
	\label{thm:SlackMaintain}
	Given a modified planar graph $G$ with $m$ edges and its separator tree $\ct$ with height $\eta$,
	the randomized data structure \textsc{MaintainSlack} (\cref{alg:slack-maintain-main})
	implicitly maintains the slack solution $\vs$ undergoing IPM changes,
	and explicitly maintains its approximation $\os$,
	and supports the following procedures with high probability against an adaptive adversary:
	\begin{itemize}
		\item $\textsc{Initialize}(G,\vs^{\init} \in\R^{m},\vv\in \R^{m}, \vw\in\R_{>0}^{m},\epsilon_{\mproj}>0,\overline{\epsilon}>0)$:
		Given a graph $G$, initial solution $\vs^{\init}$, initial direction $\vv$, initial weights $\vw$,
		target step accuracy $\epsilon_{\mproj}$ and target approximation accuracy
		$\overline{\epsilon}$, preprocess in $\widetilde{O}(m \epsilon_{\mproj}^{-2})$ time, 
		and set the representations $\vs \leftarrow \vs^{\init}$ and $\ox \leftarrow \vs$.
		
		\item $\textsc{Reweight}(\vw\in\R_{>0}^{m},$ given implicitly as a set of changed weights): 
		Set the current weights to $\vw$ in $\widetilde{O}(\epsilon_{\mproj}^{-2}\sqrt{mK})$ time,
		where $K$ is the number of coordinates changed in $\vw$.
		
		\item $\textsc{Move}(\alpha\in\mathbb{R},\vv\in\R^{m} $ given implicitly as a set of changed coordinates): 
		Implicitly update  $\vs \leftarrow \vs+\alpha\mw^{-1/2}\widetilde{\mproj}_{\vw} \vv$ for some
		$\widetilde{\mproj}_{\vw}$ with	$\|(\widetilde{\mproj}_{\vw} -\mproj_{\vw}) \vv \|_2 \leq\eta\epssc \norm{\vv}_2$,
		and $\widetilde{\mproj}_{\vw} \vv \in \range{\mb}$. 
		The total runtime is $\widetilde{O}(\epsilon_{\mproj}^{-2}\sqrt{mK})$ where $K$ is the number of coordinates changed in $\vv$.
		
		\item $\textsc{Approximate}()\rightarrow\R^{m}$: Return the vector $\os$
		such that $\|\mw^{1/2}(\os-\vs)\|_{\infty}\leq\overline{\epsilon}$
		for the current weight $\vw$ and the current vector $\vs$. 
		
		\item $\textsc{Exact}()\rightarrow\R^{m}$: 
		Output the current vector $\vs$ in $\O(m \epssc^{-2})$ time.
	\end{itemize}
	Suppose $\alpha\|\vv\|_{2}\leq\beta$ for some $\beta$ for all calls to \textsc{Move}.
	Suppose in each step, \textsc{Reweight}, \textsc{Move} and \textsc{Approximate} are called in order. Let $K$ denote the total number of coordinates changed in $\vv$ and $\vw$ between the $(k-1)$-th and $k$-th \textsc{Reweight} and \textsc{Move} calls. Then at the $k$-th \textsc{Approximate} call,
	\begin{itemize}
		\item the data structure first sets $\os_e\leftarrow \vs^{(k-1)}_e$ for all coordinates $e$ where $\vw_e$ changed in the last \textsc{Reweight}, then sets $\os_e\leftarrow \vs^{(k)}_e$ for $O(N_k\defeq 2^{2\ell_{k}}(\frac{\beta}{\overline{\epsilon}})^{2}\log^{2}m)$ coordinates $e$, where $\ell_{k}$ is the largest integer
		$\ell$ with $k=0\mod2^{\ell}$ when $k\neq 0$ and $\ell_0=0$. 
		\item The amortized time for the $k$-th \textsc{Approximate} call
		is $\widetilde{O}(\epsilon_{\mproj}^{-2}\sqrt{m(K+N_{k-2^{\ell_k}})})$.
	\end{itemize}
	
\end{restatable}

\subsection{Flow projection} \label{subsec:primal_overview}

Similar to slack, we want to use a \textsc{MaintainRep} data structure to implicitly maintain the flow solution $\vf$ throughout the IPM,
and use a \textsc{MaintainApprox} data structure to explicitly maintain
the approximate flow solution $\of$.
For the overview, we focus on the non-trivial part of the flow update at every step given by
\[
\pf \leftarrow \pf + h \mw^{1/2} \widetilde{\mproj}'_{\vw} \vv.
\]
To use \textsc{MaintainRep}, it remains to define a suitable tree operator $\mm^{\flow}$ so that at IPM step $k$, the update in \textsc{MaintainRep} is the correct IPM flow update; that is:
\[
 \mw^{1/2} \widetilde{\mproj}'_{\vw} \vv = \mm^{\flow} \zprev.
\]
Rather than finding an explicit $\widetilde{\mproj}'_{\vw}$ as we did for slack, 
observe it suffices to find some \emph{weighted flow}
$\tf \approx \mproj_{\vw} \vv$ satisfying $\mb^{\top} \mw^{1/2} \tf=\mb^{\top} \mw^{1/2} \vv$.
(We use the term ``weighted flow'' to mean it is obtained by multiplying the edge weights $\mw$ to some valid flow.)
Then the IPM update becomes
\[
	h \mw^{1/2} \widetilde{\mproj}'_{\vw} \vv = h \mw^{1/2} \tf.
\]
Hence, our goal is to write $\mw^{1/2} \tf = \mm^{\flow} \zprev$ for an appropriate weighted flow $\tf$.

Let us define demands on vertices by $\vd \defeq \mb^{\top} \mw^{1/2} \vv$.
Unwrapping the definition of $\mproj_{\vw}$, we see that the condition of $\tf \approx \mproj_{\vw} \vv$ is actually $\tf \approx \mw^{1/2} \mb \ml^{-1} \vd$. 
The second condition says $\tf$ is a weighted flow routing demand $\vd$. 
Suppose we had $\tf = \mw^{1/2} \mb \ml^{-1} \vd$ exactly, then we see immediately that the second condition is satisfied with $\mb^\top \mw^{1/2} \tf = \mb^\top \mw \mb \ml^{-1} \vd = \vd$.
To realize the approximation, we make use of the approximation of $\ml^{-1}$ from \cref{eq:overview_Linv_approx}.
Hence, one important fact about our construction is that when the Schur complements are exact,
our flow $\tf$ agrees with the true electrical flow routing the demand.

In constructing $\tf$ to route the demand $\vd$,
we show that $\tf$ can be written as $\mm \zprev$, where $\mm$ is a tree operator on the tree $\ct$,
and $\zprev$ is from \textsc{MaintainRep}, and in fact correspond to electric potentials. 
Here we explain what $\mm$ captures intuitively.
For simplicity, let $\vz$ denote $\zprev$.

The first step is recognizing a decomposition of $\vd$ using the separator tree,
such that we have a demand term $\vd^{(H)}$ for each node $H \in \ct$.
Furthermore, $\vd^{(H)} = \ml^{(H)} \vz|_{F_H}$, 
for the Laplacian $\ml^{(H)}$ supported on the region $H$ maintained by \textsc{dynamicSC}.
This decomposition allows us to route each demand $\vd^{(H)}$ by electric flows using only the corresponding region $H$, 
rather than the entire graph.
The recursive nature of the decomposition allows us to bound the overall runtime.
To show that the resulting flow $\tf$ indeed is close to the electric flow, 
one key insight is that the decomposed demands are orthogonal (\cref{lem:energy-graph-decomposition}).
Hence, routing them separately by electrical flows gives a good approximation
to the true electrical flow of the whole demand (\cref{thm:flow-forest-operator-correctness}).

Let us illustrate this partially using the two-layer decomposition scheme from \cref{subsec:overview_sc}:
Suppose we have a demand term $\vd$ \emph{that is non-zero only on vertices of $C$}. 
Then, observe that
\[
\vz=\left[\begin{array}{cc}
\ml_{F,F}^{-1} & \mzero\\
\mzero & \widetilde{\sc}(\ml,C)^{-1}
\end{array}\right]
\left[\begin{array}{cc}
\mi & \mzero\\
-\ml_{C,F}\ml_{F,F}^{-1} & \mi
\end{array}\right] \vd
\]
Looking at the sub-vector indexed by $C$ on both sides, we have that
\[
\widetilde{\sc}(\ml,C) \vz = \vd
\]
where we abuse the notation to extend $\widetilde{\sc}(\ml,C)$ from
$C\times C$ to $[n]\times[n]$ by padding zeros. Using \cref{eq:ScL_sum},
we have
\[
\left( \widetilde{\sc}(\ml[H_1],C) + \widetilde{\sc}(\ml[H_2],C) \right)  \vz =\vd
\]
This gives a decomposition of the demand $\vd$ into demand terms $\widetilde{\sc}(\ml[H_i],C)\vz$ for $i = 1,2$.
Crucially, each demand $\widetilde{\sc}(\ml[H_i],C)\vz$ is supported on the vertices of the region $H_i$, and we can route the flow on the corresponding region only. 
In a $O(\log n)$-level decomposition, we recursively decompose the demand further based on the sub-regions according to the separator tree $\ct$.
This guarantees that $\tf_i$ is the electric flow on the subgraph $H_i$ that satisfies the demand $\widetilde{\sc}(\ml[H_i],C)\vz$. 
Finally, we will let the output be $\tf =\sum \tf_{i}$. 
By construction, this $\tf$ satisfies $\mb^{\top} \mw^{1/2} \tf=\vd=\mb^{\top} \mw^{1/2} \vv$. 

In \cref{sec:flow_projection}, we show that this recursive operation can be realized using a tree operator.
We then present the full proof for \cref{thm:FlowMaintain} below, and implement the data structure. 

\begin{restatable}[Flow maintenance]{theorem}{FlowMaintain}
	\label{thm:FlowMaintain}
	Given a modified planar graph $G$ with $m$ edges and its separator tree $\ct$ with height $\eta$,
	the randomized data structure \textsc{MaintainFlow} (\cref{alg:flow-maintain-main})
	implicitly maintains the flow solution $\vf$ undergoing IPM changes,
	and explicitly maintains its approximation $\of$, 
	and supports the following procedures with high probability against an adaptive adversary:
	\begin{itemize}
		\item
		$\textsc{Initialize}(G,\vf^{\init} \in\R^{m},\vv \in \R^{m}, \vw\in\R_{>0}^{m},\epsilon_{\mproj}>0,
		\overline{\epsilon}>0)$: Given a graph $G$, initial solution $\vf^\init$, initial direction $\vv$, 
		initial weights $\vw$, target step accuracy $\epsilon_{\mproj}$,
		and target approximation accuracy $\overline{\epsilon}$, 
		preprocess in $\widetilde{O}(m \epsilon_{\mproj}^{-2})$ time and set the internal representation $\vf \leftarrow \vf^{\init}$ and $\of \leftarrow \vf$.
		\item $\textsc{Reweight}(\vw\in\R_{>0}^{m}$ given implicitly as a set of changed weights): Set the current weights to $\vw$ in
		$\widetilde{O}(\epsilon_{\mproj}^{-2}\sqrt{mK})$ time, where $K$ is
		the number of coordinates changed in $\vw$.
		\item $\textsc{Move}(\alpha\in\mathbb{R},\vv\in\R^{m}$ given
		implicitly as a set of changed coordinates): 
		Implicitly update
		$\vf \leftarrow \vf+ \alpha \mw^{1/2}\vv - \alpha \mw^{1/2} \widetilde{\mproj}'_{\vw} \vv$ for
		some $\widetilde{\mproj}'_{\vw} \vv$, 
		where 
		$\|\widetilde{\mproj}'_{\vw} \vv - \mproj_{\vw} \vv \|_2 \leq O(\eta \epssc) \norm{\vv}_2$ and
		$\mb^\top \mw^{1/2}\widetilde{\mproj}'_{\vw}\vv= \mb^\top \mw^{1/2} \vv$.
		The runtime is $\widetilde{O}(\epsilon_{\mproj}^{-2}\sqrt{mK})$, where $K$ is
		the number of coordinates changed in $\vv$.
		\item $\textsc{Approximate}()\rightarrow\R^{m}$: Output the vector
		$\of$ such that $\|\mw^{-1/2}(\of-\vf)\|_{\infty}\leq\overline{\epsilon}$ for the
		current weight $\vw$ and the current vector $\vf$. 
		\item $\textsc{Exact}()\rightarrow\R^{m}$: 
		Output the current vector $\vf$ in $\O(m \epssc^{-2})$ time.
	\end{itemize}
	Suppose $\alpha\|\vv\|_{2}\leq\beta$ for some $\beta$ for all calls to \textsc{Move}.
	Suppose in each step, \textsc{Reweight}, \textsc{Move} and \textsc{Approximate} are called in order. Let $K$ denote the total number of coordinates changed in $\vv$ and $\vw$ between the $(k-1)$-th and $k$-th \textsc{Reweight} and \textsc{Move} calls. Then at the $k$-th \textsc{Approximate} call,
	\begin{itemize}
	\item the data structure first sets $\of_e\leftarrow \vf^{(k-1)}_e$ for all coordinates $e$ where $\vw_e$ changed in the last \textsc{Reweight}, then sets $\of_e\leftarrow \vf^{(k)}_e$ for $O(N_k\defeq 2^{2\ell_{k}}(\frac{\beta}{\overline{\epsilon}})^{2}\log^{2}m)$ coordinates $e$, where $\ell_{k}$ is the largest integer
		$\ell$ with $k=0\mod2^{\ell}$ when $k\neq 0$ and $\ell_0=0$. 
	\item The amortized time for the $k$-th \textsc{Approximate} call
	is $\widetilde{O}(\epsilon_{\mproj}^{-2}\sqrt{m(K+N_{k-2^{\ell_k}})})$.
	\end{itemize}
\end{restatable}


\subsection{Main proof\label{subsec:overview_proof}}

We are now ready to prove our main result.
\cref{algo:IPM_impl} presents the implementation of RIPM \cref{alg:IPM_centering} using our data structures.

\begin{algorithm}

\caption{Implementation of Robust Interior Point Method\label{algo:IPM_impl}}

\begin{algorithmic}[1]

\Procedure{$\textsc{CenteringImpl}$}{$\mb,\phi,\vf,\vs,t_{\mathrm{start}},t_{\mathrm{end}}$}

\State $G$: graph on $n$ vertices and $m$ edges with incidence matrix $\mb$
\State $\mathcal{S}, \mathcal{F}$: data structures for slack and flow maintenance 
	\Comment \cref{thm:SlackMaintain,thm:FlowMaintain}

\State $\alpha \defeq \frac{1}{2^{20}\lambda}, \lambda \defeq 64\log(256m^{2})$

\State $t\leftarrow t_{\mathrm{start}}$, $\of\leftarrow\vf,\os\leftarrow\vs,\ot\leftarrow t$, $\mw\leftarrow\nabla^{2}\phi(\of)^{-1}$
\Comment variable initialization
\State $\vv_{i} \leftarrow \sinh(\lambda\gamma^{\ot}(\of,\os)_{i})$ for all $i \in [n]$

\Comment data structure initialization
\State $\mathcal{F}.\textsc{Initalize}(G,\vf,\vv, \mw,\epssc=O(\alpha/\log m),\overline{\epsilon}=\alpha)$   
\Comment choose $\epssc$ so $\eta\epssc \leq \alpha$ in \cref{thm:SlackMaintain} 
\State $\mathcal{S}.\textsc{Initalize}(G,\ot^{-1}\vs,\vv, \mw,\epssc=O(\alpha/\log m),\overline{\epsilon}=\alpha)$ 
\Comment and $O(\eta\epssc) \leq \alpha$ in \cref{thm:FlowMaintain}

\While{$t\geq t_{\mathrm{end}}$}

\State $t\leftarrow\max\{(1-\frac{\alpha}{\sqrt{m}})t, t_{\mathrm{end}}\}$

\State Update $h=-\alpha/\|\cosh(\lambda\gamma^{\ot}(\of,\os))\|_{2}$
\Comment $\gamma$ as defined in \cref{eq:mu_t_def}\label{line:step_given_begin_impl}

\State Update the diagonal weight matrix $\mw=\nabla^{2}\phi(\of)^{-1}$\label{line:step_given_3_impl}

\State $\mathcal{F}.\textsc{Reweight}(\mw)$ \Comment update the implicit representation of $\vf$ with new weights
\State $\mathcal{S}.\textsc{Reweight}(\mw)$ \Comment update the implicit representation of $\vs$ with new weights

\State $\vv_{i} \leftarrow \sinh(\lambda\gamma^{\ot}(\of,\os)_{i})$ for all $i$ where $\of_i$ or $\os_i$ has changed \label{line:step_given_end_impl} 
\Comment update direction $\vv$

\State \Comment $\mproj_{\vw}\defeq\mw^{1/2}\mb(\mb^{\top}\mw\mb)^{-1}\mb^{\top}\mw^{1/2}$
\State $\mathcal{F}.\textsc{Move}(h,\vv)$
\Comment  Update $\vf\leftarrow\vf+h\mw^{1/2}\vv - h \mw^{1/2} \tf$ 
with $\tf \approx \mproj_{\vw} \vv$ 

\State $\mathcal{S}.\textsc{Move}(h,\vv)$ \Comment Update $\vs\leftarrow\vs+\ot h\mw^{-1/2}\ts$ with $\ts \approx\mproj_{\vw} \vv$

\State $\of\leftarrow\mathcal{F}.\textsc{Approximate}()$ \Comment Maintain $\of$ such that $\|\mw^{-1/2}(\of-\vf)\|_{\infty}\leq\alpha$

\State  $\os\leftarrow\bar{t}\mathcal{S}.\textsc{Approximate}()$ \Comment Maintain $\os$ such that  $\|\mw^{1/2}(\os-\vs)\|_{\infty}\leq\ot\alpha$

\label{line:step_user_end_impl}

\If {$|\ot-t|\geq\alpha\ot$} 
	\State $\vs\leftarrow \ot\mathcal{S}.\textsc{Exact}()$
	\State $\ot\leftarrow t$ 
	\State  $\mathcal{S}.\textsc{Initalize}(G, \ot^{-1}\vs, \vv, \mw,\epssc=O(\alpha/\log m),\overline{\epsilon}=\alpha)$ 
\EndIf 

\EndWhile

\State \Return $(\mathcal{F}.\textsc{Exact}(),\ot\mathcal{S}.\textsc{Exact}())$\label{line:step_user_output_impl}

\EndProcedure

\end{algorithmic}

\end{algorithm}

We first prove a lemma about how many coordinates change in $\vw$ and $\vv$ in each step. This is useful for bounding the complexity of each iteration.
\begin{lemma}\label{lem:cor_change}
When updating $\vw$ and $\vv$ at the $(k+1)$-th step of the $\textsc{CenteringImpl}$ algorithm,
$\vw$ and $\vv$ change in $O(2^{2\ell_{k-1}}\log^{2}m+2^{2\ell_{k}}\log^{2}m)$ coordinates, where $\ell_{k}$
is the largest integer $\ell$ with $k \equiv 0\mod2^{\ell}$.
\end{lemma}
\begin{proof}
	Since both $\vw$ and $\vv$ are an entry-wise function of $\of,\os$ and $\ot$, we need to examine these variables.
	First, $\ot$ changes every $\otilde(\sqrt{m})$ steps, and when $\ot$ changes, every coordinate of $\vw$ and $\vv$ changes.
	Over the entire \textsc{CenteringImpl} run, $\ot$ changes $\otilde(1)$ number of times, so we may incur an additive $\otilde(m)$ term overall, and assume $\ot$ does not change for the rest of the analysis.
	
	By \cref{thm:IPM}, we have $h\|\vv\|_2 = O(\frac{1}{\log m})$ at all steps. So we apply \cref{thm:SlackMaintain}  and \cref{thm:FlowMaintain} both with parameters $\beta = O(\frac{1}{\log m})$ and $\overline{\epsilon}=\alpha=\Theta(\frac{1}{\log m})$. We use their conclusions in the following argument.
	Let the superscript $^{(k)}$ denote the variable at the end of the $k$-th step.
	
	By definition, $\vw^{(k+1)}$ is an entry-wise function of $\of^{(k)}$, and recursively, $\of^{(k)}$ is an entry-wise function of $\vw^{(k)}$.
	We first prove inductively that at step $k$, $O(2^{2\ell_{k}}\log^{2}m)$ coordinates of $\of$ change to $\vf^{(k)}$ where $\vf^{(k)}$ is the exact solution, and there are no other changes. This allows us to conclude that $\vw^{(k+1)}$ differ from $\vw^{(k)}$ on $O(2^{2\ell_{k}} \log^2 m)$ coordinates.
	
	In the base case at step $k=1$, because $\vw^{(1)}$ is equal to the initial weights $\vw^{(0)}$, 
	only $O(2^{2\ell_{1}}\log^{2}m)$ coordinates $\of_e$ change to $\vf^{(1)}_e$.
	Suppose at step $k$, a set $S$ of $O(2^{2\ell_{k}}\log^{2}m)$ coordinates of $\of$ change; that is, $\of|_S$ is updated to $\vf^{(k)}|_S$, and there are no other changes. 
	Then at step $k+1$, by definition, $\vw^{(k+1)}$ differ from $\vw^{(k)}$ exactly on $S$, 
	and in turn, $\of^{(k+1)}|_S$ is set to $\vf^{(k)}|_S$ again (\cref{line:D-induced-ox-change} of \cref{algo:maintain-vector}). 
	In other words, there is no change from this operation. 
	Then, $O(2^{2\ell_{k+1}}\log^{2}m)$ additional coordinates $\of_e$ change to $\vf^{(k+1)}_e$.
	
	Now, we bound the change in $\os$: \cref{thm:SlackMaintain} guarantees that in the $k$-th step, there are $O(2^{2\ell_{k}}\log^{2}m)+D$ coordinates in $\os$ that change, where $D$ is the number of changes between $\vw^{(k-1)}$ and $\vw^{(k)}$ and is equal to $O(2^{2\ell_{k-1}} \log^2 m)$ as shown above.
	
	Finally, $\vv^{(k+1)}$ is an entry-wise function of $\of^{(k)}$ and $\os^{(k)}$, so we conclude that $\vv^{(k+1)}$ and $\vv^{(k)}$ differ on at most $O(2^{2\ell_{k}}\log^{2}m) + 2 \cdot O(2^{2\ell_{k-1}}\log^{2}m)$ coordinates.
\end{proof}

\mainthm*

\begin{proof}
The proof is structured as follows. We first write the minimum cost
flow problem as a linear program of the form \cref{eq:LP}. We prove
the linear program has an interior point and is bounded, so to satisfy the assumptions
in \cref{thm:IPM}. Then, we implement the IPM algorithm using the
data structures from \cref{subsec:overview_representation,subsec:overview_sketch,subsec:dual_overview,subsec:primal_overview}.
Finally, we bound the cost of each operations of the data structures.

To write down the min-cost flow problem as a linear program of the form \cref{eq:LP}, we
add extra vertices $s$ and $t$.
Let $\vd$ be the demand vector of the min-cost flow problem.
For every vertex $v$ with $\vd_{v}<0$, we add a directed edge from
$s$ to $v$ with capacity $-\vd_{v}$ and cost $0$. For every vertex
$v$ with $\vd_{v}>0$, we add a directed edge from $v$ to $t$
with capacity $\vd_{v}$ and cost $0$.
Then, we add a directed edge
from $t$ to $s$ with capacity $4nM$ and cost $-4nM$. The modified graph is no longer planar but it has only two extra vertices $s$ and $t$.

The cost and capacity on the $t\rightarrow s$ edge is chosen such that
the minimum cost flow problem on the original graph is equivalent
to the minimum cost circulation on this new graph. Namely, if the
minimum cost circulation in this new graph satisfies all the demand
$\vd_{v}$, then this circulation (ignoring the flow on the new edges)
is the minimum cost flow in the original graph.

Since \cref{thm:IPM} requires an interior point in the polytope,
we first remove all directed edges $e$ through which no flow from $s$ to $t$ can
pass. To do this, we simply check, for every directed edge
$e=(v_{1},v_{2})$, if $s$ can reach $v_{1}$ and if $v_{2}$ can
reach $t$. This can be done in $O(m)$ time by a BFS from $s$ and
a reverse BFS from $t$.
With this preprocessing, we write the minimum cost circulation problem
as the following linear program
\[
\min_{\mb^{\top}\vf=\vzero,\; \vl^{\mathrm{new}}\leq\vf\leq\vu^{\mathrm{new}}}(\vc^{\mathrm{new}})^{\top}\vf
\]
where $\mb$ is the signed incidence matrix of the new graph, $\vc^{\mathrm{new}}$
is the new cost vector (with cost on extra edges), and $\vl^{\mathrm{new}},\vu^{\mathrm{new}}$
are the new capacity constraints. If an edge $e$ has only one direction,
we set $\vl_{e}^{\mathrm{new}}=0$ and $\vu_e^{\new} = \vu_e$, otherwise, we orient the edge arbitrarily and set $-\vl_{e}^{\mathrm{new}} = \vu_e^{\mathrm{new}} = \vu_e$.

Now, we bound the parameters $L,R,r$ in \cref{thm:IPM}. Clearly,
$L=\|\vc^{\mathrm{new}}\|_{2}=O(Mm)$ and $R=\|\vu^{\mathrm{new}}-\vl^{\mathrm{new}}\|_{2}=O(Mm)$.
To bound $r$, we prove that there is an ``interior'' flow $\vf$
in the polytope $\mathcal{F}$. We construct this $\vf$ by $\vf=\sum_{e\in E}\vf^{(e)}$,
where $\vf^{(e)}$ is a circulation passing through edges $e$ and $(t,s)$
with flow value $1/(4m)$. All such circulations exist because of
the removal preprocessing. This $\vf$ satisfies the capacity constraints
because all capacities are at least $1$. This shows $r \geq \frac{1}{4m}$.

The RIPM in \cref{thm:IPM} runs the subroutine \textsc{Centering} twice.
In the first run, the constraint matrix is the incidence matrix of a new underlying graph,
constructed by making three copies of each edge in the original graph $G$. 
Since copying edges does not affect planarity, and our data structures allow for duplicate edges,
we use the implementation given in \textsc{CenteringImpl} (\cref{algo:IPM_impl}) for both runs.

By the guarantees of \cref{thm:SlackMaintain} and \cref{thm:FlowMaintain},
we correctly maintain $\vs$ and $\vf$ at every step in \textsc{CenteringImpl}, 
and the requirements on $\of$ and $\os$ for the RIPM are satisfied.
Hence, \cref{thm:IPM} shows that we can find a circulation $\vf$
such that $(\vc^{\mathrm{new}})^{\top}\vf\leq\mathrm{OPT}-\frac{1}{2}$
by setting $\epsilon=\frac{1}{CM^{2}m^{2}}$ for some large constant
$C$ in \cref{alg:IPM_centering}. 
Note that $\vf$, when restricted to the original graph, is almost
a flow routing the required demand with flow value off by at most $\frac{1}{2nM}$. This
is because sending extra $k$ units of fractional flow from $s$ to
$t$ gives extra negative cost $\leq-knM$. Now we can round $\vf$
to an integral flow $\vf^{\mathrm{int}}$ with same or better flow
value using no more than $\widetilde{O}(m)$ time \cite{kang2015flow}.
Since $\vf^{\mathrm{int}}$ is integral with flow value at least the total demand
minus $\frac{1}{2}$, $\vf^{\mathrm{int}}$ routes the demand completely.
Again, since $\vf^{\mathrm{int}}$ is integral with cost at most $\mathrm{OPT}-\frac{1}{2}$,
$\vf^{\mathrm{int}}$ must have the minimum cost.

Finally, we bound the runtime of one call to \textsc{CenteringImpl}.
We initialize the data structures for flow and slack by \textsc{Initialize}.
Here, the data structures are given the first IPM step direction $\vv$ for preprocessing; the actual step is taken in the first iteration of the main while-loop.
At each step of \textsc{CenteringImpl}, we perform the implicit update of $\vf$ and $\vs$ using $\textsc{Move}$;  we update $\mw$ in the data structures using $\textsc{Reweight}$; and we construct the explicit approximations 
$\of$ and $\os$ using $\textsc{Approximate}$; 
each in the respective flow and slack data structures. 
We return the true $(\vf,\vs)$ by $\textsc{Exact}$. 
The total cost of \textsc{CenteringImpl} is dominated by $\textsc{Move}$, $\textsc{Reweight}$, and $\textsc{Approximate}$. 

Since we call \textsc{Move}, \textsc{Reweight} and \textsc{Approximate} in order in each step and the runtime for \textsc{Move}, \textsc{Reweight} are both dominated by the runtime for \textsc{Approximate}, it suffices to bound the runtime for \textsc{Approximate} only.  
\cref{thm:IPM} guarantees that there are $T=O(\sqrt{m}\log n\log(nM))$ total \textsc{Approximate} calls.
\cref{lem:cor_change} shows that at the $k$-th call, the number of coordinates changed in $\vw$ and $\vv$ is bounded by $K \defeq O(2^{2\ell_{k-1}} \log^2 m+2^{2\ell_{k-2}} \log^2 m)$, where $\ell_k$ is the largest integer $\ell$ with $k \equiv 0 \mod 2^{\ell}$, or equivalently, the number of trailing zeros in the binary representation of $k$.
\cref{thm:IPM} further guarantees we can apply \cref{thm:SlackMaintain} and \cref{thm:FlowMaintain} with parameter $\beta = O(1/\log m)$,
which in turn shows the amortized time for the $k$-th call is
\[
	\otilde(\epssc^{-2} \sqrt{m(K + N_{k-2^{\ell_k}})}).
\]
where $N_{k} \defeq 2^{2\ell_k} (\beta/\alpha)^2 \log^2 m = O(2^{2\ell_k} \log^2 m)$,
where $\alpha = O(1/\log m)$ and $\epsilon_{\mproj}= O(1/\log m)$ are defined in \textsc{CenteringImpl}.

Observe that $K + N_{k-2^{\ell_k}} = O(N_{k-2^{\ell_k}})$. Now, summing over all $T$ calls, the total time is
\begin{align*}
	O(\sqrt{m} \log m)  \sum_{k=1}^T \sqrt{N_{k-2^{\ell_k}}} &= 
	O(\sqrt{m} \log^2 m)  \sum_{k=1}^T 2^{\ell_{(k - 2^{\ell_k})}} \\
	&= O(\sqrt{m} \log^2 m) \sum_{k'=1}^T  2^{\ell_{k'}}\sum_{k=1}^{T}[k-2^{\ell_k}=k'],
\intertext{
where we use $[\cdot]$ for the indicator function, i.e., $[k-2^{\ell_k}=k']=1$ if $k-2^{\ell_k}=k'$ is true and $0$ otherwise. As there are only $\log T$ different powers of $2$ in $[1, T]$, the count  $\sum_{1\le k\le T}[k-2^{\ell_k}=k']$ is bounded by $O(\log T)$ for any $k'\in \{1,\dots, T\}$. Then the above expression is
}
	&= O(\sqrt{m} \log^2 m \log T) \sum_{k'=1}^T  2^{\ell_{k'}}.
\intertext{
Since $\ell_k$ is the number of trailing zeros on $k$, 
it can be at most $\log T$ for $k \leq T$. We again rearrange the summation by possible values of $\ell_{k'}$, 
and note that there are at most $T/2^{i + 1}$ numbers between 1 and $T$ with $i$ trailing zeros, so}
 \sum_{k'=1}^T  2^{\ell_{k'}} &= \sum_{i=0}^{\log T} 2^{i} \cdot T/2^{i+1} = O(T \log T).
\end{align*}

So the overall runtime is $O(\sqrt{m} T \log m \log^2 T)$. Combined with \cref{thm:IPM}'s guarantee of $T=O(\sqrt{m}\log n\log(nM))$, we conclude the overall runtime is $\otilde (m \log M)$.
\end{proof}


\section{Preliminaries}

\emph{We assume all matrices and vectors in an expression have matching dimensions.} That is, we will trivially pad matrices and vectors with zeros when necessary. This abuse of notation is unfortunately unavoidable as we will be considering lots of submatrices and subvectors.

\paragraph{General Notations.}
An event holds with high probability if it holds with probability at least $1-n^c$ for arbitrarily large constant $c$. The choice of $c$ affects guarantees by constant factors. 


We use boldface lowercase variables to denote vectors, and boldface uppercase variables to denote matrices.
We use $\|\vv\|_2$ to denote the 2-norm of vector $\vv$ and $\|\vv\|_{\mm}$ to denote $\vv^\top \mm\vv$.
For any vector $\vv$ and scalar $x$, we define $\vv+x$ to be the vector obtained by adding $x$ to each coordinate of $\vv$
and similarly $\vv-x$ to be the vector obtained by subtracting $x$ from each coordinate of $\vv$. 
We use $\vzero$ for all-zero vectors and matrices where dimensions are determined by context. We use $\vone_{A}$ for the vector with value $1$ on coordinates in $A$ and $0$ everywhere else. We use $\mi$ for the identity matrix and $\mi_{S}$ for the identity matrix in $\mathbb{R}^{S \times S}$. 
For any vector $\vx \in \mathbb{R}^{S}$,
$\vx|_{C}$ denotes the sub-vector of $\vx$ supported on $C\subseteq S$; 
\emph{more specifically, $\vx|_C \in \R^S$, where $\vx_i = 0$ for all $i \notin C$.}

For any matrix $\mm \in \mathbb{R}^{A\times B}$, 
we use the convention that $\mm_{C, D}$ denotes the sub-matrix of $\mm$ supported on $C\times D$ where $C\subseteq A$ and $D\subseteq B$. 
When $\mm$ is not symmetric and only one subscript is specified, as in $\mm_D$, this denotes the sub-matrix of $\mm$ supported on $A \times D$.
To keep notations simple, $\mm^{-1}$ will denote the inverse of $\mm$ if it is an invertible matrix and the Moore-Penrose pseudo-inverse otherwise.

For two positive semi-definite matrices $\ml_1$ and $\ml_2$, we write $\ml_1 \approx_t \ml_2$ if $e^{-t} \ml_1\preceq \ml_2 \preceq e^{t} \ml_1$, where $\ma\preceq \mb$ means $\mb-\ma$ is positive semi-definite. 
Similarly we define $\geq_t$ and $\leq_t$ for scalars, that is, $x\leq_t y$ if $e^{-t}x\le y\le e^t x$. 

\paragraph{Graphs and Trees.} 
We define \emph{modified planar graph} to mean a graph obtained from a planar graph by adding $2$ new vertices $s, t$ and 
any number of edges incident to the new vertices. We allow distinguishable parallel edges in our graphs. 
We assume the input graph is connected.

We use $n$ for the number of vertices and $m$ for the number of edges in the input graph. 
We will use $\vw$ for the vector of edge weights in a graph. We define $\mw$ as the diagonal matrix $\diag(\vw)$.

We define $\ml = \mb^\top \mw \mb$ be the Laplacian matrix associated with an undirected graph $G$ with non-negative edge weights $\mw$. 
We at times use a graph and its Laplacian interchangeably.
For a subgraph $H \subseteq G$, we use $\ml[H]$ to denote the weighted Laplacian on $H$, and $\mb[H]$ to denote the incidence matrix of $H$.

For a tree $\ct$, we write $H \in \ct$ to mean $H$ is a node in $\ct$.
We write $\ct_H$ to mean the complete subtree of $\ct$ rooted at $H$.
We say a node $A$ is an ancestor of $H$ if $H$ is in the subtree rooted at $A$, and $H \neq A$.

The \emph{level} of a node in a tree is defined so that leaf nodes have level 0, and the root has level $\eta$, where $\eta$ is the height of the tree. For interior nodes, the level is the length of the longest path from the node to a leaf. 
By this definition, note that the level of a node and its child can differ by more than 1.

For binary tree data structures, we assume there is constant time access to each node.

\paragraph{IPM data structures.} 

When we discuss the data structures in the context of the IPM, step 0 means the initialization step. For $k > 0$, step $k$ means the $k$-th iteration of the while-loop in \textsc{Centering} (\cref{alg:IPM_centering,algo:IPM_impl}); that is, it is the $k$-th time we update the current solutions.
For any vector or matrix $\vx$ used in the IPM, we use $\vx^{(k)}$ to denote the value of $\vx$ at the end of the $k$-th step.

In all procedures in these data structures, we assume inputs are given by the set of changed coordinates and their values, 
\emph{compared to the previous input}. 
Similarly, we output a vector by the set of changed coordinates and their values, compared to the previous output. 
This can be implemented by checking memory for changes.

We use \textsc{smallCaps} to denote function names and data structure classes, and \texttt{typewriterFont} to denote an instantiation of a data structure.

We say a data structure B \emph{extends} A in the object-oriented sense. Inside data structure B, we directly access functions and variables of A when the context is clear, or use the keyword \texttt{super}.

In the data structure where we write $\ml^{-1} \vx$ for some Laplacian $\ml$ and vector $\vx$, we imply the use of an SDD-solver as a black box in nearly-linear time:
\begin{theorem}[\cite{spielman2004nearly, JambulapatiS21}]
	\label{thm:laplacianSolver}
	There is a randomized algorithm which is an $\eps$-approximate Laplacian system solver for the any input $n$-vertex $m$-edge graph
	and $\eps\in (0,1)$ and has the following runtime
	$O(m\poly(\log \log n)\log(1/\eps))$.
\end{theorem}

\section{Nested dissection and approximate Schur complements} \label{sec:apxsc}

This section lays the foundation for a recursive decomposition of the input graph.
Our goal is to set up the machinery necessary for approximating $\mproj_{\vw} \defeq \mw^{1/2} \mb (\mb^\top \mw \mb)^{-1} \mb^{\top} \mw^{1/2}$ as needed in the robust IPM. 
In particular, we are interested in the weighted Laplacian matrix $\ml \defeq \mb^\top \mw \mb$.

We begin with a discussion of nested dissection and the associated Schur complements.

\subsection{Cholesky decomposition and Schur complement}\label{sec:schur-complement-graphs}

Let $G$ be a weighted graph. Consider the partition of vertices in $G$ into two subsets $C$ and $F = V(G) \setminus C$ called \emph{boundary} and \emph{interior} vertices. This partitions $\ml$ into four blocks:
\begin{align*}
\ml = \left[ \begin{array}{cc}
\ml_{F,F} & \ml_{F,C} \\
\ml_{C,F} & \ml_{C,C}
\end{array} \right].
\end{align*}

\begin{definition}[Block Cholesky decomposition]
	The \emph{block Cholesky decomposition} of a symmetric $\ml$ with blocks indexed by $F$ and $C$ defined as  above is:
\begin{equation}\label{eq:basic_chol}
	\ml
	 = 
	\left[
	\begin{array}{cc}
		\mi
		& \mzero \\
		\ml_{C,F} (\ml_{F,F})^{-1}
		& \mi
	\end{array}
	\right] 
	\left[
	\begin{array}{cc}
		\ml_{F,F}
		& \mzero \\
		\mzero
		& \sc(\ml, C)
	\end{array}
	\right] 
	\left[
	\begin{array}{cc}
		\mi
		& (\ml_{F,F})^{-1}\ml_{F,C} \\
		\mzero & \mi
	\end{array}\right].
\end{equation}
\end{definition}
The middle matrix in the decomposition is a block-diagonal matrix with blocks indexed by $F$ and $C$, with the lower-right block being:

\begin{definition}[Schur complement]
The \emph{Schur complement} $\sc(\ml,C)$ of $\ml$ onto $C$ is the Laplacian matrix resulting from a partial symmetric Gaussian elimination on $\ml$. Formally, 
\[
	\sc(\ml,C) = \ml_{C,C} - \ml_{C,F}\ml_{F,F}^{-1}\ml_{F,C}. 
\]
\end{definition}      

It is known that $\sc(\ml,C)$ is the Laplacian of another graph with vertex set $C$. 
We further use the convention that if $H$
is a subgraph of $G$ and $V(H) \subset C$, then $\sc(H, C)$ simply
means $\sc(H, C \cap V(H))$.
Graph theoretically, the Schur complement has the following interpretation:
\begin{lemma}
  Let $V(G) = \{v_1, \dots, v_n\}$. 
  Let $C = V(G) - v_1$. Let $\vw_{ij}$ denote the weight of edge $v_i v_j$. Then
  \[
    \sc(\ml, C) = G[C] + H,
  \]
  where $G[C]$ is the subgraph of $G$ induced on the vertex set $S$,
  and $H$ is the graph on $S$ with edges $v_i v_j$ where
  $i, j \in N(v_1)$, and $\vw_{ij} = \vw_{1i} \vw_{1j}/\vw_1$, where $\vw_1$ is
  the total weight of edges incident to $v_1$ in $G$. 
  Note that on the right hand side, we use a graph to mean its Laplacian.
  \qed
\end{lemma}

Taking Schur complement is an associative operation. Furthermore, it commutes with edge deletion, and more generally, edge weight deletion. Finally, for our purposes, it can be decomposed under certain special circumstances.

\begin{lemma}
  \label{lem:sc_transitivity}
	If $X \subseteq Y \subseteq V(G)$, then
$	\sc(\sc (\ml, Y), X) = \sc(\ml, X).$ \qed
\end{lemma}

\begin{lemma}  \label{lem:sc-delete-edges}
	Let $\vw_e$ denote the weight of edge $e$ in $G$.
	Suppose $C \subseteq V(G)$, and 
	$H$ is a subgraph of $G$ on the vertex set $C$ with edge weights $\vw'_e \leq \vw_e$ for all edges in $G[C]$.
	Let $\ml'$ denote the Laplacian of $H$. 
	Then,
			$\sc(\ml - \ml', C) = \sc(\ml, C) - \ml'.$ \qed
\end{lemma}

\begin{lemma}\label{lem:sc-decomposition}
  Let $\ml$ be the Laplacian of graph $G$ with the decomposition 
  $\ml = \ml_1 + \ml_2$, where $\ml_1$ is a Laplacian supported on the vertex set $V_1$ and $\ml_2$ on $V_2$. 
  Furthermore, suppose  $V_1 \cap V_2 \subseteq C$ for some vertex set $C \subseteq V(G)$.
  Then 
	\[
		\sc(\ml, C) = \sc(\ml_1, C \cap V_1) + \sc(\ml_2, C \cap V_2).
	\]
\end{lemma}	
\begin{proof}
	We have
	\begin{align*}
		\sc(\ml, C) &= \sc(\ml_1 + \ml_2, C) \\
		&= \sc(\sc(\ml_1 + \ml_2, C \cup V_2), C) \\
		&= \sc(\sc(\ml_1, C \cup V_2) + \ml_2, C)  \tag{by \cref{lem:sc-delete-edges}} \\
		&= \sc(\sc(\ml_1, C) + \ml_2, C)  \tag{since $(C \cup V_2) \cap V_1 \subseteq C$}\\
		&= \sc(\ml_1, C) + \sc(\ml_2, C)  \tag{by \cref{lem:sc-delete-edges}},\\
		&= \sc(\ml_1, C \cap V_1) + \sc(\ml_2, C \cap V_2) \tag{since $\ml_i$ is supported on $V_i$ for $i=1,2$} 
	\end{align*}
	as desired.
\end{proof}

\subsection{Separator tree} \label{subsec:construct_tree}

In the overview, we briefly gave the intuition for a 2-level partition of the input graph; here we extend it to
a recursive partitioning scheme with $O(\log n)$-levels. We begin with the formal definitions.

\begin{restatable}[Separable graph]{definition}{defSeparableGraph}
 \label{defn:separable-graph}
A graph $G=(V, E)$ is $\alpha$-separable if there exists two constants $c>0$ and $b\in (0, 1)$ such that every nonempty subgraph $H=(V(H)\subseteq V, E(H)\subseteq E)$ with $|E(H)|\ge 2$ of $G$ can be partitioned into $\region_1$ and $\region_2$ such that 
\begin{itemize}
\item $E(\region_1)\cup E(\region_2) = E(H)$, $E(\region_1)\cap E(\region_2) = \emptyset$,
\item $|V(\region_1)\cap V(\region_2)|\le c\lceil |E(H)|^{\alpha} \rceil$,
\item $|E(\region_i)|\le b|E(H)|$, for $i=1,2$.
\end{itemize}
We call 
$S(H) \defeq V(\region_1)\cap V(\region_2)$ the \textit{balanced vertex separator} of $H$. 
\end{restatable}

It is known that any planar graph is $1/2$-separable. 

\begin{remark}
	As we discussed in \cref{subsec:overview_proof}, our LP formulation for the IPM uses a
	\emph{modified planar graph} which is the original planar graph with two additional vertices and $O(n)$
	additional edges incident to them.
	By adding two vertices and edges incident to them to a planar graph, the modified graph is also $1/2$-separable with the constant $c$ in \cref{defn:separable-graph} increased by $2$.
\end{remark}


We apply nested dissection recursively to each region using balanced vertex separators, until the regions are of constant size. 
The resulting hierarchical structure can be represented by a tree $\ct$, which is known as the \emph{separator tree} of $G$:

\begin{definition} [Separator tree $\ct$] \label{defn:separator-tree}
Let $G$ be a modified planar graph. 
A separator tree $\ct$ is a binary tree whose nodes represent subgraphs of $G$ such that the children of each node $H$ form a balanced partition of $H$.

Formally, each node of $\ct$ is a \emph{region} (edge-induced subgraph) $\region$ of $G$; we denote this by $\region \in \ct$. 
At a node $\region$, we store subsets of vertices $\bdry{\region}, \sep{\region}, \elim{\region} \subseteq V(\region)$, 
where $\bdry{\region}$ is the set of \emph{boundary vertices} that are incident to vertices outside $\region$ in $G$;
$\sep{\region}$ is the balanced vertex separator of $\region$;
and $\elim{\region}$ is the set of \emph{eliminated vertices} at $\region$. 
Concretely, the nodes and associated vertex sets are defined recursively in a top-down way as follows: 

	\begin{enumerate}
		\item The root of $\ct$ is the node $\region = G$, with $\bdry{\region} = \emptyset$ and $\elim{\region} = \sep{\region}$.
		\item A non-leaf node $\region \in \ct$ has exactly two children $D_1, D_2 \in \ct$ that form an edge-disjoint partition of $\region$ in \cref{defn:separable-graph}, and their vertex sets intersect on the balanced separator $\sep{\region}$ of $\region$. $D_1$ and $D_2$ does not have  any isolated vertex.
		Define $\bdry{D_1} = (\bdry{\region} \cup \sep{\region}) \cap V(D_1)$, and similarly $\bdry{D_2} = (\bdry{\region} \cup \sep{\region}) \cap V(D_2)$.
		Define $\elim{\region} = \sep{\region} \setminus \bdry{\region}$.\label{property: boundary}
		
		\item If a region $\region$ contains a constant number of edges, then we stop the recursion and $\region$ becomes a leaf node. Further, we define $\sep{\region} = \emptyset$ and $\elim{\region} = V(\region) \setminus \bdry{\region}$. Note that by construction, each edge of $G$ is contained in a unique leaf node. 
	\end{enumerate}

Let $\eta(H)$ denote the height of node $H$ which is defined as the maximum number of edges on a tree path from $H$ to one of its descendants. $\eta(H)=0$ if $H$ is a leaf. Note that the height difference between a parent and child node could be greater than one. Let $\eta$ denote the height of $\ct$ which is defined as the maximum height of nodes in $\ct$. We say $H$ is at \emph{level} $i$ if $\eta(H)=i$. 
\end{definition}

\begin{observation}\label{lem:F_partitions}
	Using the above definition, $\{\elim{\region} : \region \in \ct\}$ partitions the vertex set $V(G)$. 
\end{observation}

\begin{observation} \label{lem:parent-child-boundry-relation}
Suppose $H$ is a node in $\ct$ with children $D_1$ and $D_2$. We have $\partial D_1 \cup \partial D_2 = \partial H \cup F_H$.
\end{observation}

\begin{observation} \label{lem:bdry-containment}
	Suppose $H$ is a node in $\ct$. Then $\bdry{H} \subseteq \cup_{\text{ancestor $A$ of $H$}} F_A$.
\end{observation}

Fakcharoenphol and Rao~\cite{fakcharoenphol2006planar} gave an algorithm that computes the separator tree for any planar graph. 

\begin{theorem}[Separator tree construction \cite{fakcharoenphol2006planar}] \label{thm:separator_tree_construction}
	Given a planar graph $G$, 
	there is an algorithm that computes a separator tree $\ct$ of $G$ of height $\eta = O(\log n)$ in $O(n\log n)$ time. 
\end{theorem}

For computing the separator tree $\ct$ of a modified planar graph, we may apply their method to the original planar graph to get the separator $\ct'$, and add the two new vertices $s,t$ to $F_G$ at the root node $G$,
and to the boundary sets $\bdry{H}$ at every non-root node $H$. 
The additional edges incident to $s,t$ can be recursively partitioned from a node to its children, which increases the height of $\ct$ by $O(\log n)$. 
Thus, we have the following corollary:

\begin{corollary}[Separator tree construction for modified planar graph] \label{thm:separator_tree_construction_modified}
	Given a \emph{modified planar graph} $G$, 
	there is an algorithm that computes a separator tree $\ct$ of $G$ of height $\eta = O(\log n)$ in $O(n\log n)$ time. 
\end{corollary}

To discuss the structures in the separator tree, we define the following terms:

\begin{definition}
	Let $\ct(i)$ be the subset of nodes in $\ct$ at level $i$.
	For a node $\region$, let $\ct_\region$ be the subtree of $\ct$ rooted at $\region$. Let $\pathT{\region}$ be the set nodes on the path from $\region$ to the root of $\ct$, including $\region$.
	Given a set of nodes $\collN = \{\region : \region \in \ct\}$, define
	\[ \pathT{\collN} := \bigcup_{\region \in \collN} \pathT{\region}.\]
	
	Finally, we partition these nodes by their level in $\ct$, and use $\pathT{\collN,i}$ to denote all the nodes in $\pathT{\collN}$ at level $i$ in $\ct$.
\end{definition}

Fakcharoenphol and Rao~\cite[Section 3.5]{fakcharoenphol2006planar} showed
that for a set $\collN$ of $K$ nodes in $T$, the total number of boundary vertices from the nodes in
$\pathT{\collN}$ is $O(\sqrt{mK})$. However, their claim is not stated
as a result we can cite here. 
We provide a simple,
self-contained proof in \cref{sec:appendix} of a slightly weaker bound that in addition
requires bounding the number of separator vertices.

\begin{restatable}{lemma}{planarBoundChangeCost}
	\label{lem:planarBoundChangeCost}
	Let $G$ be a modified planar graph with separator tree $\ct$. Let $\mathcal{\region}$ be a set of $K$ nodes in $\ct$. Then
	\begin{align*}
		\sum_{\region \in \pathT{\collN}}| \bdry{\region}| +|\elim{\region}| \leq \otilde(\sqrt{mK} ). 
	\end{align*}
\end{restatable}


\subsection{Approximating $\ml^{-1}$ using the separator tree}

For a height-$\eta$ separator tree, we generalize the sets $C$ and $F$ from the block Cholesky decomposition (\cref{eq:basic_chol}) to a sequence of sets $C_0, \ldots, C_{\eta}$, and $F_0, \dots, F_\eta$ based on $\ct$.

\begin{definition}[$C_i, F_i$] \label{defn:tree-sets}
	Let $\ct$ be the separator tree from \cref{thm:separator_tree_construction_modified}.
	For all $0\le i\le \eta$, we define $F_i=\bigcup_{H \in \ct(i)} \elim{H}$ to be the vertices eliminated at level $i$.
	For all $0\le i\le \eta$, we define $C_i=\bigcup_{H \in \ct(i)} \bdry{H}$ to be the vertices remaining after eliminating vertices in $F_i$. We define $C_{-1}$ to be $V(G)$.
\end{definition}

By \cref{lem:F_partitions}, $F_i$ is the disjoint union of $\elim{\region}$ over all nodes $H$ at level $i$ in the separator tree. $F_0,\ldots, F_\eta$ partitions $V(G)$.  
By the definition of $\bdry{\region}$ and $\elim{\region}$, we know $F_i = C_{i-1} \setminus C_i$ for all $0\le i\le \eta$.
It follows that $V(G) = C_{-1} \supset C_0 \supset \dots \supset C_{\eta-1} \supset C_{\eta} = \emptyset$ and $C_i = \cup_{j > i} F_j$.

Now, the decomposition from \cref{eq:basic_chol} can be extended and inverted as follows:
\begin{equation}\label{eq:Linv}
	\ml^{-1} = \mmu^{(0)\top} \cdots \mmu^{(\eta-1)\top}
	\left[
	\begin{array}{ccc}
		{\sc(\ml, C_{-1})_{F_0, F_0}}^{-1} & \mzero & \mzero\\
		\mzero & \ddots & \mzero\\
		\mzero & \mzero & {\sc(\ml, C_{\eta-1})_{F_\eta, F_\eta}}^{-1}
	\end{array}
	\right]
	\mmu^{(\eta-1)} \cdots \mmu^{(0)},
\end{equation}
where the $\mmu^{(i)}$'s are upper triangular matrices with 
$$\mmu^{(i)} = \mi - \sc(\ml, C_{i-1})_{C_i, F_i} \left(\sc(\ml, C_{i-1})_{F_i,F_i} \right)^{-1},$$ 
where we assume all matrices are $n \times n$ by padding zeroes when required.
To efficiently compute parts of $\ml^{-1}$, we use approximate Schur complements instead of exact ones in \cref{eq:Linv}.

\begin{definition}[Approximate Schur Complement] \label{def:approxSchur} 
	Let $G$ be a weighted graph with Laplacian $\ml$, and let $C$ be a set of boundary vertices in $G$.
	We say that a Laplacian matrix $\tsc(\ml,C) \in \R^{C \times C}$ is an $\eps$-\emph{approximate Schur complement} of $\ml$ onto $C$ if $\tsc(\ml,C) \approx_{\eps} \sc(\ml,C)$, where we use $\approx_{\eps}$ to mean an $e^{\eps}$-spectral approximation.
\end{definition}

\begin{definition}[$\ml^{(H)}$] \label{defn:L^H}
	Let $\epssc > 0$. For each $H \in \ct$, let $\ml^{(H)}$ be a Laplacian on the vertex set $F_H \cup \partial H$ such that
	\[
		\ml^{(H)} \approx_{\epssc} \sc(\ml[H], \bdry{\region} \cup \elim{\region}).
	\]
\end{definition}
We show how to compute and maintain $\ml^{(H)}$ in the next subsection. 

Here, we define the necessary approximate matrices and show how to approximate $\ml^{-1}$.
\begin{definition}[$\mpi^{(i)}, \mx^{(H)}, \widetilde{\mga}$]
	
	To approximate $\mmu^{(i)}$, we define	
	\begin{equation} \label{def:approx_mproj}
		\mpi^{(i)} = \mi - \sum_{H \in \ct(i)} \mx^{(H)},
	\end{equation}
	where
	\begin{equation}\label{def:mx^(H)}
		\mx^{(H)} = \ml^{(H)}_{\bdry{\region}, F_H} \left( \ml^{(H)}_{F_H, F_H}\right)^{-1}
	\end{equation}
	for each $H \in \ct$.

	To approximate the block diagonal matrix in \cref{eq:Linv}, we define
	\[
	\widetilde \mga = 
	\left[
	\begin{array}{cccc}
		\sum_{H \in \ct(0)} \left(\ml^{(H)}_{F_H, F_H}\right)^{-1} & \mzero & \mzero\\
		\mzero &  \ddots & \mzero\\
		\mzero & \mzero & \sum_{H \in \ct(\eta)} \left(\ml^{(H)}_{F_H, F_H}\right)^{-1}
	\end{array}
	\right].
	\]
\end{definition}

\begin{theorem}[$\ml^{-1}$ approximation]\label{thm:L-inv-approx}
	Suppose for each $H \in \ct$, we have a Laplacian $\ml^{(H)}$ satisfying
	\[
	\ml^{(\region)} \approx_{\epssc} \sc(\ml[\region], \bdry{\region}\cup \elim{\region}).
	\]
	Then, we have
	\begin{equation} \label{eq:Linv_approx}
		\ml^{-1} \approx_{\eta \epssc}
		\mpi^{(0)\top}\cdots\mpi^{(\eta-1)\top} \widetilde{\mga}
		\mpi^{(\eta-1)}\cdots\mpi^{(0)}.
	\end{equation} 
\end{theorem}

\begin{proof}
	Let $C_i, F_i$ be defined for each $i$ according to \cref{defn:tree-sets}.
	Let $\ml^{(i)} \defeq \sum_{H \in \ct(i)} \ml^{(H)}$.
	
	Note that $\ml^{(i)}_{F_i,F_i} \defeq \sum_{H \in \ct(i)} \ml^{(H)}_{F_H, F_H}$ 
	is a block-diagonal matrix with blocks indexed by $H \in \ct(i)$, 
	since $F_i$ is a disjoint union over $F_H$ for $H \in \ct(i)$, and only $\ml^{(H)}$ is supported on $F_H$.
	Hence, ${\ml^{(i)}_{F_i, F_i}}^{-1} = \sum_{H \in \ct(i)} \left(\ml^{(H)}_{F_H, F_H}\right)^{-1}$.
	
	Recall that the regions in $\ct(i)$ partition the graph $G$. Furthermore, the intersection of $H, H' \in \ct(i)$ is on their boundary, which is contained in $C_i \subseteq C_{i-1}$. Thus, we apply \cref{lem:sc-decomposition} to get
	\begin{equation}
		\label{eq:sum-L^(H)}
	\begin{aligned} 
		\sc(\ml, C_{i-1}) &= \sum_{H \in \ct(i)} \sc(\ml[H], C_{i-1} \cap V(H)) \\
		&\approx_{\epssc} \sum_{H \in \ct(i)} \tsc(\ml[H], \partial H \cup F_H)
		= \sum_{H \in \ct(i)} \ml^{(H)} = \ml^{(i)}.
	\end{aligned}
	\end{equation}
	
	Now, we prove inductively that
	\begin{equation}\label{eq:Linv-recurse}
		\ml^{-1} \approx_{i \epssc} 
		\mpi^{(0)\top}\cdots\mpi^{(i-1)\top}
		\left[
		\begin{array}{cccc}
			\left(\ml^{(0)}_{F_0, F_0}\right)^{-1} & \mzero & \mzero & \mzero\\
			\mzero & \ddots & \mzero & \mzero\\
			\mzero & \mzero & \left(\ml^{(i-1)}_{F_{i-1}, F_{i-1}} \right)^{-1} & \mzero\\
			\mzero & \mzero & \mzero & \left(\ml^{(i)}\right)^{-1} 
		\end{array}
		\right]
		\mpi^{(i-1)}\cdots\mpi^{(0)},
	\end{equation}

	When $i = 0$, we have the approximation trivially as $\ml^{(0)} = \ml$. 
	
	For general $i$, we factor $\ml^{(i)}$ in \cref{eq:Linv-recurse} recursively using Cholesky decomposition. 
	$\ml^{(i)}$ is supported on $C_{i-1}$, and we can partition $C_{i-1} = F_i \cup C_i$. Then,
	\begin{equation} \label{eq:L^(i)}
		\ml^{(i)} = 
		\left[
		\begin{array}{cc}
			\mi
			& \mzero \\
			\ml^{(i)}_{C_i,F_i} (\ml^{(i)}_{F_i,F_i})^{-1}
			& \mi
		\end{array}
		\right] 
		\left[
		\begin{array}{cc}
			\ml^{(i)}_{F_i,F_i}
			& \mzero \\
			\mzero
			& \sc(\ml^{(i)}, C_i)
		\end{array}
		\right] 
		\left[
		\begin{array}{cc}
			\mi
			& (\ml^{(i)}_{F_i,F_i})^{-1}\ml^{(i)}_{F_i,C_i} \\
			\mzero & \mi
		\end{array}\right].
	\end{equation}

	For the Schur complement term in the factorization, we have
	\begin{align*}
		\sc(\ml^{(i)}, C_i) &\approx_{i \epssc} \sc(\sc(\ml, C_{i-1}), C_i) \tag{by \cref{eq:sum-L^(H)}} \\
		&= \sc(\ml, C_i) \tag{by transitivity of Schur complements} \\
		&\approx_{\epssc} \ml^{(i)}. \tag{by \cref{eq:sum-L^(H)}}
	\end{align*}
	So we can use $\ml^{(i)}$ in place of the Schur complement term, and the equality becomes an approximation with factor $(i+1) \epssc$.
	Furthermore, in \cref{eq:L^(i)}, we can rewrite
	\[
	\ml^{(i)}_{C_i,F_i} = \sum_{H \in \ct(i)} \ml^{(H)}_{C_i,F_i} = \sum_{H \in \ct(i)} \ml^{(H)}_{\partial H,F_H}. 
	\]
	Plugging the inverse of \cref{eq:L^(i)} into \cref{eq:Linv-recurse}, we get the correct recursive approximation.
	
	Finally, we note that at the $\eta$-th level, $\ml^{(\eta)}_{F_\eta, F_\eta} = \ml^{(\eta)}$ since $C_\eta = \emptyset$. 
	So we have the overall expression.
\end{proof}

\subsection{Recursive Schur complements on separator tree}

In this section, we prove \cref{thm:apxsc} which maintains approximate Schur complements onto the boundary vertices of each node $\region$ in $\ct$. 

We use the following result as a black-box for computing sparse approximate Schur complements:

\begin{lemma}[$\textsc{ApproxSchur}$ procedure \cite{DurfeeKPRS17}]
	\label{lem:fastApproxSchur} 
	Let $\ml$ be the weighted Laplacian of a graph with $n$ vertices and $m$ edges, and
	let $C$ be a subset of boundary vertices of the graph.
	Let $\gamma = 1/n^3$ be the error tolerance.
	Given approximation parameter $\eps \in (0,1/2)$, 
	there is an algorithm \textsc{ApproxSchur}$(\ml, C, \eps)$ that 
	computes and outputs a $\eps$-approximate Schur complement $\tsc(\ml, C)$ that satisfies
	the following properties with probability at least $1-\gamma$:
	\begin{enumerate}
		\item The graph corresponding to $\tsc(\ml, C)$ has $O(\eps^{-2} |C| \log (n/\gamma))$ edges. \label{approxSC: sparsity}
		\item The total running time is $O(m \log^3 (n/\gamma) + \eps^{-2} n \log^4(n/\gamma))$. \label{approxSC: runtime}
	\end{enumerate}
\end{lemma}

\begin{algorithm}
	\caption{Data structure to maintain dynamic approximate Schur complements} \label{alg:dynamicSC}
	\begin{algorithmic}[1]
		\State \textbf{data structure} \textsc{DynamicSC}
		\State \textbf{private: member}
		\State \hspace{4mm} Graph $G$ with incidence matrix $\mb$
		\State \hspace{4mm} $\vw \in \R^m$, $\mw \in \R^{m \times m}$: Weight vector and diagonal weight matrix, used interchangeably
		\State \hspace{4mm} $\epssc > 0$: Overall approximation factor
		\State \hspace{4mm} $\epslevel > 0$: Fast Schur complement approximation factor
		\State \hspace{4mm} $\ct$: Separator tree of height $\eta$. Every node $H$ of $\ct$ stores:
		\State \hspace{8mm} $\elim{\region}$, $\bdry{\region}$: Interior and boundary vertices of region $H$
		\State \hspace{8mm} $\ml^{(H)} \in \R^{m \times m}$: Laplacian supported on $\elim{\region} \cup \bdry{\region}$
		\State \hspace{8mm} $\tsc(\ml^{(H)}, \bdry{\region}) \in \R^{m \times m}$: $\epslevel$-approximate Schur complement of $\ml^{(H)}$
		\State
		
		\Procedure{\textsc{Initialize}} {$G$, $\vw \in \mathbb{R}^m$, $\epssc > 0$}
		\State $\mb \leftarrow$ incidence matrix of $G$
		\State $\ct \leftarrow$ separator tree of $G$ of height $\eta$ constructed by \cref{thm:separator_tree_construction}
		\State $\epslevel \leftarrow \epssc/(\eta+1)$
		\State $\vw \leftarrow \vw$
		\For{$i = 0,\ldots, \eta$}
		\For{each node $\region$ at level $i$ in $\ct$}
		\State \textsc{ApproxSchurNode}$(\region)$
		\EndFor
		\EndFor
		\EndProcedure
		\State
		
		\Procedure{\textsc{Reweight}}{$\vw^\new \in \R^{m}$}
		\State $\mathcal{H} \leftarrow$ set of leaf nodes in $\ct$ that contain each edge $e$ whose weight is updated
		\State $\vw \leftarrow \vw^\new$
		\State $\mathcal{P}_{\mathcal{T}}(\collN) \leftarrow$ set of all ancestor nodes of $\collN$ in $\ct$ and $\collN$
		\For{$i = 0,\ldots, \eta$}
		\For{each node $\region$ at level $i$ in $\mathcal{P}_{\ct}(\mathcal{H})$}
		\State \textsc{ApproxSchurNode}$(\region)$
		\EndFor
		\EndFor
		\EndProcedure
		\State
		
		\Procedure{\textsc{ApproxSchurNode}} {$\region \in \ct$}
		\If{$\region$ is a leaf node}
		\State \LeftComment $\mb[H]$ is the incidence matrix for the induced subgraph $H$ with edge set $E(H)$
		\State $\ml^{(H)} \leftarrow (\mb[H])^\top \mw_{E(H)} \mb[H]$
		\State $\tsc(\ml^{(H)}, \bdry{\region}) \leftarrow \textsc{ApproxSchur}(\ml^{(H)}, \bdry{\region}, \epslevel)$  \Comment\cref{lem:fastApproxSchur}
		\Else
		\State Let $D_1, D_2$ be the children of $H$ 
		\State $\ml^{(H)} \leftarrow \tsc(\ml^{(D_1)}, \bdry{D_1}) + \tsc(\ml^{(D_2)}, \bdry{D_2})$
		\State $\tsc(\ml^{(H)}, \bdry{\region}) \leftarrow \textsc{ApproxSchur}(\ml^{(H)}, \bdry{\region}, \epslevel)$ \label{line:dynamicSCsparsify}
		\EndIf
		\EndProcedure
	\end{algorithmic}
\end{algorithm}

First, we prove the correctness and runtime of \textsc{ApproxSchurNode}$(\region)$.
We say \textsc{ApproxSchurNode}$(\region)$ runs correctly on a node $H$ at level $i$ in $\ct$, if at the end of the procedure, the following properties are satisfied:
\begin{itemize}
	\item $\ml^{(\region)}$ is the Laplacian of a graph on vertices $\bdry{\region} \cup \elim{\region}$
	with $\O(\epslevel^{-2}|\bdry{\region}\cup \elim{\region}|)$ edges,
	\item $\ml^{(\region)} \approx_{(i-1) \epslevel} \sc(\ml[\region], \bdry{\region}\cup \elim{\region})$,
	\item $\tsc(\ml^{(H)}, \bdry{\region})\approx_{i \epslevel} \sc(\ml[\region], \bdry{\region})$, and the graph is on $\bdry{\region}$ with $\O(\epslevel^{-2}|\bdry{\region}|)$ edges.
\end{itemize}

\begin{lemma}\label{lem:ApproxSchurNodeCorrectness}
	Suppose $\ml^{(D)}$ and $\tsc (\ml^{(D)}, \bdry{D})$ are computed correctly for all descendants $D$ of $H$,
	then \textsc{ApproxSchurNode}$(\region)$ runs correctly. 
\end{lemma}
\begin{proof}
	When $\region$ is a leaf, the proof is trivial. 
	$\ml^{(\region)}$ is set to the exact Laplacian matrix of the induced subgraph $\region$ of constant size. 
	$\tsc(\ml^{(H)}, \bdry{\region})$ $\epslevel$-approximates $\sc(\ml^{(\region)}, \bdry{\region})=\sc(\ml[\region], \bdry{\region})$ by \cref{lem:fastApproxSchur}.
	
	Otherwise, suppose $H$ is at level $i$ with children $D_1$ and $D_2$.
	By construction of the separator tree and \cref{lem:parent-child-boundry-relation}, we have $\partial D_1 \cup \partial D_2 = \partial H \cup F_H$. For each $j=1,2$,  we know inductively $\tsc(\ml^{(D_j)}, \bdry{D_j})$ has $\O(\epslevel^{-2}|\bdry{D_j}|)$ edges.
	Since we define $\ml^{(\region)}$ to be the sum, it has $\O(\epslevel^{-2}(|\bdry{D_1}|+|\bdry{D_2}|))=\O(\epslevel^{-2}|\bdry{\region}\cup \elim{\region}|)$ edges, and is supported on vertices $\partial H \cup F_H$, so we have the first correctness property.
	
	Inductively, we know $\tsc(\ml^{(D_j)}, \bdry{D_j})\approx_{(i-1) \epslevel} \sc(\ml[D_j], \bdry{D_j})$ for both $j=1,2$. 
	(The height of $D_j$ may or may not equal to $i-1$ but it is guaranteed to be no more than $i-1$.) Then
	\begin{align*}
		\ml^{(\region)} &= \tsc(\ml^{(D_1)}, \bdry{D_1}) + \tsc(\ml^{(D_2)}, \bdry{D_2}) \\
		&\approx_{(i-1) \epslevel} \sc(\ml[D_1], \bdry{D_1}) + \sc(\ml[D_2], \bdry{D_2}) \\
		&= \sc(\ml[D_1], (\bdry{\region} \cup F_H) \cap V(D_1)) + \sc(\ml[D_2],  (\bdry{\region} \cup F_H) \cap V(D_2)) 
		\tag{by construction of the separator tree, $\partial D_j = (\partial H \cup F_H) \cap V(D_j)$ for $j = 1,2$}\\
		&= \sc(\ml[H], \partial H \cup F_H), \tag{by \cref{lem:sc-decomposition}}
	\end{align*}
	so we have the second correctness property.
	
	\cref{line:dynamicSCsparsify} returns $\tsc(\ml^{(H)}, \bdry{\region}) $ with $\O(\epslevel^{-2}|\bdry{\region}|)$ edges by \cref{lem:fastApproxSchur}. Also,
	\begin{align*}
		\tsc(\ml^{(H)}, \bdry{\region}) &\approx_{\epslevel} \sc(\ml^{(\region)}, \bdry{\region}) \\
		&\approx_{(i-1)\epslevel} \sc(\sc(\ml[H], \partial H \cup F_H), \bdry{\region}) \\
		&= \sc(\ml[H], \partial H), \tag{by \cref{lem:sc_transitivity}}
	\end{align*}
	giving us the third correctness property.
\end{proof}

\begin{lemma}
	The runtime of \textsc{ApproxSchurNode}$(\region)$ is $\O(\epslevel^{-2}|\bdry{\region}\cup \elim{\region}|)$.
\end{lemma}
\begin{proof}
	When $H$ is a leaf node, computing $\ml^{(H)} = \ml[H]$ takes time proportional to $|H| = \partial H \cup F_H$. 
	Computing $\tsc(\ml^{(H)}, \partial H)$ takes $\O (\epslevel^{-2} |H|)$ time by \cref{lem:fastApproxSchur}.
	
	Otherwise, when $H$ has children $D_1, D_2$, computing $\ml^{(H)}$ requires accessing $\tsc(\ml^{(D_j)}, \partial D_j)$ for $j=1,2$ and summing them together, in time $\O(|\partial D_1| + |\partial D_2|) = \O(|\partial H \cup F_H|)$. 
	Then, computing $\tsc (\ml^{(H)}, \partial H)$ take $\O(\epslevel^{-2}|\partial H \cup F_H|)$ by \cref{lem:fastApproxSchur}.
\end{proof}

Next, we prove the overall data structure correctness and runtime:
\apxscMaintain*

\begin{proof}[Proof of \cref{thm:apxsc}]
	Because we set $\epslevel \leftarrow \epssc/(\eta+1)$ in \textsc{Initialize}, combined with \cref{lem:ApproxSchurNodeCorrectness}, we conclude that for each $H \in \ct$, 
	$$\ml^{(\region)} \approx_{\epssc}\sc(\ml[H], \bdry{H}\cup \elim{H})$$
	and
	$$\tsc(\ml^{(H)}, \bdry{\region})\approx_{\epssc} \sc(\ml[\region], \bdry{\region}).$$
	
	
	We next prove the correctness and runtime of \textsc{Initialize}. 
	Computing the separator tree costs $O(n\log n)$ time by \cref{thm:separator_tree_construction}. 
	Because \textsc{ApproxSchurNode}$(\region)$ is called in increasing order of level of $\region$, each \textsc{ApproxSchurNode}$(\region)$ runs correctly and stores the initial value of $\ml^{(\region)}$ by \cref{lem:ApproxSchurNodeCorrectness}. 
	The runtime of \textsc{Initialize} is bounded by running \textsc{ApproxSchurNode} on each node, i.e:
	\[
	\O(\epslevel^{-2}\sum_{\region \in \ct} |\bdry{\region}\cup \elim{\region}|) =\O(\epslevel^{-2}m)=\O(\epssc^{-2}m).
	\]
	Where we bound the sum using \cref{lem:planarBoundChangeCost} with $K = O(m)$, since $\ct$ has $O(m)$ nodes in total.
	
	The proof for \textsc{Reweight} is similar to \textsc{Initialize}. 
	Let $K$ be the number of coordinates changed in $\vw$. 
	Then $\mathcal{P}_{\mathcal{T}}(\collN)$ contains all the regions with an edge with weight update.
	For each node $\region$ not in $\mathcal{P}_{\mathcal{T}}(\collN)$, no edge in $\region$ has a modified weight,
	and in this case, we do not need to update $\ml^{(H)}$. 
	For the nodes that do require updates,
	since \textsc{ApproxSchurNode}$(\region)$ is called in increasing order of level of $\region$, 
	we can prove inductively that all \textsc{ApproxSchurNode}$(\region)$ for $\region \in \mathcal{P}_{\mathcal{T}}(\collN)$ run correctly. 
	The time spent is bounded by $\O(\epslevel^{-2}\sum_{\region\in \mathcal{P}_{\mathcal{T}}(\collN)} |\bdry{\region}\cup \elim{\region}|)$. By \cref{lem:planarBoundChangeCost}, this is further bounded by $\O(\epssc^{-2}\sqrt{mK})$.
	
	For accessing $\ml^{(\region)}$ and $\tsc(\ml^{(H)}, \bdry{\region})$, we simply return the stored values. 
	The time required is proportional to the size of $\ml^{(\region)}$ and $\tsc(\ml^{(H)}, \bdry{\region})$ respectively,
	by the correctness properties of these Laplacians, we get the correct size and therefore the runtime.
\end{proof}

\section{Maintaining the implicit representation} 

In this section, we give a general data structure \textsc{MaintainRep}.

At a high level, \textsc{MaintainRep} implicitly maintains a vector $\vx$ throughout the IPM, 
by explicitly maintaining vector $\vy$, and implicitly maintaining a \emph{tree operator} $\mm$ and vector $\vz$,
with $\vx \defeq \vy + \mm \vz$.
\textsc{MaintainRep} supports the IPM operations \textsc{Move} and \textsc{Reweight} as follows: 
To move in step $k$ with direction $\vv^{(k)}$ and step size $\alpha^{(k)}$, 
the data structure computes some $\vz^{(k)}$ from $\vv^{(k)}$ and
updates $\vx \leftarrow \vx + \mm (\alpha^{(k)} \vz^{(k)})$. 
To reweight with new weights $\vw^\new$ (which does not change the value of $\vx$), 
the data structure computes $\mm^{\new}$ using $\vw^\new$,
updates $\mm \leftarrow \mm^{\new}$, and updates $\vy$ to offset the change in $\mm \vz$.
In \cref{subsec:maintain_z}, we define $\vz^{(k)}$ and show how to maintain $\vz = \sum_{i=1}^k \vz^{(i)}$ efficiently.
In \cref{subsec:tree_operator}, we define tree operators.
Finally in \cref{subsec:maintain_rep_impl}, we implement \textsc{MaintainRep} for a general tree operator $\mm$. 

Our goal is for this data structure to maintain the updates to the slack and flow solutions at every IPM step. 
Recall at step $k$, we want to update the slack solution by $\bar{t} h \mw^{1/2} \widetilde{\mproj}_{\vw} \vv^{(k)}$ 
and the partial flow solution by $ h \mw^{-1/2}  \widetilde{\mproj}'_{\vw} \vv^{(k)}$.
In later sections, we define specific tree operators $\mm^{\slack}$ and $\mm^{\flow}$ so that the slack and flow updates can be written as $\mm^{\slack} (\bar{t} h \vz^{(k)})$ and $\mm^{\flow} (h \vz^{(k)})$ respectively.
This then allows us to use two copies of \textsc{MaintainRep} to maintain the solutions throughout the IPM.

To start, recall the information stored in the \textsc{DynamicSC} data structure: at every node $H$ we have Laplacian $\ml^{(H)}$.
In the previous section, we defined matrices $\widetilde{\mga}$ and $\mpi^{(i)}$'s as functions of the $\ml^{(H)}$'s, in order to approximate $\ml^{-1}$.
\textsc{MaintainRep} will contain a copy of the \textsc{DynamicSC} data structure; therefore, the remainder of this section will freely refer to $\widetilde{\mga}$ and $\mpi^{(0)}, \cdots, \mpi^{(\eta-1)}$.

\subsection{Maintaining the intermediate vector $\vz$} \label{subsec:maintain_z}

We define a partial computation at each step of the IPM, which will be shared by both the slack and flow solutions:

\begin{definition}[$\vz^{(k)}$] \label{defn:z^k}
	At the $k$-th step of the IPM, let $\vv^{(k)}$ be the step direction.
	Let $\vd \defeq \mb^\top \mw^{1/2} \vv^{(k)}$.
	Define $\vz^{(k)}$ to be the partial computation
	\begin{equation}\label{eq:z}
		\vz^{(k)} \defeq
		\widetilde{\mga}
		\mpi^{(\eta-1)}\cdots \mpi^{(0)}  \vd.
	\end{equation}
\end{definition}

Observe that this is a partial projection: If we apply $\mw^{1/2} \mb \mpi^{(0)\top} \cdots \mpi^{(\eta -1)\top}$ to $\vz^{(k)}$, then by \cref{thm:L-inv-approx}, the result is an approximation to $\mproj_{\vw} \vv^{(k)}$.

We first show how to multiply $\widetilde{\mga} \mpi^{(\eta-1)} \cdots \mpi^{(0)}$ to a vector efficiently.
The main idea is to take advantage of the hierarchical structure of the separator tree $\mathcal{T}$ in a bottom-up fashion.
If $\vd$ is a sparse vector with only $K$ non-zero entries, then we can apply the operator while avoiding exploring parts of $\ct$ that are guaranteed to contain zero values.

\begin{algorithm}
	\caption{Data structure to maintain the intermediate vector $\vz$, Part 1}
	\label{algo:maintain_z}
	\begin{algorithmic}[1]
		\State \textbf{data structure} \textsc{MaintainZ}
		\State \textbf{private: member}
		\State \hspace{4mm} $G$: input graph $G$ with incidence matrix $\mb$
		\State \hspace{4mm} $\ct$: separator tree of $G$ of height $\eta$
		\State \hspace{4mm} $c \in \R, \zprev, \zsum \in \R^n$: coefficient and vectors to be maintained
		\State \hspace{4mm} $\vu \in \R^n$: vector to be maintained such that $\vu= \mpi^{(\eta-1)} \cdots \mpi^{(0)} \mb^{\top} \mw \vv$
		\State \hspace{4mm} $\vv \in \R^m$: direction vector from the current iteration
		\State \hspace{4mm} $\vw \in \R^m$: weight vector
		\Comment we sometimes also use $\mw \defeq \diag(\vw)$
		\State \hspace{4mm} \texttt{dynamicSC}: an instance of \textsc{DynamicSC} struct
		\Comment gives read access to $\ml^{(H)}$ for $H \in \ct$
		\State
		\Procedure{Initialize}{$G,\vv \in\R^{m},\vw\in\R_{>0}^{m},\epssc>0$}
		\State $\vw \leftarrow \vw$, $\vv \leftarrow \vv$
		\State $\texttt{dynamicSC}.\textsc{Initialize}(G, \vw, \epssc)$
		\State $\vu \leftarrow \textsc{PartialProject}(\mb^\top \mw^{1/2} \vv)$
		\State $\zprev \leftarrow \widetilde{\mga}\vu$
		\State $\zsum \leftarrow \vzero$
		\State $c \leftarrow 0$
		\EndProcedure
		\State

		\Procedure{PartialProject}{$\vd \in \R^{n}, \mathcal{H} = \{H \in \ct : \vd|_{F_H} \neq \vzero\}$}
		\State \Comment if $\mathcal{H}$ is not given in the argument, then it takes the default value above
		\State $\vu \leftarrow \vd$
		\For{$i$ from $0$ to $\eta-1$}  \label{line:outerloop}
		\State $
		\vu \leftarrow
		\vu - \sum_{\region \in \pathT{\collN, i}} \ml^{(\region)}_{\bdry{\region}, \elim{\region}} (\ml^{(\region)}_{\elim{\region},\elim{\region}})^{-1}\cdot \vu|_{\elim{\region}}
		$ \label{line:projectFBoundary}
		\EndFor


		\State \Return $\vu$
		\EndProcedure
		\State
		\Procedure{InversePartialProject}{$\vu \in \R^n, \collN$}

		\For{$i$ from $\eta-1$ to $0$}
		\State $
		\vu \leftarrow
		\vu + \sum_{\region \in \pathT{\collN, i}} \ml^{(\region)}_{\bdry{\region}, \elim{\region}} (\ml^{(\region)}_{\elim{\region},\elim{\region}})^{-1}\cdot \vu|_{\elim{\region}}
		$
		\label{line:reverseProjectFBoundary}
		\EndFor
		\State $\vd \leftarrow \vu$
		\State \Return $\vd$
		\EndProcedure

\algstore{maintain-z-break-point-1}
\end{algorithmic}
\end{algorithm}
\addtocounter{algorithm}{-1}

\begin{lemma} \label{lem:partial_project}
	Given a vector $\vd \in \R^{n}$, let $\mathcal{H} \supseteq \{H \in \ct : \vd|_{F_H} \neq \vzero\}$ and suppose $|\mathcal{H}| = K$.
	Then the procedure $\textsc{PartialProject}(\vd, \collN)$ in the \textsc{MaintainZ} data structure (\cref{algo:maintain_z}) returns the vector
	\[
	\vu = \mpi^{(\eta-1)}\cdots\mpi^{(1)}\mpi^{(0)} \vd,
	\]
	where the $\mpi^{(i)}$'s and $\epssc$ are from the \textsc{DynamicSC} data structure in \textsc{MaintainZ}.

	The procedure runs in $\otilde(\epsilon_{\mproj}^{-2} \sqrt{mK})$ time,
	and $\tz|_{\elim{\region}}$ is non-zero for at most $\O(K)$ nodes $\region\in \pathT{\collN}$.
\end{lemma}
\begin{proof}
	First, we consider the runtime.
	We remark that the creation of vector $\vu$ is for readability; the procedure can in fact be computed using $\vd$ in-place.

	The bottleneck of \textsc{PartialProject} is \cref{line:projectFBoundary}.
	For each $H \in \mathcal{P}_{\ct}(\mathcal{H})$, recall from \cref{thm:apxsc} that $\ml^{(\region)}$ is supported on the vertex set $\elim{\region} \cup \bdry{\region}$ and has $\O(\epssc^{-2}|\elim{\region}\cup \bdry{\region}|)$ edges. Hence,
	$(\ml^{(\region)}_{\elim{\region},\elim{\region}})^{-1} \vu|_{\elim{\region}}$ can be computed by an exact Laplacian solver in $\O(\epssc^{-2}|\elim{\region}\cup \bdry{\region}|)$ time, and the subsequent left-multiplying by $\ml^{(\region)}_{\bdry{\region}, \elim{\region}}$ also takes $\O(\epssc^{-2}|\elim{\region}\cup \bdry{\region}|)$ time.
	Finally, we can add the resulting vector to $\vu$ in time linear in the sparsity.
	Summing this over all $\region\in \pathT{\collN}$, we get that the total runtime is $\O(\epssc^{-2} \sqrt{mK})$ by \cref{lem:planarBoundChangeCost}.

	To show the correctness of \textsc{PartialProject}, we have the following claim:

	\begin{claim} \label{claim:partial_project_mpi}
		Let $\vu^{(-1)} = \vd$ be the value of $\vu$ in \textsc{PartialProject}$(\vd, \mathcal{H})$ before the first double for-loop.
		Let $\vu^{(i)}$ be the value of $\tz$ after iteration $i$ of the outer loop (\cref{line:outerloop}) for $0\le i<\eta$. Then
		\[
		\vu^{(i)} = \mpi^{(i)} \cdots \mpi^{(0)} \vd.
		\]
		Furthermore, $\vu^{(i)}|_{F_H} \neq \vzero$ only if $H \in \pathT{\collN}$.
	\end{claim}
	\begin{proof}
		We prove the claim by induction.
		For $i=-1$, we are given $\vu^{(-1)}|_{F_H} = \vd|_{F_H} \neq \vzero$ exactly for all $H \in \mathcal{H} \subseteq \pathT{\collN}$.

		For $i+1$, we have, by inductive hypothesis and definition of $\mpi^{(i)}$,
		\begin{align*}
			\mpi^{(i+1)} \mpi^{(i)} \cdots \mpi^{(0)} \vd &=
			 \mpi^{(i+1)} \vu^{(i)} \\
			&= \left( \mi - \sum_{H \in \ct(i+1)} 
			\mx^{(H)} \right) \vu^{(i)}. \\
			\intertext{Since $\mx^{(H)} \in \R^{\bdry{\region} \times \elim{\region}}$ and $\vu^{(i)}|_{F_H} \neq \vzero$ only if $H \in \pathT{\collN}$, the summation above can be taken over the smaller set $\mathcal{T}(i+1) \cap \pathT{\collN} \defeq \pathT{\collN, i+1}$, giving}
			&= \vu^{(i)} - \sum_{H \in \pathT{\collN, i+1}} 
			\mx^{(H)} \vu^{(i)}|_{F_H}.
		\end{align*}
		This is exactly what is computed as $\vu$ after iteration $i$ of the outer loop at \cref{line:outerloop}. Hence, this is equal to $\vu^{(i+1)}$ by definition.

		For the sparsity condition, we note that if $\vu^{(i+1)}|_{F_H'}$ differs from $\vu^{(i)}|_{F_H'}$ at a node $H'$,
		then it was changed by a term in the summation above,
		and so we must have $F_{H'} \cap \partial H \neq \emptyset$ for some $H \in \pathT{\collN, i+1}$.
		By construction of the separator tree, this occurs only if $H'$ is an ancestor of $H$, which implies $H' \in \pathT{\collN}$.
		Combined with the inductive hypothesis, we have that $\vu^{(i+1)}|_{F_H} \neq \vzero$ only if $H \in \pathT{\collN}$.
	\end{proof}

Setting $i = \eta-1$ in the above claim immediately shows that at the end of the first double for-loop in \textsc{PartialProject}, we have $\vu = \mpi^{(\eta-1)} \cdots \mpi^{(1)} \mpi^{(0)} \vd$.

Finally, to complete the sparsity argument, we have $|\mathcal{H}| = K$, and consequently $|\pathT{\mathcal{H}}| = O(K \cdot \eta) = \otilde(K)$.
Combined with the claim, we get the overall sparsity guarantee.
\end{proof}

For the correctness of our data structure, we will need a more specific structural property of \textsc{PartialProject}:
\begin{lemma}
	\label{lem:project_by_node}
	Let $\collN$ be any subset of nodes in $\ct$. Let $H_1, \ldots, H_r$ be any permutation of all nodes from $\pathT{\collN}$
	such that if $H_i$ is an ancestor of $H_j$, then $i<j$.
	Then
	\[
	\textsc{PartialProject}(\vd, \collN)=(\mi-\mx^{(H_1)})\ldots(\mi-\mx^{(H_r)})\vd.
	\]
\end{lemma}
\begin{proof}
	First, we observe that $\mi-\mx^{(H_i)}$ and $\mi-\mx^{(H_j)}$ are commutative if $H_i$ and $H_j$ are not ancestor-descendants.
	The reason is that $\mx^{(H_i)}\mx^{(H_j)}=\vzero$, since $\mx^{(H_i)} \in \R^{\bdry{H_i}\times \elim{H_i}}$, and $\elim{H_i} \cap \bdry{H_j} \neq \emptyset$ only if $H_i$ is an ancestor of $H_j$.

	From the proof of \cref{claim:partial_project_mpi}, we observe that iteration $i$ of the for-loop in \textsc{PartialProject} applies the operator
	\[
	\mi - \sum_{\region \in \pathT{\collN, i}} \mx^{(H)} =\prod_{H\in\pathT{\collN, i}}(\mi-\mx^{(H)}),
	\]
	where the equality follows from expanding the RHS and applying the property $\mx^{(H_i)}\mx^{(H_j)}=\vzero$.
	Thus, we have a stricter version of the claim:
	\[ \textsc{PartialProject}(\vd, \collN)=(\mi-\mx^{(H_1)})\ldots(\mi-\mx^{(H_r)})\vd,
	\]
	where $H_1,\ldots, H_r$ is any permutation of $\pathT{\collN}$ such that nodes at lower levels come later.
	Then we apply commutativity to allow $H_1,\ldots, H_r$ to be any permutation such that if $H_i$ is an ancestor of $H_j$ then $i < j$.
\end{proof}

Next, we show there is a procedure that reverses \textsc{PartialProject} using select nodes of $\ct$.

\begin{lemma} \label{lem:inverse_partial_project}
Given a set of $K$ nodes $\mathcal{H}$ in $\ct$ and a vector $\vu$,
$\textsc{InversePartialProject}(\vu, \collN)$ in the \textsc{MaintainZ} data structure (\cref{algo:maintain_z}) is a procedure that returns $\vd$ such that
\[
\vd= (\mi+\mx^{(H_r)})\ldots(\mi+\mx^{(H_1)})\vu,
\] where $H_1, \ldots, H_r$ is any permutation of all nodes from $\pathT{\collN}$
	such that if $H_i$ is an ancestor of $H_j$, then $i<j$. 
The procedure runs in $\otilde(\epsilon_{\mproj}^{-2} \sqrt{mK})$ time, where $K = |\collN|$.
\end{lemma}
\begin{proof}
Intuitively, observe that \textsc{InversePartialProject} is reversing all the operations in \textsc{PartialProject}.
The runtime analysis is analogous to \textsc{PartialProject}.
The proof of the equation is also analogous to \textsc{PartialProject}. We first observe that iteration $i$ of the for-loop applies the operator
\[
	\mi + \sum_{\region \in \pathT{\collN, i}} \mx^{(H)} =\prod_{H\in\pathT{\collN, i}}(\mi+\mx^{(H)}).
	\]
	Then by commutativity as in \cref{lem:project_by_node}, we have
	\[
\vd= (\mi+\mx^{(H_r)})\ldots(\mi+\mx^{(H_1)})\vu.
\]
where $H_1,\ldots, H_r$ is any permutation of $\pathT{\collN}$ such that nodes at lower levels come later.
	Then we apply commutativity to allow $H_1,\ldots, H_r$ to be any permutation such that if $H_i$ is an ancestor of $H_j$ then $i < j$.
\end{proof}

\begin{algorithm}
	\renewcommand{\thealgorithm}{}
	\caption{\textbf{\ref{algo:maintain_z}} Data structure to maintain the intermediate vector $\vz$, Part 2}
	\begin{algorithmic}[1]
		\algrestore{maintain-z-break-point-1}
		\Procedure{Reweight}{$\vw^\new \in\R_{>0}^{m}$}
		\State $\vw \leftarrow \vw^\new$

		\State $\collN \leftarrow$ set of leaf nodes in $\mathcal{T}$ that contain all the edges of $G$ whose weight has changed
		\State $\Delta \vu \leftarrow \textsc{PartialProject}(\mb^{\top} (\mw^{\new 1/2} - \mw^{1/2}) \vv)$\label{line:maintainzreweightproject}
		\State $\vu \gets \vu + \Delta \vu$ \label{line:maintainzupdated}
		\State $\vd \leftarrow \textsc{InversePartialProject}(\vu, \mathcal{H})$\label{line:deproject}
		\Comment revert projection with old weights
		\State $\texttt{dynamicSC}.\textsc{Reweight}(\vw^\new)$\label{line:maintainZupdateWeights}
		\Comment update $\ml^{(\region)}$'s to use the new weights
		\State \Comment specifically, $\ml^{(H)}$ changes for each $H \in \pathT{\collN}$\label{line:reproject}
		\State $\vu \leftarrow \textsc{PartialProject}(\vd, \mathcal{H})$
		\Comment apply projection with new weights
		\State $\vy \leftarrow \zprev$ \Comment{backup copy of $\zprev$}
		\For{$\region$ in $\pathT{\collN}$}
		\State  $\zprev|_{\elim{\region}} \leftarrow (\ml^{(H)}_{\elim{H}, \elim{H}})^{-1}\vu|_{\elim{\region}}$ \label{line:reweightgamma}
		\EndFor
		\State $\zsum \gets \zsum - c \cdot (\zprev-\vy)$
		\Comment{update $\zsum$ to maintain the invariant} \label{line:maintainZcancelChanges}
		\EndProcedure
		\State
		\Procedure{Move}{$\alpha \in \R, \vnew \in \R^{m}$}
		\State $\Delta \vv \leftarrow \vnew- \vv$
		\State $\vv \leftarrow \vnew$
		\State $\Delta \vu \leftarrow \textsc{PartialProject}(\mb^\top \mw^{1/2} \Delta \vv)$
		\State $\vu \leftarrow \vu+\Delta \vu$ \label{line:maintain_z:update_u}
		\State $\vy \leftarrow \zprev$ \Comment{backup copy of $\zprev$}
		\For{$\region$ in $\pathT{\collN}$}
		\State  $\zprev|_{\elim{\region}} \leftarrow (\ml^{(H)}_{\elim{H}, \elim{H}})^{-1}\vu|_{\elim{\region}}$\label{line:maintainZmove1}
		\EndFor
		\State $\zsum \gets \zsum - c \cdot (\zprev-\vy)$\label{line:maintainZmove2}
		\State $c \leftarrow c + \alpha$ \label{line:maintainZmovec}
		\EndProcedure
	\end{algorithmic}
\end{algorithm}

Finally, we have the data structure for maintaining a vector $\vz$ dependent on $\vv$ throughout the IPM.
For one IPM step, there is one call to \textsc{Reweight} followed by one call to \textsc{Move}.

\begin{restatable}[Maintain intermediate vector $\vz$] {theorem}{MaintainZ} \label{thm:maintain_z}
	Given a modified planar graph $G$ with $n$ vertices and $m$ edges and its separator tree $\ct$ with height $\eta$,
	the deterministic data structure \textsc{MaintainZ} (\cref{algo:maintain_z})
	maintains the following variables correctly at the end of each IPM step:
	\begin{itemize}
		\item the dynamic edge weights $\vw$ is and current step direction $\vv$ from the IPM
		\item a \textsc{DynamicSC} data structure on $\ct$ based on the current edge weights $\vw$
		\item scalar $c$ and vectors $\zprev, \zsum$, which together represent $\vz = c \zprev + \zsum$,
		such that at the end of IPM step $k$,
		\begin{equation} \label{eq:z-invariant}
			\vz = \sum_{i=1}^{k} \vz^{(i)}.
		\end{equation}
		\item $\zprev$ satisfies $\zprev = \widetilde{\mga} \mpi^{(\eta-1)} \cdots \mpi^{(0)} \mb^{\top} \mw^{1/2} \vv.$
	\end{itemize}
	The data structure supports the following procedures:
	\begin{itemize}
		\item $\textsc{Initialize}(G, \text{ separator tree } \ct, \vv\in\R^{m},\vw\in\R_{>0}^{m}, \epsilon_{\mproj} > 0)$:
		Given a graph $G$, its separator tree $\ct$, initial step direction $\vv$, initial weights $\vw$, and target projection matrix accuracy $\epsilon_{\mproj}$, preprocess in $\widetilde{O}(\epssc^{-2}m)$ time
		and initialize $\vz = \vzero$.

		\item $\textsc{Reweight}(\vw\in\R_{>0}^{m}$  given implicitly as a set of changed coordinates):
		Update the current weight to $\vw$ and update \textsc{DynamicSC},
		and update the representation of $\vz$.
		The procedure runs in
		$\widetilde{O}(\epsilon_{\mproj}^{-2}\sqrt{mK})$ total time,
		where $K$ is the number of coordinates updated in $\vw$.
		There are most $\O(K)$ nodes $\region\in \ct$ for which $\zprev|_{F_H}$ and $\zsum|_{F_H}$ are updated.

		\item $\textsc{Move}(\alpha \in \R$, $\vv \in \R^{n}$ given implicitly as a set of changed coordinates):
		Update the current direction to $\vv$, and set $\vz \leftarrow \vz + \alpha \widetilde\mga \mpi^{(\eta-1)}\cdots \mpi^{(0)} \mb^{\top} \mw^{1/2} \vv$ with the correct representation.
		The procedure runs in $\widetilde{O}(\epsilon_{\mproj}^{-2} \sqrt{mK})$ time,
		where $K$ is the number of coordinates changed in $\vv$ compared to the previous IPM step.
	\end{itemize}
\end{restatable}

	\begin{proof}
	If \textsc{Move} is implemented correctly, then by the definition of the update to $\vz$, the invariant in \cref{eq:z-invariant} is correctly maintained.

	For the runtime analysis, recall $\{F_H : H \in \ct\}$ partition the vertex set of $G$.
	Therefore $\vv$ has $K$ non-zero entries, then $\vd \defeq \mb^{\top} \mw^{1/2} \vv$ has $O(K)$ non-zero entries,
	and consequently $\vd|_{F_H} \neq \vzero$ for $O(K)$ nodes $H$. There are $O(m)$ total nodes in the separator tree $\ct$.

	We maintain a vector $\vu$ with the invariant $\vu = \mpi^{(\eta-1)} \cdots \mpi^{(0)} \mb^{\top} \mw^{1/2} \vv$.
	We now prove the correctness and runtime of each procedure separately.
	\paragraph{\textsc{Initialize}:}
	By the guarantee of \cref{lem:partial_project}, at the end of \textsc{Initialize}, we have
	\[
		\vu = \mpi^{(\eta-1)} \cdots \mpi^{(0)} \mb^\top \mw^{1/2} \vv
	\] and 
	\[
		\zprev = \widetilde{\mga}\vu=\widetilde{\mga} \mpi^{(\eta-1)} \cdots \mpi^{(0)} \mb^\top \mw^{1/2} \vv.
	\]
	Since $c$ and $\zsum$ are initialized to zero, we have $\vz = c \zprev + \zsum = \vzero$.

	We initialize the \textsc{DynamicSC} data structure in $\o(\epssc^{-2} m)$ time.
	There is no sparsity guarantee for $\vv$,
	but the call to \textsc{PartialProject} takes at most $O(\epssc^{-2} m)$ time because of the size of $\ct$. To calculate  $\widetilde{\mga}\vu$, we solve a Laplacian system $(\ml^{(H)}_{\elim{H}, \elim{H}})^{-1}\vu|_{\elim{H}}$ in time $\O(|\ml^{(H)}|)$ for each node $H$. The total time is $\O(\epssc^{-2} m)$ as well by $|\ml^{(H)}|=\widetilde{O}\left( \epssc^{-2} |\elim{\region} \cup \bdry{\region} |\right)$ from \cref{thm:apxsc} and by \cref{lem:planarBoundChangeCost}.

	\paragraph{\textsc{Move}:}
	Let $\vv, \vu$ be the variables at the start of \textsc{Move}, and let ${\vv}', \vu'$ denote them at the end.
	Similarly, let $\vz = c \zprev + \zsum$ denote $\vz$ and the respective variables at the start of \textsc{Move}, and
	let $\vz' = c' {\zprev}' + {\zsum}'$ denote these variables at the end.

	First, after \cref{line:maintain_z:update_u}, we have
	\begin{align*}
		{\vu}' &= \vu + \Delta \vu \\
		&= \mpi^{(\eta-1)} \cdots \mpi^{(0)} \mb^{\top} \mw^{1/2} ({\vv} + \Delta \vv) \\
		&= \mpi^{(\eta-1)} \cdots \mpi^{(0)} \mb^{\top} \mw^{1/2} {\vv}',
	\end{align*}
	where the second equality follows from the guarantee of \textsc{PartialProject} and the guarantee from the previous IPM step.
	By \cref{lem:partial_project}, ${\vu}'$ is updated only on $\elim{H}$ where $H\in \pathT{\collN}$. Thus, to update $\zprev'=\widetilde{\mga}\vu'$, we only need to update $\zprev'|_{\elim{H}}$ for $H\in \pathT{\collN}$, which happens on \cref{line:maintainZmove1}. 
	Observe that the update in value to $\zprev$ is cancelled out by the update in $\zsum$ at \cref{line:maintainZmove2}, so that the value of $\vz$ does not change overall up to that point. 
	But we have 
	\[
	\vz = c \zprev' + \zsum' = c \widetilde{\mga} \mpi^{(\eta-1)} \cdots \mpi^{(0)} \mb^{\top} \mw^{1/2} {\vv}' + \zsum'.
	\]
	Then in \cref{line:maintainZmovec}, incrementing $c$ by $\alpha$ represents increasing the value of $\vz$ by  $\alpha \zprev'$, which is exactly the desired update.
		
	For the runtime, first note $nnz(\Delta \vv) = K$. So \textsc{PartialProject} runs in $\O(\epssc^{-2} \sqrt{mK})$ time by \cref{lem:partial_project}. \cref{line:maintainZmove1} takes $\O(\epssc^{-2} \sqrt{mK})$ time in total by \cref{thm:apxsc} and \cref{lem:planarBoundChangeCost}. The remaining operations in the procedure are adding vectors with bounded sparsity.

	\paragraph{\textsc{Reweight}:}
	Let $\vw^{\old}$ denote the weight vector immediately before this procedure is called, and $\vw^{\new}$ is the new weight passed in as an argument.

	Let $\widetilde{\mga}$ and $\mpi^{(i)}$ denote these matrices defined using the old weights, and let
	$\widetilde{\mga}'$ and $\mpi^{(i)'}$ denote the matrices using the new weights. Similarly let $\vu$ be the state of the vector at the start of the procedure call and $\vu'$ at the end.

	In \textsc{Reweight}, we do not change the value of $\vz$, but rather update $\zprev$ and $\zsum$ so that at the end of the procedure,
	\[
		\zprev = \widetilde{\mga}' \mpi^{(\eta-1)'} \cdots \mpi^{(0)'} \mb^{\top} \mw^{\new 1/2} {\vv},
	\]
	so that we maintain the invariant claimed in the theorem statement.

	To see that the value of $\vz$ does not change at the end of the procedure, observe that we modify $\zprev$ during the procedure, and cancel all the changes to $\zprev$ by updating $\zsum$ appropriately at the last line (\cref{line:maintainZcancelChanges}).

	Immediately before \cref{line:maintainzupdated}, the algorithm invariant guarantees
	\[
	\vu = \mpi^{(\eta-1)} \cdots \mpi^{(0)} \mb^{\top} \mw^{\old 1/2} \vv.
	\]
	By \cref{lem:partial_project},
	\[
	\Delta \vu = \mpi^{(\eta-1)} \cdots \mpi^{(0)} \mb^{\top} \left(\mw^{\new 1/2} - \mw^{\old 1/2} \right) \vv.
	\]
	Therefore, after executing \cref{line:maintainzupdated}, we have
	\[
	\vu \leftarrow \vu + \Delta \vu = \mpi^{(\eta-1)} \cdots \mpi^{(0)} \mb^{\top} \mw^{\new 1/2} \vv.
	\]
	Next, we need to update $\vu$ to reflect the changes to $\widetilde{\mga}, \mpi^{(i)}$.
	Updating these matrices is done via \texttt{dynamicSC}. However, calling $\textsc{PartialProject}(\mb^{\top} \mw^{\new 1/2} \vv)$ afterwards is too costly if done directly, since the argument is a dense vector.
	To circumvent this problem, we make the key observation that the change to $\vu$ is restricted to a subcollection of nodes on $\ct$ (in fact a connected subtree containing the root),
	and it suffices to partially reverse and reapply the operator $\widetilde{\mga} \mpi^{(\eta-1)} \cdots \mpi^{(0)}$. Intuitively, \textsc{InversePartialProject} revert all computations in \textsc{PartialProject} that are related to the changes to $\mw$.

	Let $H_1, \dots, H_t$ be a permutation of all nodes in $\ct$, such that the nodes in $\pathT{\collN}$ is a prefix of the permutation, and it satisfies that for any node $H_i$ with descendant $H_j$, $i < j$.
	Then by \cref{lem:project_by_node}, after executing \cref{line:maintainzupdated}, we have
	\begin{align}
	\vu &= \textsc{PartialProject}(\mb^{\top} \mw^{\new 1/2} \vv, \ct) \nonumber \\
	&= (\mi-\mx^{(H_1)})\ldots(\mi-\mx^{(H_t)})\mb^{\top} \mw^{\new 1/2} \vv. \label{eq:partial_project_zprev}
	\end{align}
	Let $r = |\pathT{\collN}|$.
	Then \textsc{InversePartialProject}($\vu, \collN$) on \cref{line:deproject} returns $\vd$ by \cref{lem:inverse_partial_project} satisfying
	\begin{align*}
		\vd&= (\mi+\mx^{(H_r)})\ldots(\mi+\mx^{(H_1)})\vu.
		\intertext{Plugging in $\vu$ from \cref{eq:partial_project_zprev}, we have}
		\vd&= (\mi+\mx^{(H_r)})\ldots(\mi+\mx^{(H_1)})(\mi-\mx^{(H_1)})\ldots(\mi-\mx^{(H_t)})\mb^{\top} \mw^{\new 1/2} \vv.
	\end{align*}
	We use the fact that each $\mi-\mx^{(H_i)}$ is nonsingular and has inverse $\mi+\mx^{(H_i)}$ to get
	\[\vd=(\mi-\mx^{(H_{r+1})})\ldots(\mi-\mx^{(H_t)})\mb^{\top} \mw^{\new 1/2} \vv.
	\]

	We then call \texttt{dynamicSC}.\textsc{Reweight}, which updates $\ml^{(H)}$ and in turn $\mx^{(H_i)}$ for precisely all nodes in $\pathT{\collN} = \{ H_1, \dots, H_r\}$.
	Let $\mx^{(H)'}$ denote the matrix after reweight. Next, we call $\textsc{PartialProject}$ again. Let us denote it by $\textsc{PartialProject}^\new$ to emphasize that it runs with new weights. This gives
	\begin{align*}
		\vu' &= \textsc{PartialProject}^\new(\vd, \collN) \\
		&= (\mi - \mx^{(H_1)'}) \dots (\mi-\mx^{(H_r)'}) \vd \\
		&= (\mi - \mx^{(H_1)'}) \dots (\mi-\mx^{(H_r)'}) (\mi-\mx^{(H_{r+1})})\ldots(\mi-\mx^{(H_t)})\mb^{\top} \mw^{\new 1/2} \vv \\
		&= (\mi - \mx^{(H_1)'}) \ldots(\mi-\mx^{(H_t)'})\mb^{\top} \mw^{\new 1/2} \vv \tag{since $\mx^{(H_i)'}=\mx^{(H_i)}$ for all $i>r$} \\
		&= \textsc{PartialProject}^\new(\mb^\top \mw^{\new 1/2} \vv, \ct).
	\end{align*}
	Because $\vu'|_{F_H}$ is updated on $H\in \pathT{\collN}$, and $\ml^{(H)}$ is updated on $H\in \pathT{\collN}$ by \cref{thm:apxsc}, running \cref{line:reweightgamma} on $H\in \pathT{\collN}$ correctly sets $\zprev'=\widetilde{\mga}'\vu'$.

	For the runtime, the first call to \textsc{PartialProject} has a vector with $O(K)$ sparsity as the argument, and therefore runs in $\O(\epssc^{-2} \sqrt{mK})$.
	Next, we know $|\collN| = O(K)$.
	The call to \textsc{InversePartialProject} and the subsequent call to \textsc{PartialProject} both have $\collN$ as an argument, so they run in $\O(\epssc^{-2} \sqrt{mK})$.
	The \texttt{DynamicSC}.\textsc{Reweight} call runs in $\O(\epssc^{-2} \sqrt{mK})$. 
	Updating $\zprev$ (\cref{line:reweightgamma}) takes $\O(\epssc^{-2} \sqrt{mK})$ time in total by \cref{thm:apxsc} and \cref{lem:planarBoundChangeCost}. And finally we can update $\zsum$ in the same time.

	We remark that although \textsc{InversePartialProject} returns a vector $\vd$ that is not necessarily sparse, and we then assign $\vu \leftarrow \textsc{PartialProject}(\vd, \collN)$, this is for readability. $\vd$ is in fact an intermediate state of $\vu$, on which we perform in-place operations.
\end{proof}

\subsection{Tree operator} \label{subsec:tree_operator}

At IPM step $k$, our goal is to write the slack update $\widetilde\mproj_{\vw} \vv^{(k)}$ as $\mm^{(\textrm{slack})} \vz^{(k)}$, and similarly, write the partial flow update $\widetilde \mproj'_{\vw} \vv^{(k)}$ approximately as $\mm^{(\textrm{flow})} \vz^{(k)}$, where $\vz^{(k)}$ is defined in the previous subsection, and $\mm^{(\textrm{slack})}$ and $\mm^{(\textrm{flow})}$ are linear operators that are efficiently maintainable between IPM steps.

In this section, we define a general class of operators called \emph{tree operators} and show how to efficiently compute and maintain them.
In later sections, we show that  $\mm^{(\textrm{slack})}$ and $\mm^{(\textrm{flow})}$ can be defined as tree operators.

We begin with the formal definitions.
Recall for a tree $\ct$ and node $H \in \ct$, we use $\ct_H$ to denote the subtree rooted at $H$.

\begin{definition}[Tree operator]
	\label{def:forest-operator} 
	Suppose $\ct$ is a rooted tree with constant degree.
	Let each node $H \in \ct$ be associated with two sets $V(H)$ and $F_{H} \subseteq V(H)$.
	Let each leaf node $H \in \ct$ be further associated with a non-empty set $E(H)$ of constant size, where \emph{the $E(H)$'s are pairwise disjoint over all leaf nodes}. For a non-leaf node $H$, define $E(H) \defeq \bigcup_{\text{leaf } D \in \ct_H} E(D)$.
	Finally, define $E \defeq E(G) \bigcup_{\text{leaf } H \in \ct} E(H)$ and $V \defeq V(G) = \bigcup_{H \in \ct} V(H) $, where $G$ is the root node of $\ct$.
	
	Let each node $H$ with parent $P$ be associated with a linear \emph{edge operator} $\mm_{(H, P)} : \R^{V(P)} \mapsto \R^{V(H)}$.
	In addition, let each leaf node $H$ be associated with a \emph{constant-time computable} linear \emph{leaf operator} $\mj_H : \R^{V(H)} \mapsto \R^{E(H)}$. 
	We extend all these operators trivially to $\R^V$ and $\R^E$ respectively, in order to have matching dimensions overall.
	When a edge or leaf operator is not given, we assume it to be $\mzero$.
	
	For a path $H_t \rightarrow H_1 \defeq (H_t, \dots, H_1)$, where each $H_i$ is the parent of $H_{i-1}$ and $H_1$ is a leaf node (call these \emph{tree paths}), we define 
	\[
	\mm_{H_1 \leftarrow H_t} = \mm_{(H_{1}, H_{2})} \mm_{(H_{2},H_{3})} \cdots \mm_{(H_{t-1},H_{t})}.
	\]
	If $t = 1$, then $\mm_{H_1 \leftarrow H_t} \defeq \mi$.
	
	We define \emph{the tree operator $\mm : \R^V \mapsto \R^E$ supported on $\ct$} to be
	\begin{equation} \label{defn:mm}
		\mm \defeq \sum_{\text{leaf $H$, node $A$} \;:\; H \in \ct_A} \mj_{H}\mm_{H\leftarrow A}\mi_{F_{A}}.
	\end{equation}
	
	We always maintain a tree operator implicitly by maintaining 
	\[
	\{ \mj_H : \text{leaf } H\} \cup \{ \mm_{(H,P)} : \text{edge } (H,P) \} \cup \{F_H : \text{node } H\}.
	\]
\end{definition}

\begin{remark}
Although we define the tree operator in general and hope it will find applications in other problems, 
we have used suggestive names in the definition to suit our min-cost flow setting.
In particular, our tree operators will be supported on the separator tree $\ct$. For each node $H$, the sets $V(H), F_H, E(H)$ associated with the tree operator are, respectively, $\partial H \cup F_H$ of region $H$, the eliminated vertices $F_H$ of region $H$, and the edge set of region $H$, all from the separator tree construction.
\end{remark}

To maintain $\mm$ using the tree efficiently, we also need some partial operators:
\begin{definition}[$\mm^{(H)}, \overline{\mm^{(H)}}$]
\label{defn:subtree-operator}
	For notational convenience, define $\ct_H$ to be the subtree of $\ct$ rooted at $H$.
	
	We define the subtree operator $\mm^{(H)} : V(H) \mapsto E(H)$ at each node $H$ to be
	\begin{equation}
		\mm^{(H)} \defeq  \sum_{\text{leaf } D \in \ct_H} \mj_{D} \mm_{D\leftarrow H}.
	\end{equation}
	We also define the partial sum
	\begin{equation}
		\overline{\mm^{(H)}} \defeq \sum_{D \in \ct_H} \mm^{(D)} \mi_{F_D}.
	\end{equation}
\end{definition}

We state a straightforward corollary based on the definitions without proof.
\begin{corollary}
	For any node $H \in \ct$,
	\[
	\mm = \sum_{H \in \ct} \mm^{(H)} \mi_{F_H} = \overline{\mm^{(G)}},
	\]
	where $G$ is the root node of $\ct$.
	
	Furthermore, if $H$ has with children $D_1, D_2$, then
	\begin{equation}
	\mm^{(H)} = \mm^{(D_1)} \mm_{(D_1, H)} + \mm^{(D_2)} \mm_{(D_2, H)}.
	\end{equation}
\end{corollary}

We define the complexity of a tree operator to be parameterized by the number of tree edges.

\begin{definition}[Complexity of tree operator]
	\label{def:forestcost}
	Let $\mm$ be a tree operator on tree $\ct$.
	We say $\mm$ has complexity function $T$, 
	if for any $k > 0$,
	for any set $S$ of $k$ distinct edges in $\ct$ and any families of vectors $\{\vu_e : e \in S\}$ and $\{\vv_e : e \in S\}$, 
	the total cost of computing 
	$\{\vu_e^{\top}\mm_{e} : e \in S\}$ and $\{\mm_{e}\vv_e : e \in S\}$ is bounded by $T(k)$. 
	
	Without loss of generality, we may assume $T(0) = 0$, $T(k) \geq k$, and $T$ is concave.
\end{definition}

We can show the structure of a tree operator by the procedure \textsc{ComputeMz}($\mm, \vz$) to compute $\mm \vz$. Intuitively, $\vz$ is given as input to each node $H$. The edge operators are concatenated in the order of tree paths from $H$ to a leaf, but we apply them level-wise in descending order.

\begin{algorithm}
	\caption{Compute $\mm\vz$ for a tree operator $\mm$}\label{algo:computeMz}
	\begin{algorithmic}[1]
		
		\Procedure{\textsc{ComputeMz}}{$\mm, \vz$}
		\State $\collN \leftarrow$  set of all nodes $H$ in $\mathcal{T}$ such that $\mm_{(H, P)}$ or $\mj_H$ is nonzero
		\State $\pathT{\collN} \leftarrow$ set of $\mathcal{H}$ and all ancestor nodes of $\collN$ in  $\mathcal{T}$
		\State $\vv_H \leftarrow \vzero$ for each $H \in \ct$ \Comment{sparse vectors for intermediate computations}
		\For{each node $H \in \pathT{\collN}$}
		\State $\vv_H \leftarrow \mi_{F_H} \vz = \vz|_{F_H}$  \Comment{apply the $\mi_{F_H}$ part of the operator}
		\EndFor
		\For{each node $H \in \pathT{\collN}$ by decreasing level}
		\State Let $P$ be the parent of $H$
		\State $\vv_H \leftarrow \vv_H + \mm_{(H, P)} \vv_P$ \Comment{apply $\mm_{(H,P)}$ as we move from $P$ to $H$}
		\EndFor
		\For{each leaf node $H \in \pathT{\collN}$}
		\State $\vx|_{E(H)} \leftarrow \mj_H \vv_H$	\Comment{apply the leaf operator}
		\EndFor
		\State \Return $\vx$
		\EndProcedure
	\end{algorithmic}
\end{algorithm}

\begin{corollary}
	\label{cor:exacttime}
	Suppose $\mm : \R^V \rightarrow \R^E$ is a tree operator on tree $\ct$ with complexity $T$, where $|V| = n$ and $|E| = m$. Then for $\vz \in \R^V$,
	$\textsc{Exact}(\mm, \vz)$ outputs $\mm \vz$ in $O(T(K))=O(T(m))$ time where $K$ is the total number of non-zero edge and leaf operators in $\mm$.
\end{corollary}
\begin{proof}
	Note only non-zero edge and leaf operators contribute to $\mm\vz$. We omit the proof of correctness as it is simply an application of the definition.
	
	Since $E = \cup_{\text{leaf $D$}} E(D)$, and each $E(D)$ has constant size, we know there are at most $O(m)$ leaves in $\ct$. 
	Hence, there are $O(m)$ edges in $\ct$, and $K=O(m)$.
	Since we define each leaf operator to be constant time computable, applying $\mj_H$ for leaves in $\pathT{\collN}$ costs $O(K)$ time in total.
	The bottleneck of the procedure is to apply the edge operator $\mm_e$ to some vector exactly once for each edge $e$ in $\ct$; the time cost is $O(T(K))$ by definition of the operator complexity. 
\end{proof}

\subsection{Proof of \crtcref{thm:maintain_representation}} \label{subsec:maintain_rep_impl}

Finally, we give the data structure for maintaining an implicit representation of the form $\vy + \mm \vz$ throughout the IPM.
For an instantiation of this data structure, there is exactly one call to \textsc{Initialize} at the very beginning, and one call to \text{Exact} at the very end.
Otherwise, each step of the IPM consists of one call to \textsc{Reweight} followed by one call to \textsc{Move}.
Note that this data structure \emph{extends} \textsc{MaintainZ} in the object-oriented programming sense.

\MaintainRepresentation*
\begin{algorithm}
	\caption{Implicit representation maintenance}\label{alg:maintain_representation}
	\begin{algorithmic}[1]
		\State \textbf{data structure} \textsc{MaintainRep} \textbf{extends} \textsc{MaintainZ}
		\State \textbf{private: member}
		\State \hspace{4mm} $\ct$: separator tree
		\State \hspace{4mm} $\vy \in \R^{m}$: offset vector
		\State \hspace{4mm} $\mm$: \emph{instructions to compute the tree operator} $\mm \in \R^{m \times n}$
		\State \LeftComment \; $\vz = c \zprev + \zsum$ maintained by \textsc{MaintainZ}, accessable in this data structure
		\State \LeftComment \; \texttt{DynamicSC}: an accessable instance of \textsc{DynamicSC} maintained by \textsc{MaintainZ}
		\State 
		\Procedure{Initialize}{$G,\ct, \mm, \vv\in \R^{m}, \vw\in\R_{>0}^{m}, \vx^\init \in \R^{m}, \epssc>0$}
		\State $\mm \leftarrow \mm$ \Comment{initialize the instructions to compute $\mm$}
		\State $\texttt{Super}.\textsc{Initialize}(G, \ct, \vv\in \R^{m}, \vw\in\R_{>0}^{m},\epssc>0)$ \Comment initialize $\vz$
		\State $\vy\leftarrow \vx^\init$ \label{line:init_y}
		\EndProcedure
		\State
		\Procedure{Reweight}{$\vw^\new$}
		\State Let $\mm^\old$ represent the current tree operator $\mm$
		\State $\texttt{Super}.\textsc{Reweight}(\vw^\new)$ \Comment{update representation of $\vz$ and \texttt{DynamicSC}}\label{line:super_reweight}
		\State \Comment $\mm$ is updated as a result of reweight in \texttt{DynamicSC}
		\State $\Delta \mm \leftarrow \mm - \mm^\old$ \Comment{$\Delta \mm$ is represented implicitly}
		\State $\vy\leftarrow \vy - \textsc{ComputeMz}(\Delta \mm, c\zprev+\zsum)$ \Comment{\cref{algo:computeMz}} \label{line:update_y}
		\EndProcedure
		\State
		\Procedure{Move}{$\alpha, \vv^\new$}
		\State $\texttt{Super}.\textsc{Move}(\alpha, \vv^\new)$ 
		\EndProcedure
		\State
		\Procedure{Exact}{$ $}
		\State \Return $\vy + \textsc{ComputeMz}(\mm, c\zprev+\zsum)$ \Comment{\cref{algo:computeMz}}
		\EndProcedure
	\end{algorithmic}
\end{algorithm}
\begin{proof}
	First, we discuss how $\mm$ is stored in the data structure:
	Recall $\mm$ is represented implicitly by a collection of edge operators and leaf operators on the separator tree $\ct$,
	so that each edge operator is stored at a corresponding node of $\ct$, and each leaf operator is stored at a corresponding leaf node of $\ct$.
	However, the data structure \emph{does not} store any edge or leaf operator matrix explicitly.
	We make a key assumption that each edge and leaf operator is computable using $O(1)$-number of $\ml^{(H)}$ matrices from \textsc{DynamicSC}. This will be true for the slack and flow operators we define.
	As a result, to store an edge or leaf operator at a node, we simply store \emph{pointers to the matrices from \textsc{DynamicSC}} required in the definition, and an $O(1)$-sized instruction for how to compute the operator.
	The computation time is proportional to the size of the matrices in the definitions, but crucially the instructions have only $O(1)$-size. 
	
	Now, we prove the correctness and runtime of each procedure separately.
	Observe that the invariants claimed in the theorem are maintained correctly if each procedure is implemented correctly.
	
	\paragraph{\textsc{Initialize}:}
	\cref{line:init_y} sets $\vy\leftarrow \vx^\init$, and $\texttt{Super}.\textsc{Initialize}$ sets $\vz\leftarrow \vzero$.
	So we have $\vx=\vy+\mm\vz$ at the end of initialization. Furthermore, the initialization of $\vz$ correctly sets $\zprev$ in terms of $\vv$.
	
	By \cref{thm:maintain_z}, \texttt{Super}.\textsc{Initialize} takes $\widetilde{O}(\epssc^{-2}m)$ time.
	Storing the implicit representation of $\mm$ takes $O(m)$ time.  
	
	\paragraph{\textsc{Reweight}:}
	By \cref{thm:maintain_z}, $\texttt{Super}.\textsc{Reweight}$ updates its current weight and \texttt{DynamicSC}, and updates $\zprev$ correspondingly to maintain the invariant, while not changing the value of $\vz$.
	Because $\mm$ is stored by instructions, no explicit update to $\mm$ is required. \cref{line:update_y} updates $\vy$ to zero out the changes to $\mm\vz$. 
	
	The instructions for computing $\Delta \mm$ require the Laplacians from \texttt{DynamicSC} before and after the update in \cref{line:super_reweight}. 
	For this, we monitor the updates of \texttt{dynamicSC} and stores the old and new values. 
	The runtime of this is bounded by the runtime of updating \texttt{dynamicSC}, which is in turn included in the runtime for \texttt{Super}.\textsc{Reweight}.
	
	Let $K$ upper bound the number of coordinates changed in $\vw$ and the number of edge and leaf operators changed in $\mm$. Then $\texttt{Super}.\textsc{Reweight}$ takes $\widetilde{O}(\epssc^{-2}\sqrt{mK})$ time, 
	and $\textsc{Exact}(\Delta \mm, \vz)$ takes $O(T(K))$ time.
	Thurs, the total runtime is $\widetilde{O}(\epssc^{-2} \sqrt{mK} +T(m))$.
	
	\paragraph{\textsc{Move}:}
	The runtime and correctness follow from \cref{thm:maintain_z}.
	
	\paragraph{\textsc{Exact}:}
	\textsc{ComputeMz} computes $\mm\vz$ correctly in $O(T(m))$ time by \cref{cor:exacttime}. Adding the result to $\vy$ takes $O(m)$ time and gives the correct value of $\vx=\vy+\mm\vz$. Thus, \textsc{Exact} returns $\vx$ in $O(T(m))$ time.
\end{proof}

\section{Maintaining vector approximation}
\label{sec:sketch}

Recall at every step of the IPM, we want to maintain approximate vectors $\os, \of$ so that
\begin{align*}
	\norm{ \mw^{-1/2} (\of - \vf)}_\infty \leq \delta \quad \text{and} \quad 
	\norm{ \mw^{1/2} (\os- \vs)}_\infty \leq \delta'
\end{align*}
for some additive error tolerances $\delta$ and $\delta'$.

In the previous section, we showed how to maintain some vector $\vx$ implicitly as $\vx \defeq \vy + \mm \vz$ throughout the IPM, where $\vx$ should represent $\vs$ or part of $\vf$.
In this section, we give a data structure to efficiently maintain an approximate vector $\ox$ to the $\vx$ from \textsc{MaintainRep}, so that at every IPM step,
\[
	\norm{\md^{1/2} \left(\ox - \vx\right)}_\infty \leq \delta,
\]
where $\md$ is a dynamic diagonal scaling matrix. (It will be $\mw^{-1}$ for the flow or $\mw$ for the slack.) 

In \cref{subsec:sketch_vector_to_change}, we reduce the problem of maintaining $\ox$ to detecting coordinates in $\vx$ with large changes. 
In \cref{subsec:sketch_change_to_sketch}, we detect coordinates of $\vx$ with large changes using a sampling technique on a binary tree, where Johnson-Lindenstrauss sketches of subvectors of $\vx$ are maintained at each node the tree. 
In \cref{subsec:sketch_maintenance}, we show how to compute and maintain the necessary collection of JL-sketches on the separator tree $\ct$; in particular, we do this efficiently with only an implicit representation of $\vx$.
Finally, we put the three parts together to prove \cref{thm:VectorTreeMaintain}.

We use the superscript $^{(k)}$ to denote the variable at the end of the $k$-th step of the IPM; that is, $\md^{(k)}$ and $\vx^{(k)}$ are $\md$ and $\vx$ at the end of the $k$-th step. Step 0 is the state of the data structure immediately after initialization.

\subsection{Reduction to change detection} 
\label{subsec:sketch_vector_to_change}

In this subsection, we show that in order to maintain an approximation $\ox$ to some vector $\vx$, it suffices to detect coordinates of $\vx$ that change a lot. 

Here, we make use of dyadic intervals, and at step $k$ of the IPM, for each $\ell$ such that $k = 0 \bmod 2^\ell$, we find the set $I_{\ell}^{(k)}$ that contains all coordinates $i$ of $\vx$ such that $\vx_i^{(k)}$ changed significantly compared to $\vx_i^{(k-2^\ell)}$, that is, compared to $2^\ell$ steps ago. Formally:

\begin{definition} \label{defn:sketch-I_ell^k}
	At step $k$ of the IPM, for each $\ell$ such that $k  = 0 \bmod 2^\ell$, we define 
	\begin{align*}
		I_{\ell}^{(k)} \defeq &\; \{i\in[n]: \sqrt{\md_{ii}^{(k)}} \cdot|\vx_{i}^{(k)}-\vx_{i}^{(k-2^{\ell})}|\geq\frac{\delta}{2\left\lceil \log m\right\rceil } \\
		&\;\text{and $\ox_i$ has not been updated 
			since the $(k-2^{\ell})$-th step}\}.
	\end{align*}
	We say that $\ox_i$ has not been updated since the $(k-2^{\ell})$-th step if $\ox_{i}^{(j)}=\ox_{i}$ and $\md_{ii}^{(j)}=\md_{ii}^{(k-2^\ell)}$ for $j\geq k-2^{\ell}$, i.e. $\ox_i$ was not updated by \cref{line:D-induced-ox-change} or \cref{line:update-ox} in the $(k-2^{\ell}+1),\ldots, (i-1)$-th steps.
\end{definition}

We show how to find the sets $I_\ell^{(k)}$ with high probability in the next subsection.
Assuming the correct implementation, we have the following data structure for maintaining the desired approximation $\ox$:

\begin{algorithm}
	\caption{Data structure \textsc{AbstractMaintainApprox}, Part 1}
	\label{algo:maintain-vector}
	
	\begin{algorithmic}[1]
		\State \textbf{data structure} \textsc{AbstractMaintainApprox} 
		\State \textbf{private : member}
		\State \hspace{4mm} $\ct$: constant-degree rooted tree with height $\eta$ and $m$ leaves 
		\Comment leaf $i$ corresponds to $\vx_i$
		\State \hspace{4mm} $\sketchlen \defeq \Theta(\eta^{2}\log(\frac{m}{\rho}))$: sketch dimension
		\State \hspace{4mm} $\mphi \sim \mathbf{N}(0,\frac{1}{w})^{w \times m}$: JL-sketch matrix
		\State \hspace{4mm} $\delta>0$: additive approximation error
		\State \hspace{4mm} $k$: current IPM step
		\State \hspace{4mm} $\ox \in \R^{m}$: current valid approximate vector
		\State \hspace{4mm} $\{\vx^{(j)} \in \R^{m}\}_{j=0}^{k}$: list of previous inputs 
		\State \hspace{4mm} $\{\md^{(j)}\in\R^{m\times m}\}_{j=0}^{k}$: list of previous diagonal scaling matrices 
		
		\State
		\Procedure{\textsc{Initialize}}{$\ct, \vx\in\R^{m}, \md\in\R_{>0}^{m\times m}, \rho>0, \delta>0$}
		\State $\ct \leftarrow \ct$, $\delta\leftarrow\delta$, $k \leftarrow 0$
		\State $\ox\leftarrow\vx, \vx^{(0)} \leftarrow \vx, \md^{(0)} \leftarrow \md$
		\State sample $\mphi \sim \mathbf{N}(0,\frac{1}{w})^{w \times m}$
		\EndProcedure
		
		\State
		\Procedure{\textsc{Approximate}}{$\vx^\new \in\R^{m}, \md^\new \in\R_{>0}^{m\times m}$}
		\State $k\leftarrow k+1$, $\vx^{(k)}\leftarrow\vx^\new$, $\md^{(k)}\leftarrow\md^\new$
		\State $\ox_{i}\leftarrow\vx_{i}^{(k-1)}$ for all $i$ such that
		$\md_{ii}^{(k)}\neq\md_{ii}^{(k-1)}$ 
		\label{line:D-induced-ox-change}
		
		\State $I \leftarrow \emptyset$
		
		\ForAll {$0 \leq \ell < \left\lceil \log m\right\rceil $ such that $k \equiv 0\bmod2^{\ell}$}
		
		\State $I_{\ell}^{(k)} \leftarrow \textsc{FindLargeCoordinates}(\ell)$
		\label{line:set-I_ell^k}
		\State $I \leftarrow I\cup I_{\ell}^{(k)}$
		\EndFor
		
		\If{$k = 0 \bmod 2^{\left\lceil \log m\right\rceil}$} 
		\State $I\leftarrow[m]$ \label{line:set-ox-to-x}
		\Comment Update $\ox$ in full every $2^{\left\lceil \log m\right\rceil}$ steps
		\EndIf
		\State $\ox_{i}\leftarrow\vx_{i}^{(k)}$ for all $i\in I$ \label{line:update-ox}
		\State \Return $\ox$
		
		\EndProcedure
		\algstore{approx-vector-break-point}
	\end{algorithmic}
\end{algorithm}
\addtocounter{algorithm}{-1}

\begin{lemma}[Approximate Vector Maintenance] \label{lem:MaintainApprox}
	Suppose \textsc{FindLargeCoordinates}$(\ell)$ is a procedure in \textsc{AbstractMaintainApprox} that correctly computes the set $I_\ell^{(k)}$ at the $k$-th step.
	Then the deterministic data structure \textsc{AbstractMaintainApprox} in \cref{algo:maintain-vector} maintains an approximation $\ox$ of $\vx$ with the following procedures:
	\begin{itemize}
		\item \textsc{Initialize}$(\ct, \vx \in\R^{m}$, $\md \in \R_{> 0}^{m \times m}$, $\rho > 0$, $\delta>0)$: 
		Initialize the data structure at step 0 with tree ${\ct}$, initial vector $\vx$, initial diagonal scaling matrix $\md$, target additive approximation error $\delta$, and success probability $1 - \rho$. 
		\item \textsc{Approximate}$(\vx^\new \in\R^{m}$, $\md^\new \in\R_{> 0}^{m\times m}$): 
		Increment the step counter and update vector $\vx$ and diagonal scaling matrix $\md$. 
		Output a vector $\ox$ such that $\|\md^{1/2} (\vx-\ox)\|_{\infty}\leq\delta$ for the latest $\vx$ and $\md$.
	\end{itemize}
	Furthermore, if $\|\vx^{(k)}-\vx^{(k-1)}\|_{\md^{(k)}}\leq\beta$
	\emph{for all $k$}, then at the $k$-th step, the data structure first updates $\ox_i\leftarrow \vx_i^{(k-1)}$ for the coordinates $i$ with $\md_{ii}^{(k)}\neq\md_{ii}^{(k-1)}$, then updates $\ox_i\leftarrow \vx_i^{(k)}$ for $O(2^{2\ell_{k}}(\beta/\delta)^{2}\log^{2}m)$ coordinates, where $\ell_{k}$ is the largest integer $\ell$ with $k \equiv 0\bmod2^{\ell}$.
\end{lemma}

\begin{remark}
	In our problem setting of maintaining approximate flows and slacks, we do not have full access to the exact vector. 
	The algorithms in the next two subsections however will refer to the exact vector $\vx$ for readability and modularity. 
	We observe that access to $\vx$ is limited to two types: 
	accessing the JL-sketches of specific subvectors, and accessing exact coordinates and other specific subvectors of sufficiently small size.
	In later sections, we show how to implement these oracle accesses to $\vx$.
\end{remark}

\begin{proof}[Proof of \cref{lem:MaintainApprox}]
	We first prove the correctness of \textsc{Approximate} in $\textsc{AbstractMaintainApprox}$. 
	Fix some coordinate $i\in[m]$ and fix some IPM step $k$. 
	Suppose the latest update to $\ox_i$ is $\ox_i\leftarrow \vx_i^{(k')}$. This may happen in \cref{line:update-ox} at step $k'$ or in \cref{line:D-induced-ox-change} at step $k'+1$. In both case, we have that $\md_{ii}^{(d)}$ is the same for all $k\ge d >k'$ and that $i$ is not in the set $I_{\ell}^{(d)}$ returned by \textsc{FindLargeCoordinates} for all $k\ge d>k'$. (In the former case, we further have $\md_{ii}^{(k'+1)}=\md_{ii}^{(k')}$ but this is not required in the proof.)
	Since we set $\ox\leftarrow\vx$ every $2^{\left\lceil \log m\right\rceil }$
	steps by \cref{line:set-ox-to-x}, 
	we have $k-2^{\left\lceil \log m\right\rceil }\leq k'<k$.
	Using dyadic intervals, we can write $k'=k_{0}<k_{1}<k_{2}<\cdots<k_{s}=k$ such that $k_{j+1}-k_{j}$
	is a power of $2$, $k_{j+1}-k_{j}$ divides $k_{j+1}$, and $|s|\leq2\left\lceil \log m\right\rceil $.
	Hence, we have that 
	\[
	\vx_{i}^{(k)}-\ox_{i}^{(k)}=\vx_{i}^{(k_{s})}-\ox_{i}^{(k_{0})}=\vx_{i}^{(k_{s})}-\vx_{i}^{(k_{0})}=\sum_{j=0}^{s-1}(\vx_{i}^{(k_{j+1})}-\vx_{i}^{(k_{j})}).
	\]
	We know that $\md_{ii}^{(d)}$ is the same for all $k\ge d >k'$. By the guarantees of \textsc{FindLargeCoordinates}, we have
	\[
	\sqrt{\md_{ii}^{(k)}} \cdot|\vx_{i}^{(k_{j+1})}-\vx_{i}^{(k_{j})}|= \sqrt{\md_{ii}^{(k_{j+1})}} \cdot|\vx_{i}^{(k_{j+1})}-\vx_{i}^{(k_{j})}|\leq\frac{\delta}{2\left\lceil \log m\right\rceil }
	\] for all $0\le j<s$. (
	Summing over all $j=0,1,\ldots,s-1$ gives 
	\[
	\sqrt{\md_{ii}^{(k)}} \cdot|\vx_{i}^{(k)}-\ox_{i}^{(k)}|\leq\delta.
	\]
	Hence, we have $\|\md^{1/2} (\vx-\ox)\|_{\infty}\leq\delta$.
	
	Next, we bound the number of coordinates changed from $\ox^{(k-1)}$ to $\ox^{(k)}$. 
	Fix some $\ell$ with $k = 0 \bmod 2^\ell$.
	For any $i\in I_{\ell}^{(k)}$, we know $\md_{ii}^{(j)}=\md_{ii}^{(k)}$ for all $j>k-2^{\ell}$ because $\ox_{i}$ did not change in the meanwhile. By definition of $I_\ell^{(k)}$, we have
	\[
	\sqrt{\md_{ii}^{(k)}} \cdot\sum_{j=k-2^{\ell}}^{k-1}|\vx_{i}^{(j+1)}-\vx_{i}^{(j)}|\geq \sqrt{\md_{ii}^{(k)}} \cdot|\vx_{i}^{(k)}-\vx_{i}^{(k-2^{\ell})}|\geq\frac{\delta}{2\left\lceil \log m\right\rceil }.
	\]
	
	Using $\md_{ii}^{(j)}=\md_{ii}^{(k)}$ for all $j>k-2^{\ell}$ again, the above inequality yields 
	\begin{align*}
		\frac{\delta}{2\left\lceil \log m\right\rceil } &\leq\sum_{j=k-2^{\ell}}^{k-1} \sqrt{\md_{ii}^{(j+1)}}|\vx_{i}^{(j+1)}-\vx_{i}^{(j)}| \\
		&\leq \sqrt{2^{\ell}\sum_{j=k-2^{\ell}}^{k-1}\md_{ii}^{(j+1)}|\vx_{i}^{(j+1)}-\vx_{i}^{(j)}|^{2}}. \tag{by Cauchy-Schwarz}
	\end{align*}
	Squaring and summing over all $i\in I_{\ell}^{(k)}$ gives 
	\begin{align*}
		\Omega \left(\frac{2^{-\ell}\delta^{2}}{\log^{2}m}\right)|I_{\ell}^{(k)}| &\leq \sum_{i\in I_{\ell}^{(k)}}\sum_{j=k-2^{\ell}}^{k-1}\md_{ii}^{(j+1)}|\vx_{i}^{(j+1)}-\vx_{i}^{(j)}|^{2} \\
		&\leq
		\sum_{i=1}^{m}\sum_{j=k-2^{\ell}}^{k-1}\md_{ii}^{(j+1)}|\vx_{i}^{(j+1)}-\vx_{i}^{(j)}|^{2} \\
		&\leq2^{\ell}\beta^{2},
	\end{align*}
	where we use $\|\vx^{(j+1)}-\vx^{(j)}\|_{\md^{(j+1)}}\leq\beta$
	at the end. Hence, we have 
	\[
	|I_{\ell}^{(k)}|=O(2^{2\ell}(\beta/\delta)^{2}\log^{2}m).
	\]
	Recall this expression is for a fixed $\ell$. 
	At the $k$-th step, summing over all $\ell$ with $k = 0 \bmod 2^\ell$, we have that the total number of coordinates changed, excluding those induced by a change in $\md$, is
	\[
	\sum_{\ell=0}^{\ell_{k}}|I_{\ell}^{(k)}|=O(2^{2\ell_{k}}(\beta/\delta)^{2}\log^{2}m).
	\]
	
\end{proof}

\subsection{From change detection to sketch maintenance}
\label{subsec:sketch_change_to_sketch}

Now we discuss the implementation of \textsc{FindLargeCoordinates}$(\ell)$ to find the set $I_{\ell}^{(k)}$ in \cref{line:set-I_ell^k} of \cref{algo:maintain-vector}.
We accomplish this by repeatedly sampling a coordinate $i$ with probability proportional to $\md_{ii}^{(k)}\cdot|\vx_{i}^{(k)}-\vx_{i}^{(k-2^{\ell})}|^{2}$,
among all coordinates $i$ where $\ox_i$ has not been updated since $2^\ell$ steps ago.
With high probability, we can find all $i \in I_\ell^{(k)}$ in this way efficiently.
To implement the sampling procedure, we make use of a data structure based on segment trees~\cite{CLRS} along with sketching based on the Johnson-Lindenstrauss lemma.

Formally, we define the vector $\vq \in \R^m$ where $\vq_i \defeq {\md^{(k)}_{ii}}^{1/2} (\vx^{(k)}_i-\vx^{(k-2^{\ell})}_i)$ if $\ox_i$ has not been updated after the $k-2^\ell$-th step, and $\vq_i = 0$ otherwise. 
Our goal is precisely to find all large coordinates of $\vq$.

Let $\ct$ be a constant-degree rooted tree with $m$ leaves, where leaf $i$ represents coordinate $\vq_i$.
For each node $u \in \ct$, we define $E(u) \subseteq [m]$ to be set of indices of leaves in the subtree rooted at $u$.
We make a random descent down $\ct$, in order to sample a coordinate $i$ with probability proportional to $\vq_i^2$.
At a node $u$, for each child $u'$ of $u$, the total probability of the leaves under $u'$ is given precisely by $\norm{\vq|_{E(u')}}_2^2$. We can estimate this by the Johnson-Lindenstrauss lemma using a sketching matrix $\mphi$. Then we randomly move from $u$ down to child $u'$ with probability proportional to the estimated value.
To tolerate the estimation error, when reaching some leaf node representing coordinate $i$, we accept with probability proportional to the ratio between the exact probability of $i$ and the estimated probability of $i$. If $i$ is rejected, we repeat the process from the root again independently. 

\begin{algorithm}
	\caption{Data structure \textsc{AbstractMaintainApprox}, Part 2}
	\label{algo:change-detection}
	\begin{algorithmic}[1]
		\renewcommand{\thealgorithm}{}
		\algrestore{approx-vector-break-point}
		
		\Procedure{FindLargeCoordinates}{$\ell$}
		\State \Comment $\overline{\md}$ and $\vq$ are symbolic definitions
		\State \LeftComment $\overline{\md}$: diagonal matrix such that 
		\[
		\overline{\md}_{ii}=\begin{cases}
			\md_{ii}^{(k)} & \text{\text{if }}\text{$\ox_i$ has not been updated after the  $(k-2^{\ell})$-th step}\\
			0 & \text{otherwise.}
		\end{cases}
		\]
		\State \LeftComment $\vq \defeq \overline{\md}^{1/2} (\vx^{(k)}-\vx^{(k-2^{\ell})})$ \Comment{vector to sample coordinates from}
		\State
		\State $I \leftarrow \emptyset$ \Comment{set of candidate coordinates}
		\For{$N \defeq \Theta(2^{2\ell}(\beta/\delta)^{2}\log^{2}m\log(m/\rho))$ iterations}\label{line:outer-loop}
		
		\State \LeftComment Sample coordinate $i$ of $\vq$ w.p. proportional to $\vq_{i}^{2}$ by random descent down $\ct$ to a leaf
		\While{\textbf{true}} \label{line:sample-loop}
		\State $u\leftarrow\textrm{root}(\ct)$, $p_{u}\leftarrow1$ \label{line:sample-start}
		\While{$u$ is not a leaf node} 
		\State Sample a child $u'$ of $u$ with probability 
		\[
		\mathbf{P} (u\rightarrow u')\defeq\frac{\|\mphi_{E(u')} \vq\|_{2}^{2}}{\sum_{\text{child $u''$ of $u$}} \|\mphi_{E(u'')}\vq\|_{2}^{2}}
		\] \label{line:sample-a-child}
		\Comment let $\mphi_{E(u)} \defeq \mphi\mi_{E(u)}$ for each node $u$
		\State $p_{u}\leftarrow p_{u}\cdot \mathbf{P} (u\rightarrow u')$
		\State $u\leftarrow u'$
		\EndWhile 
		\State \textbf{break} with probability $p_{\mathrm{accept}}\defeq \norm{\vq|_{E(u)}}^{2} /(2\cdot p_{u}\cdot\|\mphi\vq\|_{2}^{2})$ \label{line:sample-filter} 
		\EndWhile
		
		\State $I\leftarrow I\cup E(u)$ \label{line:sample-one-cor}
		\EndFor
		
		\State \Return $\{i\in I \;:\; \vq_{i} \geq\frac{\delta}{2\left\lceil \log m\right\rceil }\}$.
		\EndProcedure
	\end{algorithmic}
\end{algorithm}

\begin{lemma} 
	\label{lem:change-detection}
	Assume that $\|\vx^{(k+1)}-\vx^{(k)}\|_{\md^{(k+1)}}\leq\beta$ for all IPM steps $k$. 
	Let $\rho < 1$ be any given failure probability, and let $N \defeq \Theta(2^{2\ell}(\beta/\delta)^{2}\log^{2}m\log(m/\rho))$ be the number of samples \cref{algo:change-detection} takes.
	Then with probability $\geq1-\rho$, during the $k$-th call of \textsc{Approximate},  \cref{algo:change-detection} finds the set $I_{\ell}^{(k)}$ correctly. 
	Furthermore, the while-loop in \cref{line:sample-loop} happens only $O(1)$ times in expectation per sample.  
\end{lemma}

\begin{proof}The proof is similar to Lemma 6.17 in \cite{treeLPArxivV2}. We
	include it for completeness. 
	For a set $S$ of indices, let $\mi_{S}$ be the $m \times m$ diagonal matrix that is one on $S$ and zero otherwise. 
	
	We first prove that \cref{line:sample-filter} breaks with probability
	at least $\frac{1}{4}$. By the choice of $\sketchlen$, Johnson--Lindenstrauss
	lemma shows that $\|\mphi_{E(u)}\vq\|_{2}^{2}=(1\pm\frac{1}{9\eta})\|\mi_{E(u)}\vq\|_{2}^{2}$
	for all $u \in \ct$ with probability at least $1-\rho$. 
	Therefore, the probability we move from a node $u$ to its child node $u'$ is given by
	\[
	\mathbf{P} (u\rightarrow u')= \left(1\pm\frac{1}{3\eta} \right) \frac{\|\mi_{E(u')}\vq\|_{2}^{2}}{\sum_{u''\text{ is a child of }u}\|\mi_{E(u'')}\vq\|_{2}^{2}}
	=
	\left(1\pm\frac{1}{3\eta}\right)\frac{\|\mi_{E(u')}\vq\|_{2}^{2}}{\|\mi_{E(u)}\vq\|_{2}^{2}}.
	\]
	Hence, the probability the walk ends at a leaf $u \in \ct$ is given by
	\[
	p_{u}= \left(1\pm\frac{1}{3\eta} \right)^{\eta} \frac{\|\mi_{u}\vq\|_{2}^{2}}{\|\vq\|_{2}^{2}}=(1\pm\frac{1}{3\eta})^{\eta}\frac{\norm{\vq|_{E(u)}}^{2}}{\|\vq\|_{2}^{2}}.
	\]
	Now, $p_{\mathrm{accept}}$ on \cref{line:sample-filter} is at least
	\[
	p_{\mathrm{accept}}=\frac{\norm{\vq|_{E(u)}}^{2}} {2\cdot p_{u}\cdot\|\mphi\vq\|_{2}^{2}}
	\geq 
	\frac{\norm{\vq|_{E(u)}}^{2}} {2\cdot(1+\frac{1}{3\eta})^{\eta} \frac{\norm{\vq|_{E(u)}}^{2}}{\|\vq\|_{2}^{2}}\cdot\|\mphi\vq\|_{2}^{2}}
	\geq
	\frac{\|\vq\|_{2}^{2}}{2\cdot(1+\frac{1}{3\eta})^{\eta}\|\mphi\vq\|_{2}^{2}}
	\geq 
	\frac{1}{4}.
	\]
	On the other hand, we have that $p_{\mathrm{accept}}\leq\frac{\|\vq\|_{2}^{2}}{2(1-\frac{1}{3\eta})^{\eta}\|\mphi\vq\|_{2}^{2}}<1$
	and hence this is a valid probability.
	
	Next, we note that $u$ is accepted on \cref{line:sample-filter} with probability
	\[
	p_{\mathrm{accept}} p_{u}=\frac{\norm{\vq|_{E(u)}}^{2}}{2\cdot\|\mphi\vq\|_{2}^{2}}.
	\]
	Since $\|\mphi\vq\|_{2}^{2}$ remains the same in all iterations, this probability is proportional to $\norm{\vq|_{E(u)}}^{2}$.
	Since the algorithm repeats when $u$ is rejected, on \cref{line:sample-one-cor},
	$u$ is chosen with probability exactly $\norm{\vq|_{E(u)}}^{2}/\|\vq\|^{2}$.
	
	Now, we want to show the output set is exactly $\{i \in [n] :|\vq_{i}|\geq\frac{\delta}{2\left\lceil \log m\right\rceil }\}$.
	Let $S$ denote the set of indices where $\ox$ did not update between the $(k-2^\ell)$-th step and the current $k$-th step.
	Then
	\begin{align*}
		\|\vq\|_{2} &=\|\mi_{S}(\md^{(k)})^{1/2}(\vx^{(k)}-\vx^{(k-2^{\ell})})\|_{2}\\
		& \leq\sum_{i=k-2^{\ell}}^{k-1}\|\mi_{S}(\md^{(k)})^{1/2}(\vx^{(i+1)}-\vx^{(i)})\|_{2}\\
		& =\sum_{i=k-2^{\ell}}^{k-1}\|\mi_{S}(\md^{(i+1)})^{1/2}(\vx^{(i+1)}-\vx^{(i)})\|_{2}\\
		& \leq\sum_{i=k-2^{\ell}}^{k-1}\|(\md^{(i+1)})^{1/2}(\vx^{(i+1)}-\vx^{(i)})\|_{2} \\
		&\leq 2^{\ell}\beta,
	\end{align*}
	where we used $\mi_S \md^{(i+1)} = \mi_S \md^{(k)}$, because $\ox_{i}$ changes whenever $\md_{ii}$ changes at a step. Hence,
	each leaf $u$ is sampled with probability at least $\norm{\vq|_{E(u)}}^{2}/(2^{\ell}\beta)^{2}$.
	If $|\vq_{i}|\geq\frac{\delta}{2\left\lceil \log m\right\rceil }$,
	and $i \in E(u)$ for a leaf node $u$, then the coordinate $i$ is not in $I$ with probability at most
	\[
	\left(1-\frac{\norm{\vq|_{E(u)}}^{2}}{(2^{\ell}\beta)^{2}} \right)^{N}
	\leq \left(1-\frac{1}{2^{2\ell+2}(\beta/\delta)^{2}\left\lceil \log m\right\rceil ^{2}}\right)^{N}\leq\frac{\rho}{m},
	\]
	by our choice of $N$. Hence, all $i$ with $|\vq_{i}|\geq\frac{\delta}{2\left\lceil \log m\right\rceil }$
	lies in $I$ with probability at least $1-\rho$. This proves that
	the output set is exactly $I_{\ell}^{(k)}$ with probability at least
	$1-\rho$. 
\end{proof} 

\begin{remark}
	
	In \cref{algo:change-detection}, we only need to compute $\|\mphi_{E(u)}\vq\|_{2}^{2}$ for $O(N)$ many nodes $u \in \ct$. 
	Furthermore, 
	the randomness of the sketch is not leaked and we can use the same
	random sketch $\mphi$ throughout the algorithm. 
	This allows us to efficiently maintain $\mphi_{E(u)}\vq$ for each $u \in \ct$ throughout the IPM.
	
\end{remark}
	
\subsection{Sketch maintenance} \label{subsec:sketch_maintenance}

In \textsc{FindLargeCoordinates} in the previous subsection, we assumed the existence of a constant degree tree $\ct$, 
and for the dynamic vector $\vq$, 
the ability to access $\mphi_{E(u)} \vq$ at each node $u \in \ct$ and $\vq|_{E(u)}$ at each leaf node $u \in \ct(0)$.

In this section, we consider when the required tree is the separator tree $\ct$ of the overall input graph, and
the vector $\vq$ is of the form $\vq = \vy + \mm \vz$,
where $\mm$ is a tree operator supported on $\ct$, and each of $\vy, \mm, \vz$  undergo changes at every IPM step. 
We present a data structure that implements two features efficiently on $\ct$:
\begin{itemize}
	\item access $(\vy + \mm \vz)|_{E(H)}$ at every leaf node $H$, where $E(H) \defeq \range{\mj_H}$.
	\item access $\mphi_{E(H)} (\vy + \mm \vz)$ at every node $H$, where $\mphi_{E(H)}$ is $\mphi$ restricted to columns given by $E(H) \defeq \bigcup_{\text{leaf $D\in \ct_H$}} E(D)$.
\end{itemize}

\begin{remark}
	As seen in the pseudocode, sketches for $\vy$ and $\mm \vz$ can be maintained separately. We collected them together to represent $\vx$ as a whole for simplicity.
\end{remark}

First, we present some lemmas about the structure of the expression $\mm \vz$ which will help us to implement the requirements above. For any node $H \in \ct$, let $\ct_H$ be the subtree of $\ct$ rooted at $H$.

\begin{lemma} \label{lem:query-correctness}
	At any leaf node $H \in \ct(0)$, we have
	\[
	(\mm \vz)|_{E(H)} = \sum_{A : H \in \ct_A} \mj_H \mm_{H \leftarrow A} \mi_{F_A} \vz 
	= \mj_H \mi_{F_H} \vz + \sum_{\text{ancestor $A$ of $H$}} \mj_H \mm_{H \leftarrow A} \mi_{F_A} \vz. 
	\]
\end{lemma}

\begin{proof}
	Recall from the definition of the tree operator that $\range{\mj_H}$ are disjoint. So to get $(\mm \vz)|_{E(H)}$, it suffices to only consider the terms corresponding to the leaf $H$ in the expression \cref{defn:mm} for $\mm$; this gives the first equality. The second equality simply splits the sum into two parts. (We do not consider a node to be its own ancestor.)
\end{proof}

\begin{lemma} \label{lem:estimate-correctness}
	At any node $H \in \ct$, we have
	\[
	\mphi_{E(H)} \mm \vz = \mphi \overline{\mm^{(H)}} \vz + \mphi \mm^{(H)} \sum_{\text{ancestor $A$ of $H$}} \mm_{H \leftarrow A} \mi_{F_A} \vz.
	\]
\end{lemma}
Intuitively, the lemma shows that the sketch of $\mm \vz$ restricted to $E(H)$ can be split into two parts. 
The first part involves some sum over all nodes in $\ct_H$, ie. descendants of $H$ and $H$ itself, and the second part involves a sum over all ancestors of $H$. 
\begin{proof}
	First, note that since $\mphi$ is restricted to $E(H)$, it suffices to consider the terms in the sum for $\mm$ that map into to $E(H)$. In particular, this is the set of leaf nodes $\ct_H$ in the subtree rooted at $H$.
	\[
	\mphi_{E(H)} \mm \vz = \mphi \sum_{\text{leaf } D \in \ct_H} \sum_{A : D \in \ct_A} \mj_D \mm_{D \leftarrow A} \mi_{F_A} \vz.
	\]
	The right hand side involves a sum over the set $\{ (D, A) \;:\; D \in \ct_H \text{ is a leaf node}, D \in \ct_A\}$. Observe that $(D,A)$ is in this set if and only if $A \in \ct_H$ or $A$ is an ancestor of $H$. Hence, the summation can be written as
	\[
	\sum_{\text{leaf $D \in \ct_H$}} \sum_{A \in \ct_H} \mj_D \mm_{D \leftarrow H}\mi_{F_H}\vz + 
	\sum_{\text{leaf $D \in \ct_H$}} \sum_{\text{ancestor $A$ of $H$}} \mj_{D} \mm_{D \leftarrow A} \mi_{F_A} \vz.
	\]
	The first term is precisely $ \overline{\mm^{(H)}} \vz$. For the second term, we can use the fact that $A$ is an ancestor of $H$ to expand $\mm_{D \leftarrow A} = \mm_{D \leftarrow H} \mm_{H \leftarrow A}$. Then, the second term is
	\begin{align*}
		&\phantom{{}={}} \sum_{\text{leaf $D \in \ct_H$}} \sum_{\text{ancestor $A$ of $H$}} \mj_{D} \mm_{D \leftarrow H} \mm_{H \leftarrow A} \mi_{F_A} \vz \\
		&= \sum_{\text{leaf $D \in \ct_H$}} \mj_D \mm_{D \leftarrow H} \left( \sum_{\text{ancestor $A$ of $H$}} \mm_{H \leftarrow A} \mi_{F_A} \vz \right) \\
		&= \mm^{(H)} \left(\sum_{\text{ancestor $A$ of $H$}} \mm_{H \leftarrow A} \mi_{F_A} \vz \right),
	\end{align*}
	by definition of $\mm^{(H)}$.
\end{proof}

\begin{algorithm}
	\caption{Data structure for maintaining $\mphi (\vy + \mm \vz)$, Part 1}
	\label{algo:maintain-sketch}
	\begin{algorithmic}[1]
		\State \textbf{data structure} \textsc{MaintainSketch} 
		\State \textbf{private : member}
		\State \hspace{4mm} $\ct$ : rooted constant degree tree, where at every node $H$, there is
		\State \hspace{12mm} $\boxed{\mphi \mm^{(H)}}$ : sketch of partial tree operator
		\State \hspace{12mm} $\boxed{\mphi\overline{\mm^{(H)}} \vz}$ : sketched vector 
		\Comment This gives $\mphi \mm \vz$ at the root
		\State \hspace{12mm} $\boxed{\mphi\vy|_{E(H)}}$ : sketched subvector of $\vy$ 
		
		\State \hspace{4mm} $\mphi\in\R^{\sketchlen\times m}$ : JL-sketch matrix
		\State \hspace{4mm} $\mm$ : tree operator on $\ct$
		\State \hspace{4mm} $\vz\in\R^n$ : vector $\vz$
		\State \hspace{4mm} $\vy\in\R^n$ : vector $\vy$ 
		\Comment $\mm,\vz,\vy$ are pointers to read-only memory
		
		\State
		\Procedure{$\textsc{Initialize}$}{rooted tree $\ct$, $\mphi\in\R^{\sketchlen\times m}$, tree operator $\mm$, $\vz$, $\vy$}
		\State $\mphi \leftarrow\mphi$, $\ct \leftarrow \ct$
		\State $\boxed{\mphi \mm^{(H)}} \leftarrow \vzero$, $\boxed{\mphi\overline{\mm^{(H)}} \vz} \leftarrow \vzero$, 
		$\boxed{\mphi \vy|_{E(H)}} \leftarrow \vzero$ for all $H \in \ct$
		\State {$\textsc{Update}$}({$\mm, \vz, \vy, V(\ct)$}) 
		\EndProcedure
		
		\State
		\Procedure{$\textsc{Update}$}{$\mm^\new, \vz^\new, \vy^\new, \mathcal{S} \defeq \text{set of nodes admitting changes}$}
		\State $\mm \leftarrow \mm^\new$, $\vz \leftarrow \vz^\new$, $\vy\leftarrow \vy^\new$
		\For{$H \in \mathcal{P}_{\ct}(\mathcal{S})$ by increasing node level}
		\If {$H$ is a leaf}
		\State $\boxed{\mphi\mm^{(H)}} \leftarrow \mphi \mj_H$ \label{line:collectl}
		\State $\boxed{\mphi\overline{\mm^{(H)}}\vz} \leftarrow \mphi \mj_H \vz|_{F_H}$ \label{line:collectl2}
		\State $\boxed{\mphi\vy|_{E(H)}} \leftarrow \mphi\vy|_{E(H)}$
		\Else
		\State $\boxed{\mphi\mm^{(H)}} \leftarrow \sum_{\text{child $D$ of $H$}}\boxed{\mphi\mm^{(D)}} \mm_{(D, H)}$ \label{line:collect}
		\State $\boxed{\mphi\overline{\mm^{(H)}}\vz} \leftarrow \boxed{\mphi\mm^{(H)}}\vz|_{F_H}+\sum_{\text{child  $D$ of $H$}} \boxed{\mphi\overline{\mm^{(D)}}\vz}$ \label{line:collect2}
		\State $\boxed{\mphi\vy|_{E(H)}} \leftarrow \sum_{\text{child $D$ of $H$}}\boxed{\mphi\vy|_{E(D)}}$ \label{line:collecty}
		\EndIf
		\EndFor
		\EndProcedure
		
		\State
		\Procedure{$\textsc{SumAncestors}$}{$H \in \ct$}
		\If {\textsc{Update} has not been called since the last call to \textsc{SumAncestors}$(H)$}
		\State \Return the result of the last \textsc{SumAncestors}$(H)$
		\EndIf
		\If {$H$ is the root} 
		\Return $\vzero$
		\EndIf
		\State \Return $\mm_{(H,P)} (\vz|_{F_P} + \textsc{SumAncestors}(P))$ 
		\Comment $P$ is the parent of $H$
		\EndProcedure
		\algstore{sketch-maintain-break-point}
	\end{algorithmic}
\end{algorithm}
\addtocounter{algorithm}{-1}

\begin{algorithm}
	\caption{Data structure for maintaining $\mphi(\vy + \mm \vz)$, part 2}
	\begin{algorithmic}[1]
		\renewcommand{\thealgorithm}{}
		\algrestore{sketch-maintain-break-point}
		
		\Procedure{$\textsc{Estimate}$}{$H \in \ct$}
		\State Let $\vu$ be the result of \textsc{SumAncestors}$(H)$
		\State \Return $\boxed{\mphi\mm^{(H)}}\vu+\boxed{\mphi\overline{\mm^{(H)}}\vz} + \boxed{\mphi\vy|_{E(H)}}$
		\EndProcedure
		\State
		\Procedure{$\textsc{Query}$}{leaf $H \in \ct$}
		\State \Return $\vy|_{E(H)} + \mj_H (\vz|_{F_H} + \textsc{SumAncestors}(H))$
		\EndProcedure
	\end{algorithmic}
\end{algorithm}

\begin{lemma} \label{lem:maintain-sketch}
	Let $\ct$ be a rooted tree with height $\eta$ supporting tree operator $\mm$ with complexity $T$.
	Let $\sketchlen = \Theta(\eta^{2}\log(\frac{m}{\rho}))$ be as defined in \cref{algo:maintain-vector},
	and let $\mphi\in\R^{\sketchlen\times m}$ be a JL-sketch matrix.
	Then \textsc{MaintainSketch} (\cref{algo:maintain-sketch}) is a data structure that maintains $\mphi (\vy + \mm \vz)$, as $\vy$, $\mm$ and $\vz$ undergo changes in the IPM.
	The data structure supports the following procedures: 
	\begin{itemize}
		\item \textsc{Initialize}(rooted tree $\ct$, $\mphi \in\R^{\sketchlen\times m}$, tree operator $\mm^{\init} \in \R^{m \times n}$, $\vz^{\init} \in \R^n$, $\vy^{\init} \in \R^m$):
		Initialize the data structure with tree operator $\mm \leftarrow \mm^{\init}$, and vectors $\vz \leftarrow \vz^{\init}$, $\vy \leftarrow \vy^{\init}$, and compute the initial sketches in $O(w \cdot m)$ time.
		\item $\textsc{Update}(\mm^\new, \vz^\new, \vy^\new)$:
		Update $\mm \leftarrow \mm^\new$, $\vz \leftarrow \vz^\new$, $\vy\leftarrow \vy^\new$ and all the necessary sketches in $O(\sketchlen\cdot T(\eta \cdot |\mathcal{S}|))$ time, where
		$\mathcal{S}$ is the set of all nodes $H$ where one of $\mm_{(H,P)}, \mj_H, \vz|_{F_H}, \vy|_{E(H)}$ is updated.
		
		\item $\textsc{SumAncestors}(H \in \ct)$: Return $\sum_{\text{ancestor $A$ of $H$}} \mm_{H \leftarrow A} \mi_{F_A} \vz$.
		\item $\textsc{Estimate}(H \in \ct)$: Return $\mphi_{E(H)} \left(\vy + \mm \vz \right)$.
		\item $\textsc{Query}(H \in \ct)$: Return $(\vy + \mm \vz)|_{E(H)}$.
	\end{itemize}

	If we call $\textsc{Query}$ on $N$ nodes, the total runtime is $O(\sketchlen \cdot T(\eta N))$.

	If we call $\textsc{Estimate}$ along a sampling path (by which we mean starting at the root, calling estimate at both children of a node, and then recursively descending to one child until reaching a leaf), and then we call $\textsc{Query}$ on the resulting leaf, and we repeat this $N$ times with no updates during the process,
	then the total runtime of these calls is $O(\sketchlen\cdot T(\eta N))$.
\end{lemma}

\begin{proof}
	First, we note that each edge operator $\mm_e$ should be stored implicitly. 
	In particular, it suffices to only support the operation of computing $\vu^{\top}\mm_{e}$ and $\mm_{e}\vx$ for any vectors $\vu$ and $\vx$. 
	
	We prove the running time and correctness for each procedure.
	
	\paragraph{\textsc{Initialize}:}
	It sets the sketches to $\vzero$ in $O(\sketchlen \cdot m)$ time. 
	It then calls \textsc{Update} with the initial $\mm$, $\vz$, $\vy$, and updates the sketches everywhere on $\ct$. 
	By the runtime and correctness of \textsc{Update}, this step is correct and runs in $\O(w \cdot T(m))$ time.
	
	\paragraph{\textsc{Update}$(\mm^\new, \vz^\new, \vy^\new)$:}
	Let $\mathcal{S}$ denote the set of nodes admitting changes as defined in the theorem statement.
	If a node $H$ is not in $\mathcal{S}$ and it has no descendants in $\mathcal{S}$, then by definition, $\mm^{(H)}$ and $\overline{\mm^{(H)}} \vz$ are not affected by the updates in $\mm$ and $\vz$. Similarly, in this case, $\vy|_{E(H)}$ is not affected by the updates to $\vy$. 
	Hence, it suffices to update the sketches only at all nodes in $\mathcal{P}_{\ct}(\mathcal{S})$. 
	We update the nodes from the bottom level of the tree upwards, so that when we're at a node $H$, all the sketches at its descendant nodes are correct. Hence, by definition, the sketch at $H$ is also correct.
	
	To compute the runtime, first note $|\mathcal{P}_{\ct}(\mathcal{S})| = O(\eta |\mathcal{S}|)$, 
	since for each node $H \in \mathcal{S}$, the set includes all the $O(\eta)$ nodes on the path from $H$ to the root. 
	For each leaf node $H \in \mathcal{P}_\ct(\mathcal{S})$, we can compute its sketches in constant time.
	For each non-leaf node $H \in \mathcal{S}$ with children $D_1, D_2$, 
	\cref{line:collect} multiplies each row of $\boxed{\mphi\mm^{(D_1)}}$ with $\mm_{(D_1, H)}$,  
	each row of $\boxed{\mphi\mm^{(D_2)}}$ with $\mm_{(D_2, H)}$, and sums the results.
	For a fixed row number, the total time over all $H \in \mathcal{P}_{\ct}(\mathcal{S})$ is bounded by $O(T(|\mathcal{P}_\ct(\mathcal{S})|))$.
	So the total time for \cref{line:collect} in the procedure is $O(\sketchlen \cdot T(\eta |S|))$.
	
	\cref{line:collect2} multiply each row of $\boxed{\mphi\mm^{(H)}}$ with a vector and then performs a constant number of additions of $O(\sketchlen)$-length vectors. Since $\boxed{\mphi\mm^{(H)}}$ is computed for all $H \in T(|\mathcal{P}_\ct(\mathcal{S})|)$ in $O(\sketchlen \cdot T(\eta |S|))$ total time, this runtime must also be a bound on the number of total non-zero entries. Since each $\boxed{\mphi\mm^{(H)}}$ is used once in \cref{line:collect2} for a matrix-vector multiplication, the total runtime over all $H$ is also $O(\sketchlen \cdot T(\eta |S|))$. Lastly, the vector additions across all $H$ takes $O(\sketchlen \cdot \eta |S|)$ time.
	
	\cref{line:collecty} adds two vectors of length $\sketchlen$. This is not the bottleneck.
	
	\paragraph{\textsc{SumAncestors}$(H)$:}
	At the root, there are no ancestors, hence we return the zero matrix. 
	When $H$ is not the root, suppose $P$ is the parent of $H$. Then we can recursively write 
	\[
	\sum_{\text{ancestor $A$ of $H$}} \mm_{H \leftarrow A} \mi_{F_A} \vz = \mm_{(H, P)} \left(\mi_{F_P} \vz + \sum_{\text{ancestor $A$ of $P$}} \mm_{P \leftarrow A} \mi_{F_A} \vz \right).
	\]
	The procedure implements the right hand side, and is therefore correct.
	
	\paragraph{\textsc{Estimate} and \textsc{Query}:} Their correctness follow from \cref{lem:estimate-correctness,lem:query-correctness}, and the correctness of \boxed{\mphi\vy|_{E(H)}} maintained by \textsc{Update}.
	
	\paragraph{Overall \textsc{Estimate} and \textsc{Query} time along $N$ sampling paths:}
	We show that if we call \textsc{Estimate} along $N$ sampling paths
	each from the root to a leaf, and we call $\textsc{Query}$
	on the leaves, the overall cost for these calls is $O(\sketchlen\cdot T(\eta N))$:
	
	Suppose the set of nodes visited is given by $\mathcal{H}$, then $|\mathcal{H}| \leq \eta N$.
	Since there is no update, and \textsc{Estimate} is called for a node only after it is called for its parent, 
	we know that $\textsc{SumAncestors}(H)$ is called exactly once for each $H \in \mathcal{H}$.
	Each $\textsc{SumAncestor}(H)$ multiplies a unique edge operator $\mm_{(H,P)}$ with a vector. 
	Hence, the total runtime of \textsc{SumAncestors} is $T(|\mathcal{H}|)$.
	Furthermore, the total number of non-zero entries of the return values of these \textsc{SumAncestors} is also $O(T(|\mathcal{H}|))$.
	
	Finally, each \textsc{Query} applies a constant-time operator $\mj_H$ to the output of a unique \textsc{SumAncestors} call, so the overall runtime is certainly bounded by $O(T(|\mathcal{H}|))$. Adding a constant-sized $\vy|_{E(H)}$ can be done efficiently.
	Similarly, each \textsc{Estimate} multiplies $\boxed{\mphi\mm^{(H)}}$ with the output of a unique \textsc{SumAncestors} call.
	This can be computed as $\sketchlen$-many vectors each multiplied with the \textsc{SumAncestors} output. Then two vectors of length $\sketchlen$ are added. Summing over all $H \in \mathcal{H}$, the overall runtime is $O(\sketchlen \cdot T(|\mathcal{H}|)) = O(\sketchlen \cdot T(\eta N))$.
	
	\paragraph{\textsc{Query} time on $N$ leaves:}
	Since this is a subset of the work described above, the runtime must also be bounded by $O(\sketchlen \cdot T(\eta N))$.
	
\end{proof}

\subsection{Proof of \crtcref{thm:VectorTreeMaintain}}
\label{subsec:sketch_final_proof}

We combine the previous three subsections for the overall approximation procedure.
It is essentially \textsc{AbstractMaintainApprox} in \cref{algo:maintain-vector}, 
with the abstractions replaced by a data structure implementation.
We did not provide the corresponding pseudocode. 

\vectorTreeMaintain*

\begin{proof}
The data structure \textsc{AbstractMaintainApprox} in \cref{algo:maintain-vector}
performs the correct vector approximation maintenance, however, it is not completely implemented.
\textsc{MaintainApprox} simply replaces the abstractions with a concrete implementation using the data structure \textsc{MaintainSketch} from \cref{algo:maintain-sketch}.

First, for notation purposes, let $\vz \defeq c \zprev + \zsum$, and let $\vx \defeq \vy + \mm \vz$, so that at step $k$,  \textsc{Approximate} procedure has $\vx^{(k)}$ (in implicit form) as input, and return $\ox$. 

Let $\ell \in \{1, \dots, O(\log m)\}$. We define a new dynamic vector $\vx_{\ell}$ \emph{symbolically}, which is represented at each step $k$ for $k \geq 2^\ell$ by
\[
\vx^{(k)}_{\ell} \defeq \vy^{(k)}_{\ell}  + \mm^{(k)}_{\ell}  \vz^{(k)}_{\ell},
\]
where the new tree operator $\mm_{\ell}$ at step $k$ is given by
\begin{itemize}
	\item ${\mm^{(k)}_{\ell}}_{(H, P)}=\diag\left(\mm^{(k)}_{(H, P)}, \mm^{(k-2^\ell)}_{(H, P)}\right)$ for each child-parent edge $(H,P)$ in $\ct$,
	\item ${\mj^{(k)}_{\ell}}_H=\overline{\md}_{E(H), E(H)}\left[\mj^{(k)}_H~\mj^{(k-2^\ell)}_H\right]$ for each leaf node $H \in \ct$,
\end{itemize}
where $\overline{\md}$ is the diagonal matrix defined in \textsc{FindLargeCoordinates}, with $\overline{\md}_{i,i} = \md^{(k)}_{i,i}$ at step $k$ if $\ox_i$ has not been updated after step $k-2^\ell$, and zero otherwise.

At step $k$, the vector $\vy_{\ell}$ is given by $\vy_{\ell}^{(k)}= \overline{\md}^{1/2} \left( \vy^{(k)}-\vy^{(k-2^\ell)}\right)$,
and $\vz_{\ell}$ by $\vz_{\ell}^{(k)} \defeq \left[\vz^{(k)} ~ \vz^{(k-2^\ell)}\right]^\top$.
Then, at each step $k$ with $k \geq 2^\ell$, we have
\begin{align} \label{eq:x_ell^k}
\vx_{\ell}^{(k)} &\defeq \vy_{\ell}^{(k)} + \mm_{\ell}^{(k)} \vz_{\ell}^{(k)} \\
&= \left(\overline{\md}^{1/2}\vy^{(k)} + \overline{\md}^{1/2}\mm^{(k)}\vz^{(k)}\right) - \left(
 \overline{\md}^{1/2} \vy^{(k-2^\ell)} +\overline{\md}^{1/2}\mm^{(k-2^\ell)} \vz^{(k-2^\ell)} \right) \nonumber \\
&= \overline{\md}^{1/2} (\vx^{(k)} - \vx^{(k - 2^\ell)}) \nonumber.
\end{align}
Note this is precisely the vector $\vq$ for a fixed $\ell$ in \textsc{FindLargeCoordinates} in \cref{algo:maintain-vector}.
It is straightforward to see that $\mm_{\ell}$ indeed satisfies the definition of a tree operator.
Furthermore, $\mm_{\ell}$ has the same complexity as $\mm$.
\textsc{MaintainApprox} will contain $O(\log m)$ copies of the \textsc{MaintainSketch} data structures in total, where the $\ell$-th copy sketches $\vx_{\ell}$ as it changes throughout the IPM algorithm.

We now describe each procedure in words, and then prove their correctness and runtime. 
\paragraph{\textsc{Initialize}$(\ct, \mm, c, \zprev, \zsum, \vy, \md, \rho, \delta)$:}

This procedure implements the initialization of \textsc{AbstractMaintainApprox},
where the dynamic vector $\vx$ to be approximated is represented by $\vx \defeq \vy + \mm (c \zprev + \zsum)$.
The initialization steps described in \cref{algo:maintain-vector} takes $O(\sketchlen m)$ time.
Let $\mphi$ denote the JL-sketching matrix. 

We initialize two copies of the \textsc{MaintainSketch} data structure, $\texttt{ox\_cur}$ and $\texttt{ox\_prev}$. At step $k$, $\texttt{ox\_cur}$ will maintain sketches of $\mphi \vx^{(k)}$, and $\texttt{ox\_prev}$ will maintain sketches of $\mphi \vx^{(k-1)}$.
(The latter is initialized at step $1$, but we consider it as part of initialization.)

In addition, for each $0 \leq \ell \leq O(m)$, we initialize a copy $\texttt{sketch}_{\ell}$ of  \textsc{MaintainSketch}.
These are needed for the implementation of $\textsc{FindLargeCoordinates}(\ell)$ in \textsc{Approximate}.
Specifically, at step $k = 2^\ell$ of the IPM, we initialize $\texttt{sketch}_{\ell}$ by calling 
$\texttt{sketch}_{\ell}.\textsc{Initialize}(\ct, \mphi, \mm_{\ell}^{(k)}, \vz_{\ell}^{(k)}, \vy_{\ell}^{(k)})$.
(Although this occurs at step $k > 0$, we charge its runtime according to its function as part of initialization.)

The total initialization time is $O(wm\log m)=O(m\eta^2\log m\log (\frac{m}{\rho}))$ by \cref{lem:maintain-sketch}. 
By the existing pseudocode in \cref{algo:maintain-vector}, it correctly initializes $\ox \leftarrow \vx$.

\paragraph{\textsc{Approximate}$(\mm^\new, c^\new, \zprev^\new, \zsum^\new, \vy^\new, \md^\new)$:}
This procedure implements \textsc{Approximate} in \cref{algo:maintain-vector}. 
We consider when the current step is $k$ below.

First, we update the sketch data structures $\texttt{sketch}_{\ell}$ for each $\ell$ 
by calling $\texttt{sketch}_{\ell}.\textsc{Update}$.
Recall at step $k$, $\texttt{sketch}_{\ell}$ maintains sketches for the vector $\vx_{\ell}^{(k)} =  \overline{\md}^{1/2} (\vx^{(k)} - \vx^{(k-2^\ell)})$,
although the actual representation in $\texttt{sketch}_{\ell}$ of the vector $\vx_{\ell}$ is given by $\vx_{\ell} = \vy_{\ell} + \mm_{\ell} \vz_{\ell}$ as defined in \cref{eq:x_ell^k}.

Next, we execute the pseudocode given in \textsc{Approximate} in \cref{algo:maintain-vector}:

To update $\ox_e$ to $\vx^{(k-1)}_e$ for a single coordinate (\cref{line:D-induced-ox-change} of \cref{algo:maintain-vector}), 
we find the leaf node $H$ containing the edge $e$, and call $\texttt{ox\_prev}.\textsc{Query}(H)$. 
This returns the subvector $\vx^{(k-1)}|_{E(H)}$, from which we can make the assignment to $\ox_e$.
To update $\ox_e$ to $\vx^{(k)}_e$ for single coordinates (\cref{line:update-ox} of \cref{algo:maintain-vector}), 
we do the same as above, except using the data structure $\texttt{ox\_cur}$.

In the subroutine \textsc{FindLargeCoordinates}$(\ell)$, 
the vector $\vq$ defined in the pseudocode is exactly $\vx_\ell^{(k)}$.
We get the value of $\mphi_{E(u)} \vq$ at a node $u$ by calling $\texttt{sketch}_{\ell}.\textsc{Estimate}(u)$,
and we get the value of $\vq|_{E(u)}$ at a leaf node $u$ by calling
$\texttt{sketch}_{\ell}.\textsc{Query}(u)$.

\subparagraph{Number of coordinates changed in $\ox$ during \textsc{Approximate}.}

In \cref{line:D-induced-ox-change} of \textsc{Approximate} in \cref{algo:maintain-vector},
$\ox$ is updated in every coordinate $e$ where $\md_e$ differs compared to the previous step.

Next, the procedure collect a set of coordinates for which we update $\ox$,
by calling \textsc{FindLargeCoordinates}$(\ell)$ for each $0 \leq \ell \leq \ell_k$,
where $\ell_k$ is defined to be the number of trailing zeros in the binary representation of $k$. 
(These are exactly the values of $\ell$ such that $k \equiv 0 \mod 2^\ell$).
In each call of \textsc{FindLargeCoordinates}$(\ell)$, 
There are $O(2^{2\ell} (\eta/\delta)^2 \log^2 m \log(m/\rho))$ iterations of the outer for-loop, 
and $O(1)$ iterations of the inner while-loop by the assumption of $\|\vx^{(k+1)}-\vx^{(k)}\|_{\md^{(k+1)}}\leq\beta$ and 
\cref{lem:change-detection}. 
Each iteration of the while-loop adds a $O(1)$ sized set to the collection $I$ of candidate coordinates.
So overall, \textsc{FindLargeCoordinates}$(\ell)$ returns a set of size
$O(2^{2\ell} (\eta/\delta)^2 \log^2 m \log(m/\rho))$.
Summing up over all calls of $\textsc{FindLargeCoordinates}$, the total size of the set of coordinates to update is
\begin{equation} \label{eq:N_k-defn}
N_k \defeq \sum_{\ell=0}^{\ell_k} O(2^{2\ell}(\beta/\delta)^{2}\log^{2}m\log(m/\rho))= O(2^{2\ell_{k}}(\beta/\delta)^{2}\log^{2}m).
\end{equation}
We define $\ell_0=N_0=0$ for convenience.

\subparagraph{Changes to sketching data structures.}

Let $\mathcal{S}^{(k)}$ denote the set of nodes $H$, where one of (when applicable) $\mm_{(H,P)}$, $\mj_H$, $\zprev|_{F_H}$, $\zsum|_{F_H}$, $\vy_{F_H}$, $\md_{E(H)}$ changes during step $k$. 
(They are entirely induced by changes in $\vv$ and $\vw$ at step $k$.)
We store $\mathcal{S}^{(k)}$ for each step.

For each $\ell$, the diagonal matrix $\overline{\md}$ is the same as $\md$,
except $\overline{\md}_{ii}$ is temporarily zeroed out for $2^\ell$ steps after $\overline{\vx}_i$ changes at a step. 
Thus, the number of coordinate changes to $\overline{\md}$ at step $k$ is the number of changes to $\md$, plus $N_{k-1}+N_{k-2^\ell}$:
$N_{k-1}$ entries are zeroed out because of updates to $\overline{\vx}_i$ in step $k-1$.
The $N_{k-2^\ell}$ entries that were zeroed out in step $k-2^\ell+1$ because of the update to $\overline{\vx}_i$ in step $k-2^\ell$ are back. 

Hence, at step $k$,
the updates to $\texttt{sketch}_{\ell}$ are induced by updates to $\overline{\md}$, and the updates to $\vx$ at step $k$, and at step $k-2^\ell$. 
The updates to the two $\vx$ terms are restricted to the nodes $\mathcal{S}^{(k-2^\ell)} \cup \mathcal{S}^{(k)}$ in $\ct$ for \cref{algo:maintain-sketch}.
Updates to $\texttt{ox\_cur}$ and $\texttt{ox\_prev}$ can be similarly analyzed.

\subparagraph{Runtime of \textsc{Approximate}.}
First, we consider the time to update each $\texttt{sketch}_{\ell}$:
At step $k$, the analysis above combined with \cref{lem:maintain-sketch} show that $\texttt{sketch}_{\ell}.\textsc{Update}$ with new iterations of the appropriate variables run in time
\begin{align*}
	&\phantom{{}={}} O\left(\sketchlen \cdot T \left(\eta \cdot  (|\mathcal{S}^{(k)}|+|\mathcal{S}^{(k-2^\ell)}|+N_{k-1}+N_{k-2^\ell} ) \right) \right) \\
	&\leq \sketchlen \cdot O\left(T (\eta \cdot  (|\mathcal{S}^{(k)}| +N_{k-1}+N_{k-2^\ell} )) \right) + \sketchlen \cdot O \left( T(\eta \cdot |\mathcal{S}^{(k-2^\ell)}| )\right),
\end{align*}
where we use the concavity of $T$. 
The second term can be charged to step $k-2^\ell$. 
Thus, the amortized time cost for $\texttt{sketch}_{\ell}.\textsc{Update}$ at step $k$ is
 \[
 \sketchlen \cdot O(T(\eta \cdot (|\mathcal{S}^{(k)}|+N_{k-1} + N_{k-2^{\ell_k}}))).
\]
Summing over all $0 \leq \ell \leq O(\log m)$ for the different copies of $\texttt{sketch}_{\ell}$, we get an extra $O(\log m)$ factor in the overall update time.

Similarly, we can update $\texttt{ox\_prev}$ and $\texttt{ox\_cur}$ in the same amortized time.

Next, we consider the runtime for \cref{line:D-induced-ox-change} in \cref{algo:maintain-vector}:
The number of coordinate accesses to $\vx^{(k-1)}$ is $|\{i: \md_{ii}^{(k)}-\md_{ii}^{(k-1)} \neq 0\}| = O(\mathcal{S}^{(k)})$. Each coordinate is computed by calling $\texttt{ox\_cur}.\textsc{Query}$, and by \cref{lem:maintain-sketch}, the total time for these updates is $\sketchlen \cdot O(T(\eta \cdot |\mathcal{S}^{(k)}|)$.

Finally, we analyze the remainder of the procedure, which consists of \textsc{FindLargeCoordinates}($\ell$) for each $0 \leq \ell \leq \ell_k$ and the subsequent updates to entries of $\ox$:
For each \textsc{FindLargeCoordinates}$(\ell)$ call, by \cref{lem:change-detection},  $N_{k,\ell} \defeq \Theta(2^{2\ell} (\beta/\delta)^2 \log^2 m \log (m/\rho))$ sampling paths are explored in the $\texttt{sketch}_{\ell}$ data structure, where each sampling path correspond to one iteration of the while-loop.
We calculate $\|\mphi_{E(H)} \vx_{\ell} \|_{2}^{2}$ at a node $H$ in the sampling path
using $\texttt{sketch}_{\ell}.\textsc{Estimate}(H)$, and at a leaf node $H$
using $\texttt{sketch}_{\ell}.\textsc{Query}(H)$. 
The total time is $w\cdot O(T(\eta \cdot N_{k,\ell}))$ by \cref{lem:maintain-sketch}.
To update a coordinate $i \in E(H)$ that was identified to be large, we can refer to the output of $\texttt{sketch}_{\ell}.\textsc{Query}(H)$ from the sampling step.

Summing over each $0 \leq \ell \leq \ell_k$, we see that the total time for the \textsc{FindLargeCoordinates} calls and the subsequent updates fo $\ox$ is
\[
	\sum_{\ell=0}^{\ell_k} \sketchlen \cdot O(T(\eta \cdot N_{k,\ell})) = 
	\sketchlen \cdot O(T(\eta \cdot N_{k})),
\]
where $N_k$ is the number of coordinates that are updated in $\ox$ as shown in \cref{eq:N_k-defn}.

Combined with the update times, we conclude that the total amortized cost of \textsc{Approximate} at step $k$ is
\begin{align*}
&\Theta(\eta^2\log (\frac{m}{\rho})\log m) \cdot T(\eta \cdot (|\mathcal{S}^{(k)}|+N_{k-1} + N_{k - 2^{\ell_k}})).
\end{align*}
Observe that $N_{k-1} = N_{k-2^0}$ and $N_{k-2^\ell}$ are both bounded by $O(N_{k-2^{\ell_k}})$:
When $\ell\neq \ell_k$, the number of trailing zeros in $k-2^\ell$ is no more than $\ell_k$. When $\ell=\ell_k$, the number of trailing zeros of $k-2^{\ell_k}$ is $\ell_{k-2^{\ell_k}}$. In both cases, $\ell_{k-2^{\ell}}\le \ell_{k-2^{\ell_k}}$. 
So we have the desired overall runtime.
\end{proof}

\section{Slack projection} \label{sec:slack_projection}

In this section, we define the slack tree operator as required to use \textsc{MaintainRep}.
We then give the full slack maintenance data structure.

\subsection{Tree operator for slack}

The full slack update at IPM step $k$ with step direction $\vv^{(k)}$ and step size $\bar{t} h$ is
\[
	\vs\leftarrow\vs+\mw^{-1/2}\widetilde{\mproj}_{\vw} (\bar{t} h \vv^{(k)}),
\]
where we require $\widetilde{\mproj}_{\vw} \approx \mproj_{\vw}$ and $\widetilde \mproj_{\vw}\vv^{(k)} \in \range{\mw^{1/2} \mb}$.

Let $\widetilde{\ml}^{-1}$ denote the approximation of $\ml^{-1}$ from \cref{eq:overview_Linv_approx}, maintained and computable with a \textsc{DynamicSC} data structure. If we define
\[
\widetilde{\mproj}_{\vw} = \mw^{1/2} \mb \widetilde{\ml}^{-1} \mb^\top \mw^{1/2} = 
\mw^{1/2} \mb \mpi^{(0)\top} \cdots \mpi^{(\eta-1)\top} \widetilde\mga \mpi^{(\eta-1)} \cdots \mpi^{(0)} \mb^\top \mw^{1/2}.
\]
then $\widetilde{\mproj}_{\vw} \approx_{\eta \epssc} \mproj_{\vw}$, and $\range{\widetilde{\mproj}_{\vw}}=\range{\mproj_{\vw}}$ by definition, where $\eta$ and $\epssc$ are parameters in \textsc{DynamicSC}.
Hence, this suffices as our approximate slack projection matrix.
In order to use \textsc{MaintainRep} to maintain $\vs$ throughout the IPM, it remains to define a slack tree operator $\mm^{\slack}$ so that
\[
	\mw^{-1/2} \widetilde \mproj_{\vw} \vv^{(k)} = \mm^{\slack} \vz^{(k)},
\]
where $\vz^{(k)} \defeq \widetilde{\mga} \mpi^{(\eta-1)} \cdots \mpi^{(0)} \mb^\top \mw^{1/2} \vv^{(k)}$ at IPM step $k$.
We proceed by defining a tree operator $\mm$ satisfying $\widetilde \mproj_{\vw} \vv^{(k)} = \mm \vz^{(k)}$. Namely, we show that $\mm \defeq \mw^{1/2} \mb \mpi^{(0) \top} \cdots \mpi^{(\eta - 1) \top}$ is indeed a tree operator. 
Then we set $\mm^{\slack} \defeq \mw^{-1/2} \mm$.

For the remainder of the section, we abuse notation and use $\vz$ to mean $\vz^{(k)}$ for one IPM step $k$.

\begin{definition}[Slack projection tree operator] \label{defn:slack-forest-operator}
	Let $\ct$ be the separator tree from data structure \textsc{DynamicSC}, with Laplacians $\ml^{(H)}$ and $\tsc(\ml^{(H)}, \partial H)$ at each node $H \in \ct$.
	We use $\mb[H]$ to denote the adjacency matrix of $G$ restricted to the region.
	
	For a node $H \in \ct$, define $V(H)$ and $F_H$ required by the tree operator as $\bdry{\region} \cup F_H$ and $F_H$ from the separator tree construction respectively.
	Note the slightly confusing fact that $V(H)$ \emph{is not} the set of vertices in region $H$ of the input graph $G$, \emph{unless} $H$ is a leaf node.
	Suppose node $H$ has parent $P$, then define the tree edge operator $\mm_{(H, P)} : \R^{V(P)} \mapsto \R^{V(H)}$ as:
	\begin{equation}
		\label{eq:slack-forest-op-defn}
		\mm_{(H,P)} \defeq \mi_{\bdry{H} \cup F_H}-\left(\ml^{(H)}_{\elim{H},\elim{H}}\right)^{-1} \ml^{(H)}_{\elim{H}, \bdry{H}} = \mi_{\bdry{H} \cup F_H} - \mx^{(H) \top},
	\end{equation}
	where $\mx^{(H)}$ is defined in \cref{def:mx^(H)}.

	At each leaf node $H$ of $\ct$, define the leaf operator $\mj_H=\mw^{1/2} \mb[H]$. 
\end{definition}

The remainder of this section proves the correctness of the tree operator.

\begin{lemma}
\label{lem:slack-operator-correctness}
Let $\mm$ be the tree operator as defined in \cref{defn:slack-forest-operator}. We have $$\mm\vz=\mw^{1/2}\mb \mpi^{(0) \top} \cdots \mpi^{(\eta - 1) \top}\vz.$$ 
\end{lemma}

We begin with a few observations about the $\mpi^{(i)}$'s:
\begin{observation} \label{obs:img-mpi}
For any $0 \leq i < \eta$, and for any vector $\vx$, we have 
$\mpi^{(i) \top} \vx = \vx + \vy_i$, where $\vy_i$ is a vector supported on $F_i = \cup_{H \in \ct(i)} F_H$.
Extending this observation, for $0 \leq i < j < \eta$,
\[
	\mpi^{(i)\top} \cdots \mpi^{(j-1)\top}\vx = \vx + \vy,
\]
where $\vy$ is a vector supported on $F_i \cup \cdots \cup F_{j-1} = \cup_{H : i \leq \eta(H) < j} F_H$.
Furthermore, if $\vx$ is supported on $F_A$ for $\eta(A) = j$, then $\vy$ is supported on $\cup_{H \in \ct_A} F_H$.
\end{observation}

The following helper lemma describes a sequence of edge operators from a node to a leaf.

\begin{lemma}
\label{lem:slack-edge-operator-correctness}
For any leaf node $H\in \ct$, and a node $A$ with $H \in \ct_A$ ($A$ is an ancestor of $H$ or $H$ itself), we have 
\begin{equation} \label{eq:slack-edge-path}
\mm_{H\leftarrow A}\vz|_{\elim{A}} = \mi_{\partial H \cup F_H} \mpi^{(0) \top} \cdots \mpi^{(\eta - 1) \top}\vz|_{F_A}.
\end{equation}
\end{lemma}
\begin{proof}
For simplicity of notation, let $V(H) \defeq \partial H \cup F_H$ for a node $H$.

To start, observe that for a node $A$ at level $\eta(A)$, we have $\mpi^{(i)} \vz|_{F_A} = \vz|_{F_A}$ for all $i \geq \eta(A)$. So it suffices to prove 
\[
\mm_{H\leftarrow A}\vz|_{\elim{A}} = \mi_{V(H)} \mpi^{(0) \top} \cdots \mpi^{(\eta(A)-1) \top} \vz|_{F_A}.
\]

Let the path from leaf $H$ up to node $A$ in $\ct$ be denoted $(H_0 \defeq H, H_1, \ldots, H_t \defeq A)$, for some $t\le \eta(A)$. We will prove by induction for $k$ decreasing from $t$ to $0$:
\begin{equation} \label{eq:slack-op-induction}
\mm_{H_k\leftarrow A}\vz|_{\elim{A}} = \mi_{V(H_k)} \mpi^{(\eta(H_k)) \top} \mpi^{(\eta(H_k) + 1)\top} \cdots \mpi^{(\eta(A)-1) \top}\vz|_{\elim{A}}.
\end{equation}

For the base case of $H_t = A$, we have $\mm_{H_t \leftarrow A} \vz|_{\elim{A}} = \vz|_{\elim{A}} = \mi_{V(H_t)} \vz|_{\elim{A}}$.

For the inductive step at $H_k$, we first apply induction hypothesis for $H_{k+1}$ to get
\begin{align} 
\mm_{H_{k+1}\leftarrow A}\vz|_{\elim{A}} &= \mi_{V(H_{k+1})} \mpi^{(\eta(H_{k+1})) \top} \cdots \mpi^{(\eta(A)-1) \top}\vz|_{\elim{A}}.
\intertext{
	Multiplying by the edge operator $\mm_{(H_{k}, H_{k+1})}$ on both sides gives}
\mm_{H_{k}\leftarrow A}\vz|_{\elim{A}} &= \mm_{(H_k, H_{k+1})} \mi_{V(H_{k+1})} \mpi^{(\eta(H_{k+1})) \top} \cdots \mpi^{(\eta(A)-1) \top}\vz|_{\elim{A}}.
\end{align}
Recall the edge operator $\mm_{(H_k, H_{k+1})}$ maps vectors supported on $V(H_{k+1})$ to vectors supported on $V(H_k)$ and zeros otherwise. So we can drop the $\mi_{V(H_{k+1})}$ term in the right hand side.
Let $\vx \defeq \mpi^{(\eta(H_{k+1})) \top} \cdots \mpi^{(\eta(A)-1) \top}\vz|_{\elim{A}}$.
Now, by the definition of the edge operator, the above equation becomes
\begin{equation} \label{eq:slack-op}
\mm_{H_{k}\leftarrow A}\vz|_{\elim{A}} = (\mi_{V(H_k)} - \mx^{(H_k) \top})\vx.
\end{equation}
On the other hand, we have
 \begin{align*}
\mi_{V(H_k)} \mpi^{(\eta(H_k))\top} \cdots \mpi^{(\eta(H_{k+1})-1)\top} \vx &= 
\mi_{V(H_k)}  \mpi^{(\eta(H_k))\top} \left(\mpi^{(\eta(H_k)+1)\top} \cdots \mpi^{(\eta(H_{k+1})-1)\top} \vx \right)\\
&= \mi_{V(H_k)}  \mpi^{(\eta(H_k))\top} (\vx + \vy),
\intertext{
	where $\vy$ is a vector supported on $\cup F_R$ for nodes $R$ at levels $\eta(H_k)+1, \cdots, \eta(H_{k+1})-1$ by \cref{obs:img-mpi}.
	In particular, $\vy$ is zero on $F_{H_k}$. Also, $\vy$ is zero on $\partial H_k$, since by \cref{lem:bdry-containment}, $\partial H_k \subseteq \cup_{\text{ancestor $A'$ of $H_k$}} F_{A'}$, and ancestors of $H_k$ are at level $\eta(H_{k+1})$ or higher.
	Then $\vy$ is zero on $V(H_k) = \partial H_k \cup F_{H_k}$, and the right hand side is
}
&= (\mi_{V(H_k)} - \mx^{(H_k) \top}) \vx,
 \end{align*}
where we apply the definition of $\mpi^{(\eta(H_k))\top}$ and expand the left-multiplication by $\mi_{V(H_k)}$.
	
Combining with \cref{eq:slack-op} and substituting back the definition of $\vx$, we get
\[
\mm_{H_{k}\leftarrow A}\vz|_{\elim{A}} = \mi_{V(H_{k})} \mpi^{(\eta(H_{k})) \top} \cdots \mpi^{(\eta(A)-1) \top}\vz|_{\elim{A}}.
\]
which completes the induction.

\end{proof}

To prove \cref{lem:slack-operator-correctness}, we apply the leaf operators to the result of the previous lemma and sum over all nodes and leaf nodes.

\begin{proof}[Proof of \cref{lem:slack-operator-correctness}]
Let $H$ be a leaf node. We sum \cref{eq:slack-edge-path} over all $A$ with $H \in \ct_A$ to get
\begin{align*}
	\sum_{A : H \in \ct_A} \mm_{H\leftarrow A}\vz|_{\elim{A}} &= \mi_{\partial H \cup F_H} \sum_{A : H \in \ct_A} \mpi^{(0) \top} \cdots \mpi^{(\eta - 1) \top}\vz|_{F_A}\\
	&= \mi_{\partial H \cup F_H} \mpi^{(0) \top} \cdots \mpi^{(\eta - 1) \top}\vz,
\end{align*}
where we relax the sum in the right hand side to be over all nodes in $\ct$, since by \cref{obs:img-mpi}, for any $A$ with $H \notin \ct_A$, we simply have $\mi_{\partial H \cup F_H} \mpi^{(0)\top} \cdots \mpi^{(\eta-1)\top} \vz|_{F_A} = \vzero$. Next, we apply the leaf operator $\mj_H = \mw^{1/2} \mb[H]$ to both sides to get
\begin{align*}
	\sum_{A : H \in \ct_A} \mj_H \mm_{H\leftarrow A}\vz|_{\elim{A}}  &= \mw^{1/2}\mb[H]  \mi_{\partial H \cup F_H} \mpi^{(0) \top} \cdots \mpi^{(\eta - 1) \top}\vz.
\end{align*}
Since $\mb[H]$ is zero on columns supported on $V(G) \setminus (\partial H \cup F_H)$, we can simply drop the $\mi_{\partial H \cup F_H}$ in the right hand side.

Finally, we sum up the equation above over all leaf nodes. The left hand side is precisely the definition of $\mm \vz$. Recall the regions of the leaf nodes partition the original graph $G$, so we have
\begin{align*}
\sum_{H \in \ct(0)} \sum_{A : H \in \ct_A} \mj_H \mm_{H\leftarrow A}\vz|_{\elim{A}} &= \mw^{1/2} \left(\sum_{H\in \ct(0)} \mb[H] \right) \mpi^{(0) \top} \cdots \mpi^{(\eta - 1) \top}\vz \\
\mm \vz &= \mw^{1/2} \mb \mpi^{(0) \top} \cdots \mpi^{(\eta - 1) \top}\vz.
\end{align*}
\end{proof}

We now examine the slack tree operator complexity.
\begin{lemma} \label{lem:slack_operator_complexity}
	The complexity of the slack tree operator as defined in \cref{defn:flow-forest-operator} is $T(k) = \otilde(\sqrt{mk}\cdot \epssc^{-2})$, where $\epssc$ is the Schur complement approximation factor from data structure \textsc{DynamicSC}.
\end{lemma}

\begin{proof}
	Let $\mm_{(D,P)}$ be a tree edge operator. 
	Applying $\mm_{(D,P)}=\mi_{\bdry{D}}-\left(\ml^{(D)}_{\elim{D},\elim{D}}\right)^{-1} \ml^{(D)}_{\elim{D}, \bdry{D}}$ 
	to the left or right consists of three steps which are applying $\mi_{\bdry{D}}$, 
	applying $\ml^{(D)}_{\elim{D}, \bdry{D}}$ and solving for $\ml^{(D)}_{\elim{D},\elim{D}}\vv=\vb$ for some vectors $\vv$ and $\vb$. 
	Each of the three steps costs time $O(\epssc^{-2}|\bdry{D}\cup \elim{D}|)$ by \cref{lem:fastApproxSchur} and \cref{thm:laplacianSolver}.
	 
	For any leaf node $H$, $H$ has a constant number of edges, and it takes constant time to compute $\mj_H \vu$ for any vector $\vu$. The number of vertices may be larger but the nonzeros of $\mj_H=\mw^{1/2} \mb[H]$ only depends on the number of edges.
	To bound the total cost over $k$ distinct edges, we apply \cref{lem:planarBoundChangeCost}, which then gives the claimed complexity.
\end{proof}

\subsection{Proof of \crtcref{thm:SlackMaintain}}

Finally, we give the full data structure for maintaining the slack solution.

The tree operator $\mm$ defined in \cref{defn:slack-forest-operator} satisfies $\mm\vz^{(k)} =\widetilde{\mproj}_{\vw}\vv^{(k)}$ at step $k$, by the definition of $\vz^{(k)}$. 
To support the proper update $\vs\leftarrow\vs+\ot h\mw^{-1/2}\widetilde{\mproj}_{\vw} \vv^{(k)}$, we define $\mm^{\slack} \defeq \mw^{-1/2}\mm$ and note it is also a tree operator:
\begin{lemma}
\label{lem:wm}
Suppose $\mm$ is a tree operator supported on $\ct$ with complexity $T(K)$. Let $\md$ be a diagonal matrix in $\R^{E\times E}$ where $E=\bigcup_{\text{leaf }H\in \ct}E(H)$. Then $\md\mm$ can be represented by  a tree operator with complexity $T(K)$.
\end{lemma}
\begin{proof}
Suppose $\mm \in \R^{E\times V}$. For any vector $\vz\in \R^{V}$, $\md\mm\vz=\md(\mm\vz)$. Thus, to compute $\md\mm\vz$, we may first compute $\mm\vz$ and then multiply the $i$-th entry of $\mm\vz$ with $\md_{i, i}$. This can be achieved by defining a new tree operator $\mm'$ with leaf operators $\mj'$ such that $\mj'_{H}=\md_{E(H), E(H)}\mj_{H}$ and $\mm'_{(H, P)}=\mm_{(H, P)}$. The size of each leaf operator remains constant. All edge operators do not change from $\mm$. Thus, the new operator $\mm'$ has the same complexity as $\mm$. 
\end{proof}

With the lemma above, we can use \textsc{MaintainRep} (\cref{alg:maintain_representation}) to maintain the implicit representation of $\vs$ and \cref{thm:VectorTreeMaintain} to maintain an approximate vector $\os$ as required in \cref{algo:IPM_impl}.
A single IPM step calls the procedures \textsc{Reweight, Move, Approximate} in this order once. 
Note that we reinitialize the data structure when $\ot$ changes, so within each instantiation, may assume $\bar{t} = 1$ by scaling. $\ot$ changes only $\O(1)$ times in the IPM.
\begin{algorithm}
	\caption{Slack Maintenance, Main Algorithm}\label{alg:slack-maintain-main}
	\begin{algorithmic}[1]
		\State \textbf{data structure} \textsc{MaintainSlack} \textbf{extends} \textsc{MaintainRep}
		\State \textbf{private: member}
		\State \hspace{4mm} \textsc{MaintainRep} $\maintainRep$: data structure to implicitly maintain
		\[
			\vs = \vy + \mw^{-1/2} \mm (c \zprev + \zsum).
		\]
		\Comment $\mm$ is defined by \cref{defn:slack-forest-operator}
		\State \hspace{4mm} \textsc{MaintainApprox} \texttt{bar\_s}: data structure to maintain approximation $\os$ to $\vs$ (\cref{thm:VectorTreeMaintain})
		\State
		\Procedure{Initialize}{$G,\vs^\init \in\R^{m},\vv\in \R^{m}, \vw\in\R_{>0}^{m},\epssc>0,
		\overline{\epsilon}>0$}
			\State Build the separator tree $\ct$ by \cref{thm:separator_tree_construction} 
			\State $\maintainRep.\textsc{Initialize}(G, \ct, \mw^{-1/2}\mm, \vv, \vw, \vs^\init, \epssc)$
			\Comment{initialize $\vs \leftarrow \vs^{\init}$}
			\State $\texttt{bar\_s}.\textsc{Initialize}(\mw^{-1/2}\mm,c, \zprev,\zsum, \vy, \mw, n^{-5}, \overline{\epsilon})$
			\Comment{initialize $\os$ approximating $\vs$}
		\EndProcedure
		
		\State
		\Procedure{Reweight}{$\vw^\new \in \R^m_{>0}$}
			\State $\maintainRep.\textsc{Reweight}(\vw^\new)$
		\EndProcedure

		\State
		\Procedure{Move}{$\alpha,\vv^\new \in \R^m$}
			\State $\maintainRep.\textsc{Move}(\alpha, \vv^\new)$
		\EndProcedure

		\State
		\Procedure{Approximate}{ }
			\State \Comment the variables in the argument are accessed from \maintainRep
			\State \Return $\os=\texttt{bar\_s}.\textsc{Approximate}(\mw^{-1/2}\mm,c, \zprev,\zsum,\vy, \mw)$
		\EndProcedure

		\State
		\Procedure{Exact}{ }
		\State \Return $\maintainRep.\textsc{Exact}()$
		\EndProcedure
	\end{algorithmic}
\end{algorithm}

\SlackMaintain*

\begin{proof}[Proof of \cref{thm:SlackMaintain}]
We prove the runtime and correctness of each procedure separately. 

Recall by \cref{lem:slack-edge-operator-correctness}, the tree operator $\mm$ has complexity $T(K) = O(\epssc^{-2} \sqrt{mK})$.

\paragraph{\textsc{Initialize}:}
	By the initialization of \maintainRep~(\cref{thm:maintain_representation}), the implicit representation of $\vs$ in \maintainRep~is correct and $\vs = \vs^{\init}$. By the initialization of \flowSketch, $\os$ is set to $\vs$ to start.

	Initialization of $\maintainRep$ takes $\otilde(m\epssc^{-2})$ time by \cref{thm:maintain_representation},
	and the initialization of $\slackSketch$ takes $\otilde(m)$ time by \cref{thm:VectorTreeMaintain}. 

\paragraph{\textsc{Reweight}:}
In \textsc{Reweight}, the value of $\vs$ does not change, but all the variables in \texttt{MaintainRep} are updated to depend on the new weights.
The correctness and runtime follow from \cref{thm:maintain_representation}.

\paragraph{\textsc{Move}:}
	$\maintainRep.\textsc{Move}(\alpha, \vv^{(k)})$ updates the implicit representation of $\vs$ by 
	\[
		\vs \leftarrow \vs + \mw^{-1/2} \mm \alpha \vz^{(k)}.
	\]
	By the definition of the slack projection tree operator $\mm$ and \cref{lem:slack-operator-correctness}, this is equivalent to the update
	\[
		\vs \leftarrow \vs + \alpha\mw^{-1/2}\widetilde{\mproj}_{\vw}\vv^{(k)},
	\]
	where $\widetilde{\mproj}_{\vw} = \mw^{1/2} \mb \mpi^{(0)}\cdots \mpi^{(\eta-1)} \widetilde\mga \mpi^{(\eta-1)}\cdots \mpi^{(0)} \mb^{\top} \mw^{1/2}$. 
	By \cref{thm:L-inv-approx}, $\|\widetilde{\mproj}_{\vw}-\mproj_{\vw}\|_{\mathrm{op}}\leq \eta \epssc$. 
	From the definition, $\range{\mw^{1/2} \widetilde \mproj_{\vw}}\subseteq \range \mb$.

	By the guarantees of \maintainRep, if $\vv^{(k)}$ differs from $\vv^{(k-1)}$ on $K$ coordinates, then the runtime is $\O(\epssc^{-2}\sqrt{mK})$. Furthermore, $\zprev$ and $\zsum$ change on $\elim{H}$ for at most $\O(K)$ nodes in $\ct$.
	
\paragraph{\textsc{Approximate}:}
	The returned vector $\os$ satisfies $\|\mw^{1/2}(\os-\vs)\|_{\infty}\leq\overline{\epsilon}$ by the guarantee of \\ \texttt{bar\_s}.\textsc{Approximate} from \cref{thm:VectorTreeMaintain}.
	
\paragraph{\textsc{Exact}:}
	The runtime and correctness directly follow from the guarantee of $\maintainRep.\textsc{Exact}$ given in \cref{thm:maintain_representation}.

Finally, we have the following lemma about the runtime for \textsc{Approximate}. Let $\os^{(k)}$ denote the returned approximate vector at step $k$.

\begin{lemma}\label{lem:slack-approx}
	Suppose $\alpha\|\vv\|_{2}\leq\beta$ for some $\beta$ for all calls to \textsc{Move}.
	Let $K$ denote the total number of coordinates changed in $\vv$ and $\vw$ between the $k-1$-th and $k$-th \textsc{Reweight} and \textsc{Move} calls. Then at the $k$-th \textsc{Approximate} call,
	\begin{itemize}
		\item  The data structure first sets $\os_e\leftarrow \vs^{(k-1)}_e$ for all coordinates $e$ where $\vw_e$ changed in the last \textsc{Reweight}, then sets $\os_e\leftarrow \vs^{(k)}_e$ for $O(N_k\defeq 2^{2\ell_{k}}(\frac{\beta}{\overline{\epsilon}})^{2}\log^{2}m)$ coordinates $e$, where $\ell_{k}$ is the largest integer
		$\ell$ with $k=0\mod2^{\ell}$ when $k\neq 0$ and $\ell_0=0$. 
		\item The amortized time for the $k$-th \textsc{Approximate} call
		is $\widetilde{O}(\epsilon_{\mproj}^{-2}\sqrt{m(K+N_{k-2^{\ell_k}})})$.
	\end{itemize}
\end{lemma}
\begin{proof}
	Since $\os$ is maintained by \texttt{bar\_s}, we apply \cref{thm:VectorTreeMaintain} with $\vx = \vs$ and diagonal matrix $\md = \mw$. 
	We need to prove $\|\vx^{(k)}-\vx^{(k-1)}\|_{\md^{(k)}}\le O(\beta)$ for all $k$ first.  The constant factor in $O(\beta)$ does not affect the guarantees in \cref{thm:VectorTreeMaintain}. The left-hand side is
	\begin{align*}
		\norm{\vs^{(k)}-\vs^{(k-1)}}_{\mw^{(k)}} &= \norm{\alpha^{(k)} {\mw^{(k)}}^{-1/2} \widetilde{\mproj}_{\vw}\vv^{(k)}}_{\mw^{(k)}} \tag{by \textsc{Move}}\\
		&= \norm{\alpha^{(k)} \widetilde{\mproj}_{\vw}\vv^{(k)}}_2 \\
		&\le (1+\eta\epssc) \alpha^{(k)} \|\vv^{(k)}\|_2 \tag{by the assumption that $\alpha\|\vv\|_{2}\leq\beta$}\\
		&\le 2 \beta.
	\end{align*}
	Where the second last step follows from $\|\widetilde{\mproj}_{\vw}-\mproj_{\vw}\|_{\mathrm{op}}\leq \eta\epssc$ and the fact that $\mproj_{\vw}$ is an orthogonal projection. 
	Now, we can apply \cref{thm:VectorTreeMaintain} to conclude that at each step $k$, 
	\texttt{bar\_s}.\textsc{Approximate} first sets $\os_e\leftarrow \vs^{(k-1)}_e$ for all coordinates $e$ where $\vw_e$ changed in the last \textsc{Reweight}, then set $\os_e\leftarrow \vs^{(k)}_e$ for $O(N_k\defeq 2^{2\ell_{k}}(\frac{\beta}{\overline{\epsilon}})^{2}\log^{2}m)$ coordinates $e$, where $\ell_{k}$ is the largest integer
		$\ell$ with $k=0\mod2^{\ell}$ when $k\neq 0$ and $\ell_0=0$. 
	
	For the second point, \textsc{Move} updates $\zprev$ and $\zsum$ on $\elim{H}$ for $\O(K)$ different nodes $H \in \ct$ by \cref{thm:maintain_representation}.
	\textsc{Reweight} then updates $\zprev$ and $\zsum$ on $F_H$ for $\O(K)$ different nodes, and updates the tree operator $\mw^{-1/2}\mm$ on $\O(K)$ different edge and leaf operators. In turn, it updates $\vy$ on $E(H)$ for $\O(K)$ leaf nodes $H$. Now, we apply \cref{thm:VectorTreeMaintain} and the complexity of the tree operator to conclude the desired amortized runtime. 
\end{proof}
\end{proof}

\section{Flow projection} \label{sec:flow_projection}

In this section, we define the flow tree operator as required to use \textsc{MaintainRep}.
We then give the full flow maintenance data structure.

During the IPM, we maintain $\vf \defeq \hat{\vf} - \pf$ by maintaining the two terms separately. For IPM step $k$ with direction $\vv^{(k)}$ and step size $h$, we update them as follows:
\begin{align*}
	\uf &\leftarrow \uf + h \mw^{1/2} \vv^{(k)}, \\ 
	\label{eq:tf-update}
	\pf &\leftarrow \pf + h \mw^{1/2} \widetilde{\mproj}'_{\vw} \vv^{(k)},
\end{align*}
where $\widetilde{\mproj}'_{\vw} \vv^{(k)}$ satisfies $\norm{\widetilde{\mproj}'_{\vw} \vv^{(k)} - \mproj_{\vw} \vv^{(k)}}_2 \leq \eps \norm{\vv^{(k)}}_2$ for some factor $\eps$, and $\mb^\top \mw^{1/2} \widetilde{\mproj}'_{\vw} \vv^{(k)} = \mb^\top \mw^{1/2} \vv^{(k)}$. We will include the initial value of $\vf$ in $\uf$. 

Maintaining $\hat\vf$ is straightforward; in the following section, we focus on $\pf$. 

\subsection{Tree operator for flow}

We hope to use \textsc{MaintainRep} to maintain $\pf$ throughout the IPM. In order to do so, it remains to define a flow tree operator $\mm^{\flow}$ so that
\[
\mw^{1/2} \widetilde \mproj'_{\vw} \vv^{(k)} = \mm^{\flow} \vz^{(k)},
\]
where $\widetilde{\mproj}'_{\vw} \vv$ satisfies the constraints mentioned above, 
and $\vz^{(k)} \defeq \widetilde{\mga} \mpi^{(\eta-1)} \cdots \mpi^{(0)} \mb^\top \mw^{1/2} \vv^{(k)}$.
We will define a flow projection tree operator $\mm$ so that $\tf \defeq \mm \vv^{(k)}$ satisfies 
$\norm{\tf - \mproj_{\vw} \vv}_2 \leq O(\eta \epssc) \norm{\vv}_2$
and $\mb^{\top} \mw^{1/2} \tf = \mb^{\top} \mw^{1/2} \vv^{(k)}$.
This means it is feasible to set $\widetilde{\mproj}'_{\vw} \vv^{(k)} = \tf$.
Then, we define $\mm^{\flow} \defeq \mw^{-1/2} \mm$.

For the remainder of the section, we abuse notation and use $\vz$ to mean $\vz^{(k)}$ for one IPM step $k$.

\begin{definition}[Flow projection tree operator] \label{defn:flow-forest-operator}
	Let $\ct$ be the separator tree from data structure \textsc{DynamicSC}, with Laplacians $\ml^{(H)}$ and $\tsc(\ml^{(H)}, \partial H)$ at each node $H \in \ct$.
	We use $\mb[H]$ to denote the adjacency matrix of $G$ restricted to the region.
		
	To define the flow projection tree operator $\mm$, we proceed as follows: 
	The tree operator is supported on the tree $\ct$. 
	For a node $H \in \ct$ with parent $P$, define the tree edge operator $\mm_{(H, P)}$ as:
	\begin{equation}
	\label{eq:flow-forest-op-defn}
	\mm_{(H, P)} \defeq (\ml^{(H)})^{-1} \tsc(\ml^{(H)}, \partial H).
	\end{equation}

	At each node $H$, we let $F_H$ in the tree operator be the set $F_H$ of eliminated vertices defined in the separator tree. 
	At each leaf node $H$ of $\ct$, we have the leaf operator $\mj_H=\mw^{1/2} \mb[H]$.
\end{definition}


Before we give intuition and formally prove the correctness of the flow tree operator, we examine its complexity.
\begin{lemma}
\label{lem:flow _operator_complexity}
	The complexity of the flow tree operator as defined in \cref{defn:flow-forest-operator} is $T(k) = \otilde(\sqrt{mk}\cdot \epssc^{-2})$, where $\epssc$ is the overall approximation factor from data structure \textsc{DynamicSC}.
\end{lemma}

\begin{proof}
	Let $\mm_{(H, P)}$ be a tree edge operator.  Note that it is a symmetric matrix.
	For any leaf node $H$, $H$ has a constant number of edges, and it takes constant time to compute $\mj_H \vu$ for any vector $\vu$. The number of vertices may be larger but the nonzeros of $\mj_H=\mw^{1/2} \mb[H]$ only depends on the number of edges.
	
	If $H$ is not a leaf node, then $\mm_{(H, P)} \vu$ consists of multiplying with $\tsc(\ml^{(H)}, \partial H)$ and solving the Laplacian system $\ml^{(H)}$.
	By \cref{lem:fastApproxSchur} and \cref{thm:laplacianSolver}, this can be done in $\otilde(\epssc^{-2} \cdot |\partial H|)$ time. 
	To bound the total cost over $k$ distinct edges, we apply \cref{lem:planarBoundChangeCost}, which gives the claimed complexity.
\end{proof}

\begin{theorem}\label{thm:flow-forest-operator-correctness}
	Let $\vv \in \R^m$, and let $\vz = \widetilde{\mga} \mpi^{(\eta-1)} \cdots \mpi^{(0)} \mb^\top \mw^{1/2} \vv$.
	Let $\mm$ be the flow projection tree operator from \cref{defn:flow-forest-operator}. 
	Suppose $\epssc=O(1/ \log m)$ is the overall approximation factor from \textsc{DynamicSC}.
	Then $\tf \defeq \mm \vz$ satisfies $\mb^\top \mw^{1/2} \tf = \mb^\top \mw^{1/2} \vv$ and $\norm{\tf - \mproj_{\vw} \vv}_2 \leq O(\eta \epssc) \norm{\vv}_2$.
\end{theorem}

The remainder of the section is dedicated to proving this theorem. 

Fix $\vv$ for the remainder of this section. Let $\vd \defeq \mb^\top \mw^{1/2} \vv \in \R^n$; 
since it is supported on the vertices of $G$ and its entries sum to 0, it is a \emph{demand} vector. 
In the first part of the proof, we show that $\tf$ routes the demand $\vd$.
Let $\vf^\star \defeq \mproj_{\vw} \vv = \mw^{1/2}\mb\ml^{-1} \vd$. In the second part of the proof, we show that $\tf$ is close to $\vf^\star$. 
Finally, a remark about terminology:

\begin{remark}
	If $\mb$ is the incidence matrix of a graph, then any vector of the form $\mb \vx$ is a flow by definition. Often in this section, we have vectors of the form $\mw^{1/2} \mb \vx$. In this case, we refer to it as a \emph{weighted flow}. We say a weighted flow $\vf$ routes a demand $\vd$ if $(\mw^{1/2} \mb)^\top \vf = \vd$.
\end{remark}

We proceed with a series of lemmas and their intuition, before tying them together in the overall proof at the end of the section.

\begin{restatable}{lemma}{demandDecomposition}\label{thm:demand-decomposition}
	Let $\vz = \widetilde{\mga} \mpi^{(\eta-1)} \cdots \mpi^{(0)} \mb^\top \mw^{1/2} \vv$ be as given in \cref{thm:flow-forest-operator-correctness}.
	For each node $H \in \ct$, let $\vz|_{\elim{H}}$ be the sub-vector of $\vz$ supported on the vertices $\elim{H}$, and define the demand
	\[
	\vd^{(H)} \defeq \ml^{(H)} \vz|_{\elim{H}}.
	\]
	Then $\vd = \sum_{H \in \ct} \vd^{(H)}$.
\end{restatable}
\begin{proof}		
In the proof, note that all $\mi$ are $n \times n$ matrices,
and we implicitly pad all vectors with the necessary zeros to match the dimensions.
For example, $\vz|_{F_H}$ below should be viewed as an $n$-dimensional vector supported on $F_H$.
Define 
\[\mx^{(i)} = \sum_{H \in \ct(i)} \mx^{(H)}.
\]
We have
\[
\mpi^{(i)} = \mi - \mx^{(i)} 
= \mi - \sum_{H \in \ct(i)} \ml^{(H)}_{\bdry{\region}, F_H} \left( \ml^{(H)}_{F_H, F_H}\right)^{-1}.
\]

Suppose $H$ is at level $i$ of $\ct$. We have
\begin{align}
	\vz|_{F_H}
	& =  (\ml_{F_H,F_H}^{(H)})^{-1} \mpi^{(\eta-1)}\cdots\mpi^{(1)}\mpi^{(0)} \vd \nonumber \\
	& =  (\ml_{F_H,F_H}^{(H)})^{-1} \mpi^{(i-1)} \cdots \mpi^{(1)} \mpi^{(0)} \vd,   \label{eq:z_FH-alt}
\end{align}
where we use the fact $\text{Im}(\mx^{(H')}) \cap F_H = \emptyset$ if $\eta(H') \geq i$. 
From this expression for $\vz|_{F_H}$, we have
\begin{align*}
	\vd^{(H)} &\defeq \ml^{(H)} \vz|_{F_H} \\
	&= \ml^{(H)}_{\partial H, F_H} \vz|_{F_H} + \ml^{(H)}_{F_H, F_H} \vz|_{F_H} \\
	&= \mx^{(H)} (\mpi^{(i-1)} \cdots \mpi^{(1)} \mpi^{(0)}\vd)_{F_H} + (\mpi^{(\eta-1)} \cdots \mpi^{(1)} \mpi^{(0)}\vd)|_{F_H},
\end{align*}
where the last line follows from \cref{eq:z_FH-alt}. By padding zeros to $\mx^{(H)}$, we can write the equation above as

$$\vd^{(H)} = \mx^{(H)} \mpi^{(i-1)} \cdots \mpi^{(1)} \mpi^{(0)}\vd + (\mpi^{(\eta-1)} \cdots \mpi^{(1)} \mpi^{(0)}\vd)|_{F_H}.$$

Now, computing the sum, we have
\begin{align*}
	\sum_{H \in \ct} \vd^{(H)}
	&=\sum_{i=0}^{\eta}\sum_{H\in \ct(i)}\mx^{(H)}\mpi^{(i-1)} \cdots \mpi^{(1)} \mpi^{(0)} \vd +\sum_{i=0}^{\eta}\sum_{H\in \ct(i)} (\mpi^{(\eta-1)} \cdots \mpi^{(1)} \mpi^{(0)} \vd)|_{F_H}\\
	&=\left(\sum_{i=0}^{\eta}\mx^{(i)} \mpi^{(i-1)} \cdots \mpi^{(1)} \mpi^{(0)} \vd\right) + \mpi^{(\eta-1)} \cdots \mpi^{(1)} \mpi^{(0)} \vd \tag{$F_H$ partition $V(G)$}\\
	&=\left(\sum_{i=0}^{\eta-1}(\mi-\mpi^{(i)})\mpi^{(i-1)} \cdots \mpi^{(1)} \mpi^{(0)} \vd\right) + \mpi^{(\eta-1)} \cdots \mpi^{(1)} \mpi^{(0)} \vd \\
	&= \vd, \tag{telescoping sum}
\end{align*}
completing our proof.
\end{proof}

Next, we examine the feasibility of $\tf$. To begin, we introduce a decomposition of $\tf$ based on the decomposition of $\vd$, and prove its feasibility.

\begin{definition}
	Let $\mm^{(H)}$ be the flow tree operator supported on the tree $\ct_H \in \cf$ (\cref{defn:subtree-operator}). We define the flow $\tf^{(H)} \defeq \mm^{(H)} \vz = \mm^{(H)} \vz|_{F_H}$.
\end{definition}

\begin{lemma}\label{lem:flow-forest-operator-feasibility}
	We have that $(\mw^{1/2} \mb)^\top \tf^{(H)} = \vd^{(H)}$. In other words, the weighted flow $\tf^{(H)}$ routes the demand $\vd^{(H)}$ \emph{using the edges of the original graph $G$}.
\end{lemma}

\begin{proof}
	We will first show inductively that for each $H \in \ct$, we have $\mb^\top \mw^{1/2} \mm^{(H)} = \ml^{(H)}$.
	
	In the base case, if $H$ is a leaf node of $\ct$, then $\cf_H$ is a tree with root $H$ and a single leaf node under it.
	Then $\mm^{(H)} = \mw^{1/2} \mb[H]$.
	It follows that
	\[
	\mb^\top \mw^{1/2} \mm^{(H)} = \mb^\top \mw^{1/2} \mw^{1/2} \mb[H] = \ml^{(H)},
	\]
	by definition of $\ml^{(H)}$ for a leaf $H$ of $\ct$.
	
	In the other case, $H$ is not a leaf node of $\ct$. Let $D_1, D_2$ be the two children of $H$. Then
	\begin{align*}
		\mb^\top \mw^{1/2} \mm^{(H)} &= \mb^\top \mw^{1/2}  
		\left( \mm^{(D_1)} \mm_{(D_1, H)} +   
		\mm^{(D_2)} \mm_{(D_2, H)} \right) \\
		&= \ml^{(D_1)} \mm_{(D_1, H)} +   
		\ml^{(D_2)} \mm_{(D_2, H)}  \tag{by induction} \\
		&=  \ml^{(D_1)} (\ml^{(D_1)})^{-1} \tsc(\ml^{(D_1)}, \partial D_1) +   
		\ml^{(D_2)} (\ml^{(D_2)})^{-1} \tsc(\ml^{(D_2)}, \partial D_2) \\
		&= \tsc(\ml^{(D_1)}, \partial D_1) + \tsc(\ml^{(D_2)}, \partial D_2) \\
		&= \ml^{(H)}.
	\end{align*}
	
	Finally, we conclude that $\mb^\top \mw^{1/2} \tf^{(H)} = \mb^\top \mw^{1/2} \mm^{(H)} \vz|_{F_H} = \ml^{(H)} \vz_{F_H} = \vd^{(H)}$, where the last inequality follows by definition of $\vd^{(H)}$.
\end{proof}

We observe an orthogonality property of the flows, which will become useful later:
\begin{lemma} \label{lem:orthogonal-flows}
	For any nodes $H, H'$ at the same level in $\ct$, $\range{\mm^{(H)}}$ and $\range{\mm^{(H')}}$ are disjoint. 
	Consequently, the flows $\tf^{(H)}$ and $\tf^{(H')}$ are orthogonal.
\end{lemma}
\begin{proof}
	Recall leaves of $\ct$ correspond to pairwise edge-disjoint, constant-sized regions of the original graph $G$.
	Since $H$ and $H'$ are at the same level in $\ct$, we know $\ct_H$ and $\ct_{H'}$ have disjoint sets of leaves.
	The range of $\mm^{(H)}$ is supported on edges in the regions given by leaves of $\ct_H$, and analogously for the range of $\mm^{(H')}$. 
\end{proof}

Next, we set up the tools for bounding $\norm{\tf - \mproj_{\vw} \vv}_2$, involving an energy analysis drawing inspiration from electric flow routing.
We begin with the canonical definitions and properties of electric-flow energy.

\begin{definition} \label{thm:energy-minimizer} 
	Let $\mw^{1/2} \mb$ be the edge-weighted incidence matrix of some graph $G$, and let $\ml \defeq \mb^\top \mw \mb$ be the Laplacian.
	Let $\vd \defeq \ml \vz$ be a demand and $\vf$ be any weighted flow that routes $\vd$; 
	that is, $(\mw^{1/2} \mb)^\top \vf = \vd$. Then we say $\norm{\vf}_2^2$ is the \emph{energy} of the flow $\vf$.

	There is a unique energy-minimizing flow $\vf^\star$ routing the demand $\vd$ on $G$. 
	From the study of electric flows, we know $\vf^\star = \mw^{1/2} \mb \ml^{-1} \vd$.
	Hence, we can refer to its energy as the energy of the demand $\vd$ on the graph of $\ml$, given by
	\begin{equation} \label{defn:energy-minimizer}
		\mathcal{E}_{\ml}(\vd) \defeq \min_{(\mw^{1/2}\mb)^\top \vf = \vd} \norm{\vf}_2^2  = \vd^\top (\mb^\top \mw \mb)^{-1} \vd = \vd^\top \ml^{-1} \vd = \vz^\top \ml \vz.
	\end{equation}
\end{definition}

We want to understanding how the energy changes when, instead of routing $\vd$ using the edges of $G$, we use edges of some other graphs related to $G$.
In particular, we are interested in the operations of graph decompositions and taking Schur complements. 
It turns out the energy behaves nicely:

\begin{lemma} \label{lem:energy-graph-decomposition}
	Suppose $G$ is a weighted graph that can be decomposed into weighted subgraphs $G_1, G_2$.
	That is, if $\ml$ is the Laplacian of $G$, and $\ml^{(i)}$ is the Laplacian of $G_i$, then $\ml = \ml^{(1)} + \ml^{(2)}$.
	Suppose $\vd \defeq \ml \vz$ is a demand on the vertices of $G$. Then if we decompose $\vd = \vd^{(1)} + \vd^{(2)}$, where $\vd^{(i)} = \ml^{(i)} \vz$, then the energies are related as:
	\[
		\energy{\ml}{\vd} = \energy{\ml^{(1)}}{\vd^{(1)}} + \energy{\ml^{(2)}}{\vd^{(2)}}.
	\]
\end{lemma}
\begin{proof}
	We have, by definition,
	\begin{align*}
		\energy{\ml^{(1)}}{\vd^{(1)}} + \energy{\ml^{(2)}}{\vd^{(2)}} 
		&= \vz^\top \ml^{(1)} \vz + \vz^\top \ml^{(2)} \vz \\
		&= \vz^\top \ml \vz  \\
		&= \mathcal{E}_\ml(\vd).
	\end{align*}
\end{proof}

The following lemma shows if $G'$ is a graph derived from $G$ by taking Schur complement on a subset of the vertices $C$, and $\vd$ is a demand supported on $C$, then the flow routing $\vd$ on $G$ will have lower energy than the flow routing $\vd$ on $G'$.

\begin{lemma}\label{lem:energy-schur-complement}
	Suppose $G$ is a weighted graph with Laplacian $\ml$. 
	Let $C$ be a subset of vertices of $G$.
	Let $\ml' = \tsc(\ml, C)$ be an $\eps$-approximate Schur complement. 
	Then for the demand $\vd = \ml' \vz$ supported on $C$, 
	\[
	 \energy{\ml}{\vd} \leq_\eps \energy{\ml'} {\vd}.
	\]
\end{lemma}
\begin{proof}
	We have, by definition,
	\begin{align*} 
		\energy{\ml}{\ml' \vz} & = \vz^\top \ml' {\ml}^{-1} \ml' \vz \\
		&\leq \vz^\top \ml' \sc(\ml, C)^{-1} \ml' \vz 
		\tag{since $\sc(\ml,C) \preccurlyeq \ml$}\\
		& \approx_{\eps} \vz^\top \ml' \ml'^{-1} \ml' \vz \\
		&= \mathcal{E}_{\ml'}(\ml' \vz).
	\end{align*}
\end{proof}

For any $H \in \ct$, we know $\tf^{(H)}$ routes $\vd^{(H)}$ using the original graph $G$. 
Furthermore, we know the graph of $\ml^{(H)}$ is related to $G$ using the graph operations considered above. Suppose ${\vf^{(H)}}^\star$ is the energy-minimizing flow routing $\vd^{(H)}$ on the graph of $\ml^{(H)}$. Then we want to relate the energies of $\tf^{(H)}$ and ${\vf^{(H)}}^\star$:

\begin{lemma} \label{lem:node-energy-bound}
	Let $H$ be a node at level $i$ in $\ct$. 
	Given \emph{any} $\vz$, 
	let $\vd \defeq \ml^{(H)} \vz$ be a demand.
	Then the weighted flow $\vf \defeq \mm^{(H)} \vz$ satisfies
	$\norm{\vf}_2^2 \leq_{i\epssc} \mathcal{E}_{\ml^{(H)}}(\vd)$. 
	
	Consequently, $\norm{\tf^{(H)}}_2^2 \leq_{i \epssc} \energy{\ml^{(H)}}{\vd^{(H)}}$.
\end{lemma}

\begin{proof}
	We proceed by induction.
	In the base case, $H$ is a leaf node, and we have
	\begin{align*}
	\norm{\mm^{(H)} \vz}_2^2 = {\vz}^\top (\mb[H])^\top \mw \mb[H] \vz = {\vz}^\top \ml^{(H)} \vz = \energy{\ml^{(H)}}{\vd}.
	\end{align*}
	
	Suppose $H$ is at level $i > 0$ in $\ct$, with children $D_1$ and $D_2$ at level at most $i-1$. Then
	\begin{align*}
		&\phantom{{} = {}} \norm{ \mm^{(H)} \vz }_2^2 \\
		&= \norm{ \left(\mm^{(D_1)} \mm_{(D_1,H)} + \mm^{(D_2)} \mm_{(D_2, H)} \right) \vz}_2^2 
		\intertext{Since $\range{\mm^{(D_1)}}$ and $\range{\mm^{(D_2)}}$ are orthogonal, we have}
		&= \norm{\mm^{(D_1)} \mm_{(D_1,H)} \vz}_2^2 +
		 \norm{\mm^{(D_2)} \mm_{(D_2,H)} \vz}_2^2 \\
		 &\leq_{(i-1)\epssc} \mathcal{E}_{\ml^{(D_1)}} \left(\ml^{(D_1)} \mm_{(D_1,H)}  \vz \right) + 
		 \mathcal{E}_{\ml^{(D_2)}} \left(\ml^{(D_2)} \mm_{(D_2,H)}  \vz \right) 
		 \tag{by inductive hypothesis with $\vz = \mm_{(D_i,H)} \vz$} \\
		 &= \mathcal{E}_{\ml^{(D_1)}} \left(\ml^{(D_1)} (\ml^{(D_1)})^{-1} \tsc(\ml^{(D_1)}, \bdry{D_1}) \vz \right) + \mathcal{E}_{\ml^{(D_2)}} \left(\ml^{(D_2)} (\ml^{(D_2)})^{-1} \tsc(\ml^{(D_2)}, \bdry{D_2}) \vz \right) \\
		 &\leq_\epssc \mathcal{E}_{\tsc(\ml^{(D_1)}, \bdry{D_1})} \left(\tsc(\ml^{(D_1)}, \bdry{D_1}) \vz \right) + 
		 \mathcal{E}_{\tsc(\ml^{(D_2)}, \bdry{D_2})} \left(\tsc(\ml^{(D_2)}, \bdry{D_2}) \vz \right) \\
		 &= \mathcal{E}_{\ml^{(H)}} (\ml^{(H)} \vz),
	\end{align*}
	where the last two inequalities follow from \cref{lem:energy-graph-decomposition,lem:energy-schur-complement}.
\end{proof}

Next, we want to relate the energy of routing $\vd^{(H)}$ on the graph $G$ and the energy on the graph of $\ml^{(H)}$.

\begin{lemma} \label{lem:energy-on-L-vs-LH}
	For a node $H$ at level $i$ in $\ct$, 
	\[
		\energy{\ml}{\vd^{(H)}} \approx_{\epssc} \energy{\ml^{(H)}}{\vd^{(H)}}.
	\]
	\qed
\end{lemma}

\begin{proof}
	For one direction, we have
	\begin{align*}
		\energy{\ml^{(H)}}{\vd^{(H)}} &= {\vd^{(H)}}^\top {\ml^{(H)}}^{-1} \vd^{(H)} \\
		&\approx_{\epssc} {\vd^{(H)}}^\top \sc(\ml[H], \bdry{H} \cup F_H)^{-1} \vd^{(H)} \tag{by \cref{thm:apxsc}} \\
		&\leq {\vd^{(H)}}^\top \ml^{-1} \vd^{(H)} 	\tag{since $\sc(\ml, C) \succcurlyeq \ml$} \\
		&=\energy \ml {\vd^{(H)}}.
	\end{align*}
	In the other direction, we note that $\tf^{(H)}$ is a weighted flow routing $\vd^{(H)}$ on $G$. By \cref{lem:node-energy-bound} and the definition of energy,
	\[
	\energy{\ml}{\vd^{(H)}} \leq \norm{\tf^{(H)}}_2^2 \approx_{i\epssc} \energy{\ml^{(H)}}{\vd^{(H)}}.
	\]
\end{proof}

We need to further bound the sum of energies: 

\begin{restatable}{lemma}{energySumBound}
	\label{lem:energy-sum-bound} 
	We have the following approximation of the energy of $\vd$ on graph $G$:
	\[
		\sum_{H \in \ct} \energy{\ml^{(H)}}{\vd^{(H)}} \approx_{\eta\epssc} \energy \ml \vd.
	\]
\end{restatable}

\begin{proof}
	We need the following matrix multiplication property:
	For any matrices $\ma,\mb,\md$,
	\begin{equation}\label{eq:matrices-top-left-corner}
		\left[\begin{array}{cc}
			\ma^{-1} & \mzero\\
			\mzero & \mzero
		\end{array}\right]\left[\begin{array}{cc}
			\ma & \mb\\
			\mb^{\top} & \md
		\end{array}\right]\left[\begin{array}{cc}
			\ma^{-1} & \mzero\\
			\mzero & \mzero
		\end{array}\right]=\left[\begin{array}{cc}
			\ma^{-1} & \mzero\\
			\mzero & \mzero
		\end{array}\right].
	\end{equation}
	
	Recall in our setting, all matrices are padded with zeros so that their dimension is $n \times n$, and vectors padded with zeros so their dimension is $n$.
	
	Define $\vbeta \defeq \mpi^{(\eta-1)}\cdots\mpi^{(1)}\mpi^{(0)} \vd$ for simplicity.
	Recall $\vz \defeq \widetilde{\mga} \mpi^{(\eta-1)} \cdots \mpi^{(0)} \mb^\top \mw^{1/2} \vv$.
	We can write
	\[
	\vz|_{F_H} = \left(\ml^{(H)}_{F_H, F_H} \right)^{-1} \vbeta.
	\]
	Then,
	\begin{align*}
		\energy{\ml^{(H)}}{\vd^{(H)}} &= \vz^\top|_{F_H} \ml^{(H)} \vz|_{F_H} \\
		&= \vbeta^\top \left(\ml^{(H)}_{F_H, F_H} \right)^{-1} \ml^{(H)} \left(\ml^{(H)}_{F_H, F_H} \right)^{-1} \vbeta \\
		&= \vbeta^\top \left(\ml^{(H)}_{F_H, F_H} \right)^{-1} \vbeta. \tag{by \cref{eq:matrices-top-left-corner}}
	\end{align*}
	
	Summing over all $H \in \ct$, we get
	\begin{align*}
		\sum_{H \in \ct} \energy{\ml^{(H)}}{\vd^{(H)}} &= \vbeta^\top \sum_{H \in \ct} (\ml^{(H)}_{F_H, F_H})^{-1} \vbeta \\
		&= \vd^\top  \mpi^{(0)\top}\cdots\mpi^{(\eta-1)\top}\left[\sum_H (\ml_{F_H,F_H}^{(H)})^{-1}\right] \mpi^{(\eta-1)}\cdots\mpi^{(0)} \vd \\ 
		&\approx_{\eta\epssc} \vd^\top \ml^{-1} \vd \\
		&= \mathcal{E}_{\ml}(\vd).
	\end{align*}
	where the last second step follows by \cref{thm:L-inv-approx}.
\end{proof}

Lastly, the following lemma shows that our weighted flow $\tf$ routing $\vd$ can be orthogonally decomposed in terms of the unique energy minimizer $\vf^\star$, which in turn allows us to bound $\|\tf - \vf^\star\|_2^2$.

\begin{lemma}\label{lem:f-orthogonal-decomposition}
	Let $\ml$ be a weighted Laplacian as above, and let $\vd$ be a demand.
	Let $\vf^\star = \mw^{1/2} \mb \ml^{-1} \vd$ be the weighted electric flow routing $\vd$ attaining the minimum energy $\energy \ml \vd$.
	For any other weighted flow $\tf$ satisfying $\mb^\top \mw^{1/2} \tf=\vd$, if $\|\tf\|_2^2\leq_\eps \mathcal{E}_{\ml}(\vd)$, then
	\[
	\|\tf - \vf^\star \|_2^2\leq 
	(e^{\eps}-1) \norm{\vf^\star}_2^2.
	\]
\end{lemma}
\begin{proof}
	Observe that
	\[
	\vf^{\star\top}(\tf-\vf^\star)= \vd^\top \ml^{-1} \mb^\top \mw^{1/2} (\tf - \vf^\star) = \vd^\top \ml^{-1}(\vd - \vd) = \vzero.
	\]
	Hence, we have an orthogonal decomposition of $\tf$:
	\[
	\|\tf\|_2^{2} =\|\tf^\star\|_2^{2}+\|\tf-\vf^\star\|_2^{2}.
	\]
	It follows that
	\[
	\|\vf-\vf^\star\|^{2} \leq (e^{\eps}-1) \cdot \norm{\vf^*}_2^2.
	\]
\end{proof}

Finally, we put all the lemmas together for the overall proof that $\tf$ is the desired weighted flow.

\begin{proof}[Proof of \cref{thm:flow-forest-operator-correctness}]
	
	We first decompose $\vd = \sum_{H \in \ct} \vd^{(H)}$ according to \cref{thm:demand-decomposition}.
	By definition of the flow tree operator, 
	\[
	\tf \defeq \mm \vz \defeq \sum_{H \in \ct} \mm^{(H)} \vz|_{F_H} = \sum_{H \in \ct} \tf^{(H)},
	\]
	where $\tf^{(H)} \defeq \mm^{(H)} \vz|_{F_H}$ routes demand $\vd^{(H)}$ by \cref{lem:flow-forest-operator-feasibility}.
	Hence, 
	\[
	 (\mw^{1/2} \mb)^\top \tf = \sum_{H \in \ct} (\mw^{1/2} \mb)^\top \tf^{(H)} = \sum_{H \in \ct} \vd^{(H)} = \vd,
	\]
	meaning $\tf$ is feasible for routing $\vd$ on $G$.
	
	For each demand term $\vd^{(H)}$, let $\vf^{(H)\star}$ be the weighted flow on $G$ that attains the minimum energy $\mathcal{E}_{\ml}(\vd^{(H)})$ for routing it. 
	By \cref{thm:energy-minimizer} , $\vf^{(H)\star} = \mw^{1/2} \mb \ml^{-1} \vd^{(H)}$.	
	Recall $\vf^\star \defeq \mproj_{\vw} \vv = \mw^{1/2}\mb \ml^{-1} \vd$.
	Hence, 
	\[
	\vf^\star = \sum_{H \in \ct} \vf^{(H)\star}.
	\]
	
	By \cref{lem:energy-on-L-vs-LH}, 
	we know if $H$ is at level $i$ in $\ct$, then $\tf^{(H)}$ satisfies
	\begin{equation} \label{eq:tf^H-approxes-f^Hstar}
	\norm{\tf^{(H)}}_2^2 \leq_{i\epssc} \energy {\ml^{(H)}}{\vd^{(H)}} \approx_{i \epssc} \energy \ml {\vd^{(H)}} = \norm{\vf^{(H)\star}}_2^2.
	\end{equation}
	This shows that in the flow tree operator, the output $\tf^{(H)}$ of each tree operator $\mm^{(H)}$ is close to the natural corresponding term ${\vf^{(H)}}^\star$.
	Finally, we bound the overall approximation error:
	\begin{align*}
		\norm{\tf - \vf^\star}_2^2 &= \norm{ \sum_{H \in \ct} \left( \tf^{(H)}-\vf^{(H) \star} \right) }_2^{2} \\
		&\leq \left( \sum_{H \in \ct} \norm{\tf^{(H)}-\vf^{(H) \star} }_2  \right)^2 \\		
		&= \sum_{i=0}^{\eta}  \sum_{H \in \ct(i)} (e^{2i\epssc}-1) \mathcal{E}_{\ml} (\vd^{(H)})
		\tag{by \cref{lem:f-orthogonal-decomposition} and  \cref{eq:tf^H-approxes-f^Hstar}}\\
		&\leq \sum_{i=0}^{\eta}  \sum_{H \in \ct(i)} (e^{2i\epssc}-1) e^{i \epssc} \mathcal{E}_{\ml^{(H)}} (\vd^{(H)}) 
		\tag{by \cref{lem:energy-on-L-vs-LH}}\\		
		&\leq  e^{4 \eta \epssc} \sum_{H \in \ct} \mathcal{E}_{\ml} (\vd) 
		\tag{by \cref{lem:energy-sum-bound}}\\
		&=  O(\eta \epssc) \norm{\vf^\star}^2,
	\end{align*}
	which concludes the overall proof.
\end{proof}

\subsection{Proof of \crtcref{thm:FlowMaintain}}

Finally, we present the overall flow maintenance data structure. 
It is analogous to slack, except during each \textsc{Move} operation, there is an additional term of $\alpha \mw^{1/2} \vv$.

\label{subsec:flow-main-proof}
\FlowMaintain*
\begin{algorithm}
	\caption{Flow Maintenance, Main Algorithm}\label{alg:flow-maintain-main}
	\begin{algorithmic}[1]
		\State \textbf{data structure} \textsc{MaintainFlow} \textbf{extends} \textsc{MaintainZ}
		\State \textbf{private: member}
		\State \hspace{4mm} $\vw \in \R^m$: weight vector \Comment we use the diagonal matrix $\mw$ interchangeably
		\State \hspace{4mm} $\vv \in \R^m$: direction vector
		
		\State \hspace{4mm} \textsc{MaintainRep} $\maintainRep$: data structure to implicitly maintain
		\[
		\pf \defeq \vy + \mw^{1/2} \mm (c \zprev + \zsum).
		\]
		\Comment $\mm$ is defined by \cref{defn:flow-forest-operator}
		\State \hspace{4mm} $\hat c \in \R, \uf_0 \in \R^m$: scalar and vector to implicitly maintain
		\[
			\uf \defeq \uf_0 + \hat c \cdot  \mw \vv.
		\]
		\State \hspace{4mm} \textsc{MaintainApprox} \texttt{bar\_f}: data structure to maintain approximation $\of$ to $\vf$ (\cref{thm:VectorTreeMaintain})
		\State
		\Procedure{Initialize}{$G,\vf^{\init} \in\R^{m},\vv\in \R^{m}, \vw\in\R_{>0}^{m},\epssc>0,
		\overline{\epsilon}>0$}
			\State Build the separator tree $\ct$ by \cref{thm:separator_tree_construction} 
			\State $\maintainRep.\textsc{Initialize}(G, \ct, \mw^{1/2}\mm, \vv, \vw, \vzero, \epssc)$
			\Comment{initialize $\pf \leftarrow \vzero$}
			\State $\vw \leftarrow \vw, \vv \leftarrow \vv$
			\State $\hat{c}\leftarrow 0, \uf_0\leftarrow \vf^\init$ \Comment initialize $\uf  \leftarrow \vf^{\init}$\label{line:flowinitf0}
			\State $\texttt{bar\_f}.\textsc{Initialize}(-\mw^{1/2}\mm,c, \zprev,\zsum, -\vy+\uf_0+\hat{c}\cdot \mw \vv, \mw^{-1}, n^{-5}, \overline{\epsilon})$
			\State \Comment{initialize $\of \leftarrow \vf^{\init}$}

		\EndProcedure
		
		\State
		\Procedure{Reweight}{$\vw^\new \in \R^m_{>0}$}
			\State $\maintainRep.\textsc{Reweight}(\vw^\new)$
			\State $\Delta \vw \leftarrow \vw^{\new} - \vw$
			\State $\vw \leftarrow \vw^\new$
			\State $\uf_0 \leftarrow \uf_0 - \hat{c} (\Delta \mw)^{1/2} \vv$
		\EndProcedure

		\State
		\Procedure{Move}{$\alpha,\vv^\new \in \R^m$}
			\State $\maintainRep.\textsc{Move}(\alpha, \vv^\new)$
			\State $\Delta \vv \leftarrow \vv^\new - \vv$
			\State $\vv\leftarrow \vv^\new$
			\State $\uf_0 \leftarrow \uf_0 - \hat{c} \mw^{1/2} \Delta \vv$
			\State $\hat{c}\leftarrow\hat{c}+\alpha$
		\EndProcedure

		\State
		\Procedure{Approximate}{ }
		\State \Comment the variables in the argument are accessed from \maintainRep
		\State	\Return $\texttt{bar\_f}.\textsc{Approximate}({-\mw^{1/2}\mm},c, \zprev,\zsum,-\vy+\uf_0+\hat{c}\cdot \mw \vv, \mw^{-1})$
		\EndProcedure

		\State
		\Procedure{Exact}{ }
		\State  $\pf \leftarrow \maintainRep.\textsc{Exact}()$
		\State \Return $(\uf_0+\hat{c}\cdot \mw \vv)- \pf$
		\EndProcedure
	\end{algorithmic}
\end{algorithm}

\begin{proof}[Proof of \cref{thm:FlowMaintain}]
We have the additional invariant that the IPM flow solution $\vf$ can be recovered in the data structure by the identity
\begin{align} \label{eq:vf-decomp}
	\vf &= \uf - \pf,
\end{align}
where $\pf$ is implicit maintained by \maintainRep, and $\uf$ is implicitly maintained by the identity $\uf = \uf_0 + \hat c \mw \vv$.

We prove the runtime and correctness of each procedure separately. 
Recall by \cref{lem:slack-edge-operator-correctness}, the tree operator $\mm$ has complexity $T(K) = O(\epssc^{-2} \sqrt{mK})$.

\paragraph{\textsc{Initialize}:}
By the initialization of \maintainRep~(\cref{thm:maintain_representation}), the implicit representation of $\pf$ in \maintainRep~is correct and $\pf = \vzero$. 
We then set $\uf \defeq \uf_0 + \hat c \mw \vv = \vf^\init$. So overall, we have $\vf \defeq \uf + \pf = \vf^\init$.
By the initialization of \flowSketch, $\of$ is set to $\vf = \vf^\init$ to start.

Initialization of $\maintainRep$ takes $\otilde(m\epssc^{-2})$ time by \cref{thm:maintain_representation},
and the initialization of $\flowSketch$ takes $\otilde(m)$ time by \cref{thm:VectorTreeMaintain}. 

\paragraph{\textsc{Reweight}:}
The change to the representation in $\pf$ is correct via \maintainRep~in exactly the same manner as the proof for the slack solution.
For the representation of $\uf$, the change in value caused by the update to $\vw$ is subtracted from the $\uf_0$ term, so that the representation is updated while the overall value remains the same.

\paragraph{\textsc{Move}:}
	This is similar to the proof for the slack solution.
	$\maintainRep.\textsc{Move}(\alpha, \vv^{(k)})$ updates the implicit representation of $\pf$ by 
	\[
	\pf \leftarrow \pf + \mw^{1/2} \mm \alpha \vz^{(k)},
	\]
	where $\mm$ is the flow projection tree operator defined in \cref{defn:flow-forest-operator}.
	By \cref{lem:slack-operator-correctness}, this is equivalent to the update
	\[
	\pf \leftarrow \pf + \alpha \mw^{1/2} \tf,
	\]
	where $\norm{\tf - \mproj_{\vw} \vv^{(k)}}_2 \leq O(\eta \epssc) \norm{\vv^{(k)}}_2$ and $\mb^{\top} \mw^{1/2} \tf = \mb^{\top} \mw^{1/2} \vv^{(k)}$ by \cref{thm:flow-forest-operator-correctness}.
	
	For the $\uf$ term, let $\uf_0', \hat c', \vv'$ be the state of $\uf_0, \hat c$ and $\vv$ at the start of the procedure, and similarly let $\uf'$ be the state of $\uf$ at the start.
	At the end of the procedure, we have
	\[
		\uf \defeq \uf_0 + \hat c \mw \vv = \uf_0' - \hat c' \mw^{1/2} \Delta \vv + (\hat c' + \alpha) \mw \vv = \uf_0' + \hat c' \mw^{1/2} \vv' + \alpha \mw^{1/2} \vv = \uf' + \alpha \mw^{1/2} \vv,
	\]
	so we have the correct update $\uf \leftarrow \uf + \alpha \mw^{1/2} \vv$. Combined with $\pf$, the update to $\vf$ is
	\[
		\vf \leftarrow \vf + \alpha \mw^{1/2} \vv - \alpha \mw^{1/2} \tf.
	\]
	
	By \cref{thm:maintain_representation}, if $\vv^{(k)}$ differs from $\vv^{(k-1)}$ on $K$ coordinates, then the runtime of \maintainRep~is $\O(\epssc^{-2}\sqrt{mK})$. Furthermore, $\zprev$ and $\zsum$ change on $\elim{H}$ for at most $\O(K)$ nodes in $\ct$. Updating $\uf$ takes $O(K)$ time where $K \leq O(m)$, giving us the overall claimed runtime.

\paragraph{\textsc{Approximate}:}
By the guarantee of \texttt{bar\_f}.\textsc{Approximate} from \cref{thm:VectorTreeMaintain},
the returned vector satisfies 
$\|\mw^{-1/2}\left( \of - (\uf - \pf)\right)\|_{\infty}\leq\overline{\epsilon}$,
where $\uf$ and $\pf$ are maintained in the current data structure.

\paragraph{\textsc{Exact}:}
The runtime and correctness follow from the guarantee of $\maintainRep.\textsc{Exact}$ given in \cref{thm:maintain_representation} and the invariant that $\vf = \uf - \pf$.

Finally, we have the following lemma about the runtime for \textsc{Approximate}. Let $\of^{(k)}$ denote the returned approximate vector at step $k$.
\begin{lemma}
	Suppose $\alpha\|\vv\|_{2}\leq\beta$ for some $\beta$ for all calls to \textsc{Move}.
	Let $K$ denote the total number of coordinates changed in $\vv$ and $\vw$ between the $k-1$-th and $k$-th \textsc{Reweight} and \textsc{Move} calls. Then at the $k$-th \textsc{Approximate} call,
	\begin{itemize}
		\item The data structure first sets $\of_e\leftarrow \vf^{(k-1)}_e$ for all coordinates $e$ where $\vw_e$ changed in the last \textsc{Reweight}, then sets $\of_e\leftarrow \vf^{(k)}_e$ for $O(N_k\defeq 2^{2\ell_{k}}(\frac{\beta}{\overline{\epsilon}})^{2}\log^{2}m)$ coordinates $e$, where $\ell_{k}$ is the largest integer
		$\ell$ with $k=0\mod2^{\ell}$ when $k\neq 0$ and $\ell_0=0$. 
		\item The amortized time for the $k$-th \textsc{Approximate} call
		is $\widetilde{O}(\epsilon_{\mproj}^{-2}\sqrt{m(K+N_{k-2^{\ell_k}})})$.
	\end{itemize}
\end{lemma}
\begin{proof}

	The proof is similar to the one for slack.
	Since ${\of}$ is maintained by \texttt{bar\_f}, we apply \cref{thm:VectorTreeMaintain} with $\vx = \of$ and diagonal matrix $\md = \mw^{-1}$. 
	We need to prove $\|\vx^{(k)}-\vx^{(k-1)}\|_{\md^{(k)}}\le O(\beta)$ for all $k$ first. The constant factor in $O(\beta)$ does not affect the guarantees in \cref{thm:VectorTreeMaintain}. The left-hand side is
	\begin{align*}
		\norm{{\vf}^{(k)}-{\vf}^{(k-1)}}_{{\mw^{(k)}}^{-1}} &= \norm{-\alpha^{(k)} \mm \vz^{(k)}+\alpha^{(k)}\vv^{(k)}}_2 \tag{by \textsc{Move}}\\
		&\le  \norm{-\alpha^{(k)} \mm \vz^{(k)}}_2+\norm{\alpha^{(k)}\vv^{(k)}}_2\\
		&\le (2+O(\eta\epssc)) \alpha^{(k)} \|\vv^{(k)}\|_2 \tag{by the assumption that $\alpha\|\vv\|_{2}\leq\beta$}\\
		&\le 3 \beta.
	\end{align*}
	Now, we can apply the conclusions from \cref{thm:VectorTreeMaintain} to get that at the $k$-th step, the data structure first sets $\of_e\leftarrow \vf^{(k-1)}_e$ for all coordinates $e$ where $\vw_e$ changed in the last \textsc{Reweight}, then sets $\of_e\leftarrow \vf^{(k)}_e$ for $O(N_k\defeq 2^{2\ell_{k}}(\frac{\beta}{\overline{\epsilon}})^{2}\log^{2}m)$ coordinates $e$, where $\ell_{k}$ is the largest integer
		$\ell$ with $k=0\mod2^{\ell}$ when $k\neq 0$ and $\ell_0=0$. 
	
	For the second point, \textsc{Move} updates $\zprev$ and $\zsum$ on $\elim{H}$ for $\O(K)$ different nodes $H \in \ct$ by \cref{thm:maintain_representation}.
	\textsc{Reweight} then updates $\zprev$ and $\zsum$ on $F_H$ for $\O(K)$ different nodes, and updates the tree operator $\mw^{-1/2}\mm$ on $\O(K)$ different edge and leaf operators. In turn, it updates $\vy$ on $E(H)$ for $\O(K)$ leaf nodes $H$. The changes of $\uf$ cause $O(K)$ changes to the vector $-\vy+\uf_0+\hat{c}\cdot \mw \vv$, which is the parameter $\vy$ of \cref{thm:VectorTreeMaintain}. Now, we apply \cref{thm:VectorTreeMaintain} and the complexity of the tree operator to conclude the desired amortized runtime. 
\end{proof}

\end{proof}

\section{Min-Cost Flow for Separable Graphs}
\label{sec:separable}
In this section, we extend our result to $\sepConst$-separable graphs.

\separableMincostflow*

The change in running time essentially comes from the parameters of the separator tree which we shall discuss in \cref{subsec:separable_separator_tree}. We then calculate the total running time and prove \cref{thm:separable_mincostflow} in \cref{subsec:separable_main}.

\subsection{Separator Tree for Separable Graphs}
\label{subsec:separable_separator_tree}

Since our algorithm only exploits the separable property of the planar graphs, it can be applied to other separable graphs directly and yields different running times. Similar to the planar case, by adding two extra vertices to any $\sepConst$-separable graph, it is still $\sepConst$-separable with the constant $c$ in \cref{defn:separable-graph} increased by $2$. 

Recall the definition of separable graphs: 

\defSeparableGraph*

We define a separator tree $\ct$ for an $\sepConst$-separable graph $G$ in the same way as for a planar graph.
\begin{definition} [Separator tree $\ct$ for $\sepConst$-separable graph] \label{defn:separator-tree-separable}
Let $G$ be an $\sepConst$-separable graph. 
A separator tree $\ct$ is a binary tree whose nodes represent subgraphs of $G$ such that the children of each node $H$ form a balanced partition of $H$.

Formally, each node of $\ct$ is a \emph{region} (edge-induced subgraph) $\region$ of $G$; we denote this by $\region \in \ct$. 
At a node $\region$, we store subsets of vertices $\bdry{\region}, \sep{\region}, \elim{\region} \subseteq V(\region)$, 
where $\bdry{\region}$ is the set of \emph{boundary vertices} that are incident to vertices outside $\region$ in $G$;
$\sep{\region}$ is the balanced vertex separator of $\region$;
and $\elim{\region}$ is the set of \emph{eliminated vertices} at $\region$. 
Concretely, the nodes and associated vertex sets are defined recursively in a top-down way as follows: 

	\begin{enumerate}
		\item The root of $\ct$ is the node $\region = G$, with $\bdry{\region} = \emptyset$ and $\elim{\region} = \sep{\region}$.
		\item A non-leaf node $\region \in \ct$ has exactly two children $D_1, D_2 \in \ct$ that form an edge-disjoint partition of $\region$, and their vertex sets intersect on the balanced separator $\sep{\region}$ of $\region$.
		Define $\bdry{D_1} = (\bdry{\region} \cup \sep{\region}) \cap V(D_1)$, and similarly $\bdry{D_2} = (\bdry{\region} \cup \sep{\region}) \cap V(D_2)$.
		Define $\elim{\region} = \sep{\region} \setminus \bdry{\region}$.\label{property: boundary-separable}
		
		\item If a region $\region$ contains a constant number of edges, then we stop the recursion and $\region$ becomes a leaf node. Further, we define $\sep{\region} = \emptyset$ and $\elim{\region} = V(\region) \setminus \bdry{\region}$. Note that by construction, each edge of $G$ is contained in a unique leaf node. 
	\end{enumerate}

Let $\eta(H)$ denote the height of node $H$ which is defined as the maximum number of edges on a tree path from $H$ to one of its descendants. $\eta(H)=0$ if $H$ is a leaf. Note that the height difference between a parent and child node could be greater than one. Let $\eta$ denote the height of $\ct$ which is defined as the maximum height of nodes in $\ct$. We say $H$ is at \emph{level} $i$ if $\eta(H)=i$. 
\end{definition}

The only two differences between the separator trees for planar and $\sepConst$-separable graphs are their construction time and update time (for $k$-sparse updates). For the planar case, these are bounded by \cref{thm:separator_tree_construction} and \cref{lem:planarBoundChangeCost} respectively. We shall prove their analogs \cref{thm:separableDecompT} and \cref{lem:separableBoundChangeCost}.

\cite{GHP18} showed that the separator tree can be constructed in $O(s(n)\log n)$ time for any class of $1/2$-separable graphs where $s(n)$ is the time for computing the separator. The proof can be naturally extended to $\sepConst$-separable graphs. We include the extended proofs in \cref{sec:appendix} for completeness. 

\begin{restatable}{lemma}{separableDecompT} \label{thm:separableDecompT} Let $\mathcal{C}$ be an $\sepConst$-separable class such that we can compute a balanced separator for any graph in $\mathcal{C}$ with $n$ vertices and $m$ edges in $s(m)$ time for some convex function $s(m)\ge m$. Given an $\sepConst$-separable graph, there is an algorithm that computes a separator tree $\ct$ in $O(s(m) \log m)$ time.
\end{restatable}

Note that $s(\cdot)$ does not depend on $n$ because we may assume the graph is connected so that $n=O(m)$.

We then prove the update time. Same as the planar case, we define $\pathT{\region}$ to be the set of all ancestors of $\region$ in the separator tree and $\pathT{\collN}$ to be the union of $\pathT{\region}$ for all $\region \in \collN$. Then we have the following bound:
\begin{restatable}{lemma}{separableBoundChangeCost}
\label{lem:separableBoundChangeCost}
Let $G$ be an $\alpha$-separable graph with separator tree $\ct$. Let $\collN$ be a set of $K$ nodes in $\ct$. Then 
\begin{align*}
	\sum_{\region \in \pathT{\collN}}| \bdry{\region}| +|\sep{\region}| \leq \otilde(\septime). 
\end{align*}
\end{restatable}
By setting $\sepConst$ as $1/2$, we get \cref{lem:planarBoundChangeCost} for planar graphs as a corollary.

\subsection{Proof of Running time}
In this section, we prove \cref{thm:separable_mincostflow}. The data structures (except for the construction of the separator tree) will use exactly the same pseudocode as for the planar case. Thus, the correctness can be proven in the same way. We prove the runtimes only. 

For the planar case, after constructing the separator tree by \cref{thm:separator_tree_construction}, \cref{lem:planarBoundChangeCost} is the lemma that interacts with other parts of the algorithm. For $\sepConst$-separable graphs, we first construct the separator tree in $O(s(m)\log m)$ time by \cref{thm:separableDecompT}. Then we propagate the change in runtime 
($\otilde(\sqrt{mK})$ from \cref{lem:planarBoundChangeCost} to $\otilde(m^\sepConst K^{1-\sepConst})$ from \cref{lem:separableBoundChangeCost}) 
to all the data structures and to the complexity $T(\cdot)$ of the flow and slack tree operators.

We first propagate the change to the implicit representation maintenance data structure, which is the common component for maintaining the flow and the slack vectors. 

\begin{theorem} \label{thm:separable_maintain_representation}
	Given an $\sepConst$-separable graph $G$ with $n$ vertices and $m$ edges, and its separator tree $\ct$ with height $\eta$,
	the deterministic data structure \textsc{MaintainRep} (\cref{alg:maintain_representation})
	 maintains the following variables correctly at the end of every IPM step:
	\begin{itemize}
		\item the dynamic edge weights $\vw$ and step direction $\vv$ from the current IPM step,
		\item a \textsc{DynamicSC} data structure on $\ct$ based on the current edge weights $\vw$,
		\item an implicitly represented tree operator $\mm$ supported on $\ct$ with complexity $T(K)$, \emph{computable using information from \textsc{DynamicSC}},
		\item scalar $c$ and vectors $\zprev, \zsum$, which together represent $\vz = c \zprev + \zsum$,
		such that at the end of step $k$,
		\[
		\vz = \sum_{i=1}^{k} \alpha^{(i)} \vz^{(i)},
		\]
		where $\alpha^{(i)}$ is the step size $\alpha$ given in \textsc{Move} for step $i$,
		\item $\zprev$ satisfies $\zprev = \widetilde{\mga} \mpi^{(\eta-1)} \cdots \mpi^{(0)} \mb^{\top} \mw^{1/2} \vv$,
		\item an offset vector $\vy$ which together with $\mm, \vz$ represent $\vx=\vy+\mm\vz$, such that after step $k$,
		\[
			\vx = \vx^{\init}+\sum_{i=1}^{k} \mm^{(i)} (\alpha^{(i)} \vz^{(i)}),
		\]
		where $\vx^{\init}$ is an initial value from \textsc{Initialize}, and $\mm^{(i)}$ is the state of $\mm$ after step $i$.
	\end{itemize}
	The data structure supports the following procedures:
	\begin{itemize}
		\item $\textsc{Initialize}(G, \ct, \mm, \vv\in\R^{m},\vw\in\R_{>0}^{m}, \vx^{\init} \in \R^m, \epsilon_{\mproj} > 0)$:
		Given a graph $G$, its separator tree $\ct$, a tree operator $\mm$ supported on $\ct$ with complexity $T$,
		initial step direction $\vv$, initial weights $\vw$, initial vector $\vx^{\init}$, and target projection matrix accuracy $\epsilon_{\mproj}$, preprocess in $\widetilde{O}(\epssc^{-2}m+T(m))$ time and set $\vx \leftarrow \vx^{\init}$.

		\item $\textsc{Reweight}(\vw \in\R_{>0}^{m}$ given implicitly as a set of changed coordinates):
		Update the weights to $\vw^\new$.
		Update the implicit representation of $\vx$ without changing its value, so that all the variables in the data structure are based on the new weights.

		The procedure runs in 
		$\widetilde{O}(\epsilon_{\mproj}^{-2}\septime+T(K))$ total time,
		where $K$ is an upper bound on the number of coordinates changed in $\vw$ and the number of leaf or edge operators changed in $\mm$.
		There are most $\O(K)$ nodes $\region\in \ct$ for which $\zprev|_{F_H}$ and $\zsum|_{F_H}$ are updated.

		\item $\textsc{Move}(\alpha \in \R$, $\vv \in \R^{n}$ given implicitly as a set of changed coordinates):
		Update the current direction to $\vv$, and then $\zprev$ to maintain the claimed invariant.
		Update the implicit representation of $\vx$ to reflect the following change in value:
		\[
			\vx \leftarrow \vx + \mm (\alpha \zprev).
		\]
		The procedure runs in $\widetilde{O}(\epsilon_{\mproj}^{-2} \septime)$ time,
		where $K$ is the number of coordinates changed in $\vv$ compared to the previous IPM step.

		\item $\textsc{Exact}()$:
		Output the current exact value of $\vx=\vy + \mm \vz$ in $\O(T(m))$ time.
	\end{itemize}
\end{theorem}

\begin{proof}
The bottlenecks of \textsc{Move} is \textsc{PartialProject}. For each $H \in \mathcal{P}_{\ct}(\mathcal{H})$, recall from \cref{thm:apxsc} that $\ml^{(\region)}$ is supported on the vertex set $\elim{\region} \cup \bdry{\region}$ and has $\O(\epssc^{-2}|\elim{\region}\cup \bdry{\region}|)$ edges. Hence,
	$(\ml^{(\region)}_{\elim{\region},\elim{\region}})^{-1} \vu|_{\elim{\region}}$ can be computed by an exact Laplacian solver in $\O(\epssc^{-2}|\elim{\region}\cup \bdry{\region}|)$ time, and the subsequent left-multiplying by $\ml^{(\region)}_{\bdry{\region}, \elim{\region}}$ also takes $\O(\epssc^{-2}|\elim{\region}\cup \bdry{\region}|)$ time. By \cref{lem:separableBoundChangeCost}, \textsc{PartialProject} takes $\O(\epssc^{-2}\septime)$ time. \textsc{Move} also runs in $\O(\epssc^{-2}\septime)$ time.
	
\textsc{Reweight} calls \textsc{PartialProject} and \textsc{ReversePartialProject} for $O(1)$ times and \textsc{ComputeMz} once. \textsc{ReversePartialProject} costs the same as \textsc{PartialProject}. The runtime of \textsc{ComputeMz} is still bounded by the complexity of the tree operator, $O(T(K))$. Thus, \textsc{PartialProject} takes $\O(\epssc^{-2}\septime)$ time. \textsc{Move} also runs in $\O(\epssc^{-2}\septime+T(K))$ time.

 Runtimes of other procedures and correctness follow from the same argument as in the proof for \cref{thm:maintain_representation}.
\end{proof}
Then we may use \cref{thm:separable_maintain_representation} and \cref{thm:VectorTreeMaintain} to maintain vectors $\of, \os$, with the updated complexity of the operators.
\begin{lemma}
\label{lem:separable_operator_complexity}
For any $\sepConst$-separable graph $G$ with separator tree $\ct$, the flow and slack operators defined in \cref{defn:flow-forest-operator,defn:slack-forest-operator} both have complexity $T(K)=O(\epssc^{-2}\septime)$.
\end{lemma}
\begin{proof}
	The leaf operators of both the flow and slack tree operators has constant size. 
	Let $\mm_{(H, P)}$ be a tree edge operator. Note that it is a symmetric matrix.
	For the slack operator, 
	Applying $\mm_{(D,P)}=\mi_{\bdry{D}}-\left(\ml^{(D)}_{\elim{D},\elim{D}}\right)^{-1} \ml^{(D)}_{\elim{D}, \bdry{D}}$ 
	to the left or right consists of three steps which are applying $\mi_{\bdry{D}}$, 
	applying $\ml^{(D)}_{\elim{D}, \bdry{D}}$ and solving for $\ml^{(D)}_{\elim{D},\elim{D}}\vv=\vb$ for some vectors $\vv$ and $\vb$. 
	For the flow operator, $\mm_{(H, P)} \vu$ consists of multiplying with $\tsc(\ml^{(H)}, \partial H)$ and solving the Laplacian system $\ml^{(H)}$.
	
	Each of the steps costs time $O(\epssc^{-2}|\bdry{D}\cup \elim{D}|)$ by \cref{lem:fastApproxSchur} and \cref{thm:laplacianSolver}.
	To bound the total cost over $K$ distinct edges, we apply \cref{lem:separableBoundChangeCost} instead of \cref{lem:planarBoundChangeCost}, which gives the claimed complexity.
\end{proof}
We then have the following lemmas for maintaining the flow and slack vectors:

\begin{theorem}[Slack maintenance for $\alpha$-separable graphs]
	  \label{thm:separableSlackMaintain}
	Given a modified planar graph $G$ with $n$ vertices and $m$ edges, and its separator tree $\ct$ with height $\eta$,
	the randomized data structure \textsc{MaintainSlack} (\cref{alg:slack-maintain-main})
	implicitly maintains the slack solution $\vs$ undergoing IPM changes,
	and explicitly maintains its approximation $\os$,
	and supports the following procedures with high probability against an adaptive adversary:
	\begin{itemize}
		\item $\textsc{Initialize}(G,\vs^{\init} \in\R^{m},\vv\in \R^{m}, \vw\in\R_{>0}^{m},\epsilon_{\mproj}>0,\overline{\epsilon}>0)$:
		Given a graph $G$, initial solution $\vs^{\init}$, initial direction $\vv$, initial weights $\vw$,
		target step accuracy $\epsilon_{\mproj}$ and target approximation accuracy
		$\overline{\epsilon}$, preprocess in $\widetilde{O}(m \epsilon_{\mproj}^{-2})$ time, 
		and set the representations $\vs \leftarrow \vs^{\init}$ and $\ox \leftarrow \vs$.
		
		\item $\textsc{Reweight}(\vw\in\R_{>0}^{m},$ given implicitly as a set of changed weights): 
		Set the current weights to $\vw$ in $\widetilde{O}(\epsilon_{\mproj}^{-2}\septime)$ time,
		where $K$ is the number of coordinates changed in $\vw$.
		
		\item $\textsc{Move}(t \in\mathbb{R},\vv\in\R^{m} $ given implicitly as a set of changed coordinates): 
		Implicitly update  $\vs \leftarrow \vs+t \mw^{-1/2}\widetilde{\mproj}_{\vw} \vv$ for some
		$\widetilde{\mproj}_{\vw}$ with	$\|(\widetilde{\mproj}_{\vw} -\mproj_{\vw}) \vv \|_2 \leq\eta\epssc \norm{\vv}_2$,
		and $\widetilde{\mproj}_{\vw} \vv \in \range{\mb}$. 
		The total runtime is $\widetilde{O}(\epsilon_{\mproj}^{-2}\septime)$ where $K$ is the number of coordinates changed in $\vv$.
		
		\item $\textsc{Approximate}()\rightarrow\R^{m}$: Return the vector $\os$
		such that $\|\mw^{1/2}(\os-\vs)\|_{\infty}\leq\overline{\epsilon}$
		for the current weight $\vw$ and the current vector $\vs$. 
		
		\item $\textsc{Exact}()\rightarrow\R^{m}$: 
		Output the current vector $\vs$ in $\O(m \epssc^{-2})$ time.
	\end{itemize}
	Suppose $t\|\vv\|_{2}\leq\beta$ for some $\beta$ for all calls to \textsc{Move}.
	Suppose in each step, \textsc{Reweight}, \textsc{Move} and \textsc{Approximate} are called in order. Let $K$ denote the total number of coordinates changed in $\vv$ and $\vw$ between the $(k-1)$-th and $k$-th \textsc{Reweight} and \textsc{Move} calls. Then at the $k$-th \textsc{Approximate} call,
	\begin{itemize}
		\item the data structure first sets $\os_e\leftarrow \vs^{(k-1)}_e$ for all coordinates $e$ where $\vw_e$ changed in the last \textsc{Reweight}, then sets $\os_e\leftarrow \vs^{(k)}_e$ for $O(N_k\defeq 2^{2\ell_{k}}(\frac{\beta}{\overline{\epsilon}})^{2}\log^{2}m)$ coordinates $e$, where $\ell_{k}$ is the largest integer
		$\ell$ with $k=0\mod2^{\ell}$ when $k\neq 0$ and $\ell_0=0$. 
		\item The amortized time for the $k$-th \textsc{Approximate} call
		is $\widetilde{O}(\epsilon_{\mproj}^{-2}(m^{\sepConst}(K+N_{k-2^{\ell_k}})^{1-\sepConst}))$.
	\end{itemize}
\end{theorem}
\begin{proof}
Because $T(m)=\O(\epssc^{-2}m)$ (\cref{lem:separable_operator_complexity}), the runtime of \textsc{Initialize} is still $\O(\epssc^{-2}m)$ by \cref{thm:separable_maintain_representation} and \cref{thm:VectorTreeMaintain}. The runtime of \textsc{Reweight}, \textsc{Move}, and \textsc{Exact} follow from the guarantees of \cref{thm:separable_maintain_representation}. The runtime of \textsc{Approximate} follows from \cref{thm:VectorTreeMaintain} with $T(K)=O(\septime)$ (\cref{lem:separable_operator_complexity}).
\end{proof}

\begin{theorem}[Flow maintenance for $\alpha$-separable graphs]
  \label{thm:separableFlowMaintain}
	Given a $\alpha$-separable graph $G$ with $n$ vertices and $m$ edges, and its separator tree $\ct$ with height $\eta$,
	the randomized data structure \textsc{MaintainFlow} (\cref{alg:flow-maintain-main})
	implicitly maintains the flow solution $\vf$ undergoing IPM changes,
	and explicitly maintains its approximation $\of$, 
	and supports the following procedures with high probability against an adaptive adversary:
	\begin{itemize}
		\item
		$\textsc{Initialize}(G,\vf^{\init} \in\R^{m},\vv \in \R^{m}, \vw\in\R_{>0}^{m},\epsilon_{\mproj}>0,
		\overline{\epsilon}>0)$: Given a graph $G$, initial solution $\vf^\init$, initial direction $\vv$, 
		initial weights $\vw$, target step accuracy $\epsilon_{\mproj}$,
		and target approximation accuracy $\overline{\epsilon}$, 
		preprocess in $\widetilde{O}(m \epsilon_{\mproj}^{-2})$ time and set the internal representation $\vf \leftarrow \vf^{\init}$ and $\of \leftarrow \vf$.
		\item $\textsc{Reweight}(\vw\in\R_{>0}^{m}$ given implicitly as a set of changed weights): Set the current weights to $\vw$ in
		$\widetilde{O}(\epsilon_{\mproj}^{-2}\sepConst)$ time, where $K$ is
		the number of coordinates changed in $\vw$.
		\item $\textsc{Move}(t\in\mathbb{R},\vv\in\R^{m}$ given
		implicitly as a set of changed coordinates): 
		Implicitly update
		$\vf \leftarrow \vf+ t \mw^{1/2}\vv - t \mw^{1/2} \widetilde{\mproj}'_{\vw} \vv$ for
		some $\widetilde{\mproj}'_{\vw} \vv$, 
		where 
		$\|\widetilde{\mproj}'_{\vw} \vv - \mproj_{\vw} \vv \|_2 \leq O(\eta \epssc) \norm{\vv}_2$ and
		$\mb^\top \mw^{1/2}\widetilde{\mproj}'_{\vw}\vv= \mb^\top \mw^{1/2} \vv$.
		The runtime is $\widetilde{O}(\epsilon_{\mproj}^{-2} \septime)$, where $K$ is
		the number of coordinates changed in $\vv$.
		\item $\textsc{Approximate}()\rightarrow\R^{m}$: Output the vector
		$\of$ such that $\|\mw^{-1/2}(\of-\vf)\|_{\infty}\leq\overline{\epsilon}$ for the
		current weight $\vw$ and the current vector $\vf$. 
		\item $\textsc{Exact}()\rightarrow\R^{m}$: 
		Output the current vector $\vf$ in $\O(m \epssc^{-2})$ time.
	\end{itemize}
	Suppose $t\|\vv\|_{2}\leq\beta$ for some $\beta$ for all calls to \textsc{Move}.
	Suppose in each step, \textsc{Reweight}, \textsc{Move} and \textsc{Approximate} are called in order. Let $K$ denote the total number of coordinates changed in $\vv$ and $\vw$ between the $(k-1)$-th and $k$-th \textsc{Reweight} and \textsc{Move} calls. Then at the $k$-th \textsc{Approximate} call,
	\begin{itemize}
	\item the data structure first sets $\of_e\leftarrow \vf^{(k-1)}_e$ for all coordinates $e$ where $\vw_e$ changed in the last \textsc{Reweight}, then sets $\of_e\leftarrow \vf^{(k)}_e$ for $O(N_k\defeq 2^{2\ell_{k}}(\frac{\beta}{\overline{\epsilon}})^{2}\log^{2}m)$ coordinates $e$, where $\ell_{k}$ is the largest integer
		$\ell$ with $k=0\mod2^{\ell}$ when $k\neq 0$ and $\ell_0=0$. 
	\item The amortized time for the $k$-th \textsc{Approximate} call
	is $\widetilde{O}(\epsilon_{\mproj}^{-2}(m^{\sepConst}(K+N_{k-2^{\ell_k}})^{1-\sepConst}))$.
	\end{itemize}
\end{theorem}
The proof is the same as \cref{thm:separableSlackMaintain}.

Finally, we can prove \cref{thm:separable_mincostflow}.
\begin{proof}[Proof of \cref{thm:separable_mincostflow}]
	
The correctness is exactly the same as the proof for \cref{thm:mincostflow}.

For the runtime, we use the data structure runtimes given in \cref{thm:separableSlackMaintain} and \cref{thm:separableFlowMaintain}. We may assume $\alpha>1/2$ because otherwise the graph is $1/2$-separable and the runtime follows from \cref{thm:mincostflow}.
The amortized time for the $k$-th IPM step is
\[
	\otilde(\epssc^{-2} m^{\sepConst}(K + N_{k-2^{\ell_k}})^{1-\sepConst}).
\]
where $N_{k} \defeq 2^{2\ell_k} (\beta/\alpha)^2 \log^2 m = O(2^{2\ell_k} \log^2 m)$,
where $\alpha = O(1/\log m)$ and $\epsilon_{\mproj}= O(1/\log m)$ are defined in \textsc{CenteringImpl}.

Observe that $K + N_{k-2^{\ell_k}} = O(N_{k-2^{\ell_k}})$. Now, summing over all $T$ steps, the total time is
\begin{align}
	O({m}^{\sepConst} \log m)  \sum_{k=1}^T (N_{k-2^{\ell_k}})^{1-\sepConst} &= 
	O({m}^{\sepConst} \log^2 m)  \sum_{k=1}^T 2^{2(1-\sepConst)\ell_{(k - 2^{\ell_k})}}  \nonumber\\
	&= O({m}^{\sepConst} \log^2 m) \sum_{k'=1}^T  \nonumber 2^{2(1-\sepConst)\ell_{k'}}\sum_{k=1}^{T}[k-2^{\ell_k}=k'], \nonumber\\
	&= O(m^{\sepConst} \log^2 m \log T) \sum_{k'=1}^T  2^{2(1-\sepConst)\ell_{k'}}. \label{eq:separableSum}
\end{align}
Without $1-\alpha$ in the exponent, recall from the planar case that
\[
\sum_{k'=1}^T  2^{\ell_{k'}} = \sum_{i=0}^{\log T} 2^{i} \cdot T/2^{i+1} = O(T \log T).
\]
The summation from \cref{eq:separableSum} is
\begin{align*}
\sum_{k=1}^{T} 2^{2(1-\sepConst) \ell_k} &= \sum_{k=1}^{T}(2^{\ell_{k}})^{2-2\sepConst} \\
&\le \left(\sum_{k=1}^T 1^{1/(2\sepConst-1)}\right)^{2\sepConst-1}\left(\sum_{k=1}^T\left(\left(2^{\ell_{k}}\right)^{2-2\sepConst}\right)^{1/(2-2\sepConst)}\right)^{2-2\sepConst} \tag{by H\"older's Inquality}\\
&= \otilde\left(T^{2\sepConst-1} (T \log T)^{2-2\sepConst}\right)\\
&= \otilde(\sqrt{m}\log M \log T),
\end{align*}
where we use $T=O(\sqrt{m}\log n\log(nM))$ from \cref{thm:IPM}.
So the runtime for \textsc{CenteringImpl} is $\O(m^{1/2+\sepConst} \log M)$. 
By \cref{thm:separableDecompT}, the overall runtime is $\O(m^{1/2+\sepConst} \log M +s(m))$.
\end{proof}
\label{subsec:separable_main}

\bibliographystyle{alpha}
\bibliography{references}

\newcommand{\etalchar}[1]{$^{#1}$}
\begin{thebibliography}{vdBGJ{\etalchar{+}}21}

\bibitem[ABKS21]{AdilBKS21}
Deeksha Adil, Brian Bullins, Rasmus Kyng, and Sushant Sachdeva.
\newblock {Almost-Linear-Time Weighted $\ell_p$-norm Solvers in Slightly Dense
  Graphs via Sparsification}.
\newblock In {\em 48th International Colloquium on Automata, Languages, and
  Programming (ICALP 2021)}, volume 198, pages 9:1--9:15. Schloss Dagstuhl --
  Leibniz-Zentrum f{\"u}r Informatik, 2021.

\bibitem[AKLR18]{AKLR18}
Mudabir~Kabir Asathulla, Sanjeev Khanna, Nathaniel Lahn, and Sharath
  Raghvendra.
\newblock {A Faster Algorithm for Minimum-Cost Bipartite Perfect Matching in
  Planar Graphs}.
\newblock In {\em Proceedings of the Twenty-Ninth Annual {ACM-SIAM} Symposium
  on Discrete Algorithms, {SODA} 2018}, pages 457--476. {SIAM}, 2018.

\bibitem[AMO88]{ahuja1988network}
Ravindra~K Ahuja, Thomas~L Magnanti, and James~B Orlin.
\newblock {\em {Network Flows}}.
\newblock Prentice Hall, 1988.

\bibitem[AMV20]{axiotis2020circulation}
Kyriakos Axiotis, Aleksander M{\k{a}}dry, and Adrian Vladu.
\newblock Circulation control for faster minimum cost flow in unit-capacity
  graphs.
\newblock In {\em 2020 IEEE 61st Annual Symposium on Foundations of Computer
  Science (FOCS)}, pages 93--104. IEEE Computer Society, 2020.

\bibitem[AS20]{AdilS20}
Deeksha Adil and Sushant Sachdeva.
\newblock {Faster $p$-norm minimizing flows, via smoothed $q$-norm problems}.
\newblock In {\em Proceedings of the Fourteenth Annual ACM-SIAM Symposium on
  Discrete Algorithms}, pages 892--910. SIAM, 2020.

\bibitem[BGS21]{bernstein2021deterministic}
Aaron Bernstein, Maximilian~Probst Gutenberg, and Thatchaphol Saranurak.
\newblock Deterministic decremental {SSSP} and approximate min-cost flow in
  almost-linear time.
\newblock In {\em 62st {IEEE} Annual Symposium on Foundations of Computer
  Science, {FOCS} 2021}. {IEEE}, 2021.

\bibitem[BK09]{BK09:journal}
Glencora Borradaile and Philip~N. Klein.
\newblock An ${O}(n \log{n})$ algorithm for maximum \emph{st}-flow in a
  directed planar graph.
\newblock {\em J. {ACM}}, 56(2):9:1--9:30, 2009.

\bibitem[BKM{\etalchar{+}}17]{BKMNW17:journal}
Glencora Borradaile, Philip~N. Klein, Shay Mozes, Yahav Nussbaum, and Christian
  Wulff{-}Nilsen.
\newblock Multiple-source multiple-sink maximum flow in directed planar graphs
  in near-linear time.
\newblock {\em {SIAM} J. Comput.}, 46(4):1280--1303, 2017.

\bibitem[Bor08]{borradaile2008exploiting}
Glencora Borradaile.
\newblock {\em Exploiting Planarity for Network Flow and Connectivity
  Problems}.
\newblock Brown University, 2008.

\bibitem[CEN12]{CEN12:journal}
Erin~W. Chambers, Jeff Erickson, and Amir Nayyeri.
\newblock Homology flows, cohomology cuts.
\newblock {\em {SIAM} J. Comput.}, 41(6):1605--1634, 2012.

\bibitem[CKL{\etalchar{+}}22]{almostlinearmincostflow}
Li~Chen, Rasmus Kyng, Yang~P. Liu, Richard Peng, Maximilian~Probst Gutenberg,
  and Sushant Sachdeva.
\newblock Maximum flow and minimum-cost flow in almost-linear time.
\newblock {\em CoRR}, abs/2203.00671, 2022.

\bibitem[CKM{\etalchar{+}}11]{christiano2011electrical}
Paul Christiano, Jonathan~A Kelner, Aleksander Madry, Daniel~A Spielman, and
  Shang-Hua Teng.
\newblock Electrical flows, laplacian systems, and faster approximation of
  maximum flow in undirected graphs.
\newblock In {\em Proceedings of the forty-third annual ACM symposium on Theory
  of computing}, pages 273--282, 2011.

\bibitem[CLRS09]{CLRS}
Thomas~H Cormen, Charles~E Leiserson, Ronald~L Rivest, and Clifford Stein.
\newblock {\em Introduction to algorithms}.
\newblock MIT press, 2009.

\bibitem[CLS21]{CohenLS21}
Michael~B Cohen, Yin~Tat Lee, and Zhao Song.
\newblock Solving linear programs in the current matrix multiplication time.
\newblock {\em Journal of the ACM (JACM)}, 68(1):1--39, 2021.

\bibitem[CMSV17]{cohen2017negative}
Michael~B Cohen, Aleksander M{\k{a}}dry, Piotr Sankowski, and Adrian Vladu.
\newblock Negative-weight shortest paths and unit capacity minimum cost flow in
  $\widetilde{O}(m^{10/7} \log w)$ time.
\newblock In {\em Proceedings of the Twenty-Eighth Annual ACM-SIAM Symposium on
  Discrete Algorithms}, pages 752--771. SIAM, 2017.

\bibitem[DKP{\etalchar{+}}17]{DurfeeKPRS17}
David Durfee, Rasmus Kyng, John Peebles, Anup~B. Rao, and Sushant Sachdeva.
\newblock Sampling random spanning trees faster than matrix multiplication.
\newblock In {\em Proceedings of the 49th Annual {ACM} {SIGACT} Symposium on
  Theory of Computing, {STOC} 2017}, pages 730--742, 2017.

\bibitem[DLY21a]{treeLP}
Sally Dong, Yin~Tat Lee, and Guanghao Ye.
\newblock A nearly-linear time algorithm for linear programs with small
  treewidth: A multiscale representation of robust central path.
\newblock In {\em Proceedings of the 53rd Annual ACM SIGACT Symposium on Theory
  of Computing}, STOC 2021, pages 1784--1797. {ACM}, 2021.

\bibitem[DLY21b]{treeLPArxivV2}
Sally Dong, Yin~Tat Lee, and Guanghao Ye.
\newblock A nearly-linear time algorithm for linear programs with small
  treewidth: A multiscale representation of robust central path.
\newblock {\em arXiv preprint arXiv:2011.05365v2}, 2021.

\bibitem[DS08]{daitch2008faster}
Samuel~I Daitch and Daniel~A Spielman.
\newblock Faster approximate lossy generalized flow via interior point
  algorithms.
\newblock In {\em Proceedings of the 40th annual ACM symposium on Theory of
  computing}, pages 451--460, 2008.

\bibitem[FF56]{ford1956maximal}
Lester~R Ford and Delbert~R Fulkerson.
\newblock Maximal flow through a network.
\newblock {\em Canadian journal of Mathematics}, 8:399--404, 1956.

\bibitem[FR06]{fakcharoenphol2006planar}
Jittat Fakcharoenphol and Satish Rao.
\newblock Planar graphs, negative weight edges, shortest paths, and near linear
  time.
\newblock {\em Journal of Computer and System Sciences}, 72(5):868--889, 2006.

\bibitem[GHP18]{GHP18}
Gramoz Goranci, Monika Henzinger, and Pan Peng.
\newblock Dynamic effective resistances and approximate {Schur Complement} on
  separable graphs.
\newblock In {\em 26th Annual European Symposium on Algorithms, {ESA} 2018},
  volume 112 of {\em LIPIcs}, pages 40:1--40:15. Schloss Dagstuhl -
  Leibniz-Zentrum f{\"{u}}r Informatik, 2018.

\bibitem[GLP21]{GaoLP21:arxiv}
Yu~Gao, Yang~P. Liu, and Richard Peng.
\newblock Fully dynamic electrical flows: Sparse maxflow faster than
  {Goldberg-Rao}.
\newblock In {\em 62st {IEEE} Annual Symposium on Foundations of Computer
  Science, {FOCS}2021}. {IEEE}, 2021.

\bibitem[Gre96]{gremban1996combinatorial}
Keith~D Gremban.
\newblock {\em Combinatorial preconditioners for sparse, symmetric, diagonally
  dominant linear systems}.
\newblock PhD thesis, Carnegie Mellon University, 1996.

\bibitem[GT87]{gilbert1987}
J.~R. Gilbert and R.~E. Tarjan.
\newblock The analysis of a nested dissection algorithm.
\newblock {\em Numer. Math.}, 50(4):377--404, February 1987.

\bibitem[Has81]{Hassin}
Refael Hassin.
\newblock Maximum flow in $(s,t)$ planar networks.
\newblock {\em Information Processing Letters}, 13(3):107, 1981.

\bibitem[HJ85]{hassin1985n}
Refael Hassin and Donald~B Johnson.
\newblock An {O}$(n \log^2n)$ algorithm for maximum flow in undirected planar
  networks.
\newblock {\em SIAM Journal on Computing}, 14(3):612--624, 1985.

\bibitem[HJST21]{huang2021solving}
Baihe Huang, Shunhua Jiang, Zhao Song, and Runzhou Tao.
\newblock Solving tall dense {SDP}s in the current matrix multiplication time.
\newblock {\em arXiv preprint arXiv:2101.08208}, 2021.

\bibitem[HKRS97]{henzinger1997faster}
Monika~R Henzinger, Philip Klein, Satish Rao, and Sairam Subramanian.
\newblock Faster shortest-path algorithms for planar graphs.
\newblock {\em Journal of Computer and System Sciences}, 55(1):3--23, 1997.

\bibitem[II90]{II90}
Hiroshi Imai and Kazuo Iwano.
\newblock Efficient sequential and parallel algorithms for planar minimum cost
  flow.
\newblock In {\em Algorithms, International Symposium {SIGAL} '90, Tokyo,
  Japan}, volume 450 of {\em Lecture Notes in Computer Science}, pages 21--30.
  Springer, 1990.

\bibitem[INSW11]{INSW11}
Giuseppe~F. Italiano, Yahav Nussbaum, Piotr Sankowski, and Christian
  Wulff{-}Nilsen.
\newblock Improved algorithms for min cut and max flow in undirected planar
  graphs.
\newblock In {\em Proceedings of the 43rd {ACM} Symposium on Theory of
  Computing, {STOC} 2011}, pages 313--322. {ACM}, 2011.

\bibitem[IS79]{itai1979maximum}
Alon Itai and Yossi Shiloach.
\newblock Maximum flow in planar networks.
\newblock {\em SIAM Journal on Computing}, 8(2):135--150, 1979.

\bibitem[JS21]{JambulapatiS21}
Arun Jambulapati and Aaron Sidford.
\newblock Ultrasparse ultrasparsifiers and faster laplacian system solvers.
\newblock In D{\'{a}}niel Marx, editor, {\em Proceedings of the 2021 {ACM-SIAM}
  Symposium on Discrete Algorithms, {SODA} 2021, Virtual Conference, January 10
  - 13, 2021}, pages 540--559. {SIAM}, 2021.

\bibitem[KLOS14]{kelner2014almost}
Jonathan~A Kelner, Yin~Tat Lee, Lorenzo Orecchia, and Aaron Sidford.
\newblock An almost-linear-time algorithm for approximate max flow in
  undirected graphs, and its multicommodity generalizations.
\newblock In {\em Proceedings of the twenty-fifth annual ACM-SIAM symposium on
  discrete algorithms}, pages 217--226. SIAM, 2014.

\bibitem[KLP{\etalchar{+}}16]{KLPSS16}
Rasmus Kyng, Yin~Tat Lee, Richard Peng, Sushant Sachdeva, and Daniel~A.
  Spielman.
\newblock Sparsified cholesky and multigrid solvers for connection laplacians.
\newblock In {\em Proceedings of the 48th Annual {ACM} {SIGACT} Symposium on
  Theory of Computing, {STOC} 2016}, pages 842--850. {ACM}, 2016.

\bibitem[KLS20]{KathuriaLS20}
Tarun Kathuria, Yang~P. Liu, and Aaron Sidford.
\newblock Unit capacity maxflow in almost $o(m^{4/3})$ time.
\newblock In {\em 61st {IEEE} Annual Symposium on Foundations of Computer
  Science, {FOCS} 2020, Durham, NC, USA, November 16-19, 2020}, pages 119--130,
  2020.

\bibitem[KN13]{kaplan2013min}
Haim Kaplan and Yahav Nussbaum.
\newblock Min-cost flow duality in planar networks.
\newblock {\em arXiv preprint arXiv:1306.6728}, 2013.

\bibitem[KNK93]{khuller1993lattice}
Samir Khuller, Joseph Naor, and Philip Klein.
\newblock The lattice structure of flow in planar graphs.
\newblock {\em SIAM Journal on Discrete Mathematics}, 6(3):477--490, 1993.

\bibitem[KP15]{kang2015flow}
Donggu Kang and James Payor.
\newblock {Flow Rounding}.
\newblock {\em arXiv preprint arXiv:1507.08139}, 2015.

\bibitem[KPSW19]{kyng2019flows}
Rasmus Kyng, Richard Peng, Sushant Sachdeva, and Di~Wang.
\newblock Flows in almost linear time via adaptive preconditioning.
\newblock In {\em Proceedings of the 51st Annual ACM SIGACT Symposium on Theory
  of Computing}, pages 902--913, 2019.

\bibitem[KRT94]{king1994faster}
Valerie King, Satish Rao, and Rorbert Tarjan.
\newblock A faster deterministic maximum flow algorithm.
\newblock {\em Journal of Algorithms}, 17(3):447--474, 1994.

\bibitem[KS16]{KyngS16}
Rasmus Kyng and Sushant Sachdeva.
\newblock Approximate gaussian elimination for laplacians-fast, sparse, and
  simple.
\newblock In {\em 2016 IEEE 57th Annual Symposium on Foundations of Computer
  Science (FOCS)}, pages 573--582. IEEE, 2016.

\bibitem[KS19]{KS19}
Adam Karczmarz and Piotr Sankowski.
\newblock Min-cost flow in unit-capacity planar graphs.
\newblock In {\em 27th Annual European Symposium on Algorithms, {ESA} 2019,
  Munich/Garching, Germany}, volume 144 of {\em LIPIcs}, pages 66:1--66:17.
  Schloss Dagstuhl - Leibniz-Zentrum f{\"{u}}r Informatik, 2019.

\bibitem[Kyn17]{Kyng17}
Rasmus Kyng.
\newblock {\em Approximate Gaussian Elimination}.
\newblock PhD thesis, Yale University, 2017.

\bibitem[LR19]{LR19}
Nathaniel Lahn and Sharath Raghvendra.
\newblock A faster algorithm for minimum-cost bipartite matching in minor-free
  graphs.
\newblock In {\em Proceedings of the Thirtieth Annual {ACM-SIAM} Symposium on
  Discrete Algorithms, {SODA} 2019}, pages 569--588. {SIAM}, 2019.

\bibitem[LRT79]{lipton1979generalized}
Richard~J Lipton, Donald~J Rose, and Robert~Endre Tarjan.
\newblock Generalized nested dissection.
\newblock {\em SIAM journal on numerical analysis}, 16(2):346--358, 1979.

\bibitem[LT79]{LiptonT79}
RJ~Lipton and Robert Tarjan.
\newblock {A Planar Separator Theorem}.
\newblock {\em SIAM Journal of Applied Mathematics}, 36(2):177--189, 1979.

\bibitem[Mad13]{madry2013navigating}
Aleksander Madry.
\newblock Navigating central path with electrical flows: From flows to
  matchings, and back.
\newblock In {\em 2013 IEEE 54th Annual Symposium on Foundations of Computer
  Science}, pages 253--262. IEEE, 2013.

\bibitem[Mad16]{madry2016computing}
Aleksander Madry.
\newblock Computing maximum flow with augmenting electrical flows.
\newblock In {\em 2016 IEEE 57th Annual Symposium on Foundations of Computer
  Science (FOCS)}, pages 593--602. IEEE, 2016.

\bibitem[MN95]{MN95:journal}
Gary~L. Miller and Joseph Naor.
\newblock Flow in planar graphs with multiple sources and sinks.
\newblock {\em {SIAM} J. Comput.}, 24(5):1002--1017, 1995.

\bibitem[MP13]{MP13}
Gary~L. Miller and Richard Peng.
\newblock Approximate maximum flow on separable undirected graphs.
\newblock In {\em Proceedings of the Twenty-Fourth Annual {ACM-SIAM} Symposium
  on Discrete Algorithms, {SODA} 2013}, pages 1151--1170. {SIAM}, 2013.

\bibitem[Orl88]{orlin1988faster}
James Orlin.
\newblock A faster strongly polynomial minimum cost flow algorithm.
\newblock In {\em Proceedings of the Twentieth annual ACM symposium on Theory
  of Computing}, pages 377--387, 1988.

\bibitem[Rei83]{reif1983minimum}
John~H Reif.
\newblock Minimum $s$-$t$ cut of a planar undirected network in ${O}(n \log^2
  n)$ time.
\newblock {\em SIAM Journal on Computing}, 12(1):71--81, 1983.

\bibitem[She13]{sherman2013nearly}
Jonah Sherman.
\newblock Nearly maximum flows in nearly linear time.
\newblock In {\em 2013 IEEE 54th Annual Symposium on Foundations of Computer
  Science}, pages 263--269. IEEE, 2013.

\bibitem[She17]{sherman2017area}
Jonah Sherman.
\newblock Area-convexity, linf regularization, and undirected multicommodity
  flow.
\newblock In {\em Proceedings of the 49th Annual ACM SIGACT Symposium on Theory
  of Computing}, pages 452--460, 2017.

\bibitem[ST04]{spielman2004nearly}
Daniel~A Spielman and Shang-Hua Teng.
\newblock Nearly-linear time algorithms for graph partitioning, graph
  sparsification, and solving linear systems.
\newblock In {\em Proceedings of the thirty-sixth annual ACM symposium on
  Theory of computing}, pages 81--90, 2004.

\bibitem[ST18]{sidford2018coordinate}
Aaron Sidford and Kevin Tian.
\newblock Coordinate methods for accelerating linf regression and faster
  approximate maximum flow.
\newblock In {\em 2018 IEEE 59th Annual Symposium on Foundations of Computer
  Science (FOCS)}, pages 922--933. IEEE, 2018.

\bibitem[Tar71]{tarjan1971efficient}
Robert~E Tarjan.
\newblock An efficient planarity algorithm.
\newblock Technical report, 1971.

\bibitem[VA10]{vaidyanathan2010fast}
Balachandran Vaidyanathan and Ravindra~K Ahuja.
\newblock Fast algorithms for specially structured minimum cost flow problems
  with applications.
\newblock {\em Operations Research}, 58(6):1681--1696, 2010.

\bibitem[vdB20]{van2020deterministic}
Jan van~den Brand.
\newblock A deterministic linear program solver in current matrix
  multiplication time.
\newblock In {\em Proceedings of the Fourteenth Annual ACM-SIAM Symposium on
  Discrete Algorithms}, pages 259--278. SIAM, 2020.

\bibitem[vdB21]{van2021unifying}
Jan van~den Brand.
\newblock Unifying matrix data structures: Simplifying and speeding up
  iterative algorithms.
\newblock In {\em Symposium on Simplicity in Algorithms (SOSA)}, pages 1--13.
  SIAM, 2021.

\bibitem[vdBGJ{\etalchar{+}}21]{BGJLLPS21}
Jan van~den Brand, Yu~Gao, Arun Jambulapati, Yin~Tat Lee, Yang~P. Liu, Richard
  Peng, and Aaron Sidford.
\newblock Faster maxflow via improved dynamic spectral vertex sparsifiers.
\newblock {\em CoRR}, abs/2112.00722, 2021.

\bibitem[vdBLL{\etalchar{+}}21]{BrandLLSSSW21}
Jan van~den Brand, Yin~Tat Lee, Yang~P Liu, Thatchaphol Saranurak, Aaron
  Sidford, Zhao Song, and Di~Wang.
\newblock Minimum cost flows, {MDP}s, and $\ell1$-regression in nearly linear
  time for dense instances.
\newblock In {\em Proceedings of the 53rd Annual ACM SIGACT Symposium on Theory
  of Computing}, pages 859--869, 2021.

\bibitem[vdBLSS20]{BrandLSS20}
Jan van~den Brand, Yin~Tat Lee, Aaron Sidford, and Zhao Song.
\newblock Solving tall dense linear programs in nearly linear time.
\newblock In {\em Proceedings of the 52nd Annual ACM SIGACT Symposium on Theory
  of Computing}, pages 775--788, 2020.

\bibitem[Wei97]{weihe1997maximum}
Karsten Weihe.
\newblock Maximum $(s, t)$-flows in planar networks in ${O} (|v| \log|v|)$
  time.
\newblock {\em Journal of Computer and System Sciences}, 55(3):454--475, 1997.

\end{thebibliography}

\newpage

\appendix
\section{Appendix}
\label{sec:appendix}

\planarBoundChangeCost*
\begin{proof}
Note that $\elim{\region}$ is always a subset of $\sep{\region}$. We will instead prove 
	\begin{align*}
		\sum_{\region \in \pathT{\collN}}| \bdry{\region}| +|\sep{\region}| \leq \otilde(\sqrt{mK} ). 
	\end{align*}

First, we decompose the quantity we want to bound by levels in $\ct$:
\begin{equation} 
\sum_{\region \in \pathT{\collN}} | \bdry{\region} | + |\sep{\region}| = \sum_{i=0}^{\eta} \sum_{\region \in \pathT{\collN,i}} |\bdry{\region} | + |\sep{\region}|.
\end{equation}

We first bound $\sum_{\region \in \pathT{\collN, i}} | \bdry{\region} | + |\sep{\region}|$ for a fixed $i$.
Our main observation is that we can bound the total number of boundary vertices of nodes at level $i$ by the number of boundary and separator vertices of nodes at level $(i+1)$. 
Formally, our key claim is the following
\begin{equation} \label{eq: bdryNodeParent}
	\sum_{\region \in \pathT{\collN, i}} | \bdry{\region} | \leq   \sum_{\region' \in \pathT{\collN, i+1}} \left( | \bdry{\region'} |  + 2 | \sep{\region'} |\right).
\end{equation}
Without loss of generality, we may assume that if node $H$ is included in the left hand sum, then its sibling is included as well.
Next, recall by the definition of $\ct$, for siblings $H_1, H_2$ with parent $H'$, 
their boundaries are defined as 
\[
\bdry{H}_i= \left( \sep{H'} \cup \bdry{H'} \right) \cap V(H_i) = (S(H') \cap V(H_i)) \cup ((\partial H' \setminus S(H')) \cap V(H_i)),
\]
for $i = 1,2$. Furthermore, $V(H_1) \cup V(H_2) = V(H)$.
Another crucial observation is that a vertex from $\bdry H'$ exists in both $H_1$ and $H_2$ if and only if that vertex belongs to the separator $S(H')$.
\begin{align}
	|\bdry{\region_1}| + |\bdry{\region_2}| & \leq |\sep{\region'}| + |(\bdry{\region'} \setminus \sep{\region'}) \cap V(H_1)| +  |\sep{\region'}| + |(\bdry{H'} \setminus \sep{\region'}) \cap V(H_2')| \nonumber \\
	& \leq |\bdry{\region'}| + 2 |\sep{\region'}|. \label{eq: bdrySiblings}
\end{align}

By summing \cref{eq: bdrySiblings} over all pairs of siblings in
$\pathT{\collN, i}$, we get \cref{eq: bdryNodeParent}.
By repeatedly applying~\cref{eq: bdryNodeParent} until we reach the root
at height $\eta$, we have
\begin{equation} \label{eq: inductiveClaim}
	\sum_{\region \in \pathT{\collN, i}} | \bdry{\region}| \leq 2 \sum_{j=i+1}^{\eta} \sum_{\region' \in \pathT{\collN,j}} |S(\region')|.
\end{equation}

Summing over all the levels in $\ct$, we have
\begin{align}
  \sum_{i=0}^{\eta} \sum_{\region \in \pathT{\collN, i}} (|
  \bdry{\region}| + |S(\region)| )
  & \leq 2 \sum_{j=0}^{\eta} (j+1) \sum_{\region' \in \pathT{\collN,j}} |S(\region')|
   \tag{by \text{\cref{eq: inductiveClaim}}}\\
  & \leq 2 c \sum_{j=0}^{\eta} (j+1) \sum_{\region' \in \pathT{\collN,j}}\sqrt{|E(\region')|},
  \label{eq:boundChangeCostSum}
\end{align}
where $c$ is the constant such that $|S(\region')| \le c \left(|E(\region')|\right)^{1/2}$ in the definition of being 1/2-separable.
Furthermore, the set of ancestors of $\collN$ at level $j$ has size $| \pathT{\collN,j}| \le |\collN| =K$.
Applying the Cauchy-Schwarz inequality, we get that
\begin{align*}
    \sum_{\region \in \pathT{\collN}} \left( | \bdry{\region}| + |S(\region)| \right)
  & \leq  2 c \sum_{j=0}^{\eta} (j+1) \sqrt{|\pathT{\collN,j}|} \cdot \left(\sum_{\region' \in \pathT{\collN,j}} {|E(\region')|}\right)^{1/2} \\
  & \leq  2 c \sum_{j=0}^{\eta} (j+1) \sqrt{K} \cdot \left(\sum_{\region' \in \pathT{\collN,j}} {|E(\region')|}\right)^{1/2}\\
  & \leq  2 c  \eta \sqrt{K} \sum_{j=0}^{\eta}\left(\sum_{\region' \in \pathT{\collN,j}} {|E(\region')|}\right)^{1/2} \\
  & \leq O(\eta^2 \sqrt{mK}),
\end{align*}
where the final inequality follows from the fact that nodes at the same level form an edge partition of $G$. As $\eta = O(\log m)$, the lemma follows.
\end{proof}

\separableDecompT*
\begin{proof}
	First, we let $G$ be the root node of $\mathcal{T}(G)$. Let $G_1$ and $G_2$ be the two disjoint components of $G$ obtained after the removal of the vertices in $S(G)$.  We define the children $\child{G}{1}, \child{G}{2}$ of $G$ as follows: $V(\child{G}{i}) = V(G_i) \cup S(G)$ and $E(\child{G}{i}) = E(G_i)$, for $i =1, 2$. Edges connecting some vertex in $G_i$ and another vertex in $S(G)$ are added to $E(\child{G}{i})$. For each edge connecting two vertices in $S(G)$, we append it to $E(\child{G}{1})$ or $E(\child{G}{2})$, whichever has less edges. By construction, property \ref{property: boundary-separable} in the definition of $\mathcal{T}(G)$ holds. We continue by repeatedly splitting each child $\child{G}{i}$ that has at least one edge in the same way as we did for $G$, whenever possible. There are $O(m)$ components, each containing exactly $1$ edge. The components containing exactly $1$ edge form the \emph{leaf nodes} of $\mathcal{T}(G)$. Note that the height of $\mathcal{T}(G)$ is bounded by $O(\log m)=O(\log m)$ as for any child $H'$ of a node $H$, $|E(H')|\le b|E(H)|$.
	
	
	
	The running time of the algorithm is bounded by the total time to construct the separator for all nodes in the tree. Because the tree has height $O(\log m)$ and nodes with the same depth does not share any edge, the sum of edges over all tree nodes is $O(m\log m)$. Since $s(m)$ is convex, the algorithm runs in no more than $O(s(m)\log m)$ time.
	
\end{proof}

\separableBoundChangeCost*
\begin{proof}

	Using the separator tree, we have \cref{eq:boundChangeCostSum} in exactly the same way as for the planar case.
	\begin{align*}
		\sum_{\region \in \pathT{\collN}} (|
		\bdry{\region}| + |S(\region)| )
		& \leq 2 c \sum_{j=0}^{\eta} (j+1) \sum_{\region' \in \pathT{\collN,j}}\sqrt{|E(\region')|} \\
		\intertext{Applying H\"older's Inequality instead of Cauchy-Schwarz for the planar case, we get}
		& \leq  2 c \sum_{j=0}^{\eta} (j+1) |\pathT{\collN,j}|^{1-\sepConst} \cdot \left(\sum_{\region' \in \pathT{\collN,j}} {|E(\region')|}\right)^{\sepConst} \\
		& \leq  2 c \sum_{j=0}^{\eta} (j+1) {K}^{1-\sepConst} \cdot \left(\sum_{\region' \in \pathT{\collN,j}} {|E(\region')|}\right)^{\sepConst}\\
		& \leq  2 c  \eta {K}^{1-\sepConst} \sum_{j=0}^{\eta}\left(\sum_{\region' \in \pathT{\collN,j}} {|E(\region')|}\right)^{\sepConst} \\
		& \leq O(\eta^2 \septime),
	\end{align*}
	where the final inequality follows from the fact that nodes at the same level form an edge partition of $G$. As $\eta = O(\log m)$, the lemma follows.
\end{proof}

\end{document}